\documentclass[11pt,a4paper]{article}

\usepackage{hyperref}

\usepackage[fleqn]{amsmath}
\usepackage{amsfonts,amsthm,amssymb}

\usepackage[dvips]{graphicx}
\usepackage{mathrsfs}
\usepackage{makeidx}
\usepackage{xypic}
\usepackage[nottoc]{tocbibind}
\usepackage{pstricks}
\usepackage{bbm}
\usepackage[arrow, matrix, curve]{xy}
\usepackage{epic}
\usepackage{eepic}
\usepackage{times}
\usepackage{epsfig}
\usepackage{tikz}

\xymatrixcolsep{0.5cm}
\xymatrixrowsep{0.5cm}
\newtheorem{theorem}{Theorem}[section]
\newtheorem{proposition}{Proposition}[section]
\newtheorem{lemma}{Lemma}[section]

\newtheorem{corollary}{Corollary}[section]

{\theoremstyle{remark}
\newtheorem{rem}{Remark}[section]}
\newcommand{\mathsc}[1]{{\normalfont\textsc{#1}}}

\newcommand{\beqa}{\begin{eqnarray}}
\newcommand{\eeqa}{\end{eqnarray}}

\makeindex
\numberwithin{equation}{section}

\begin{document}

\begin{flushright}
LPENSL-TH-06/15
\end{flushright}

\vspace{24pt}

\begin{center}
\textbf{\Large Antiperiodic dynamical 6-vertex model by separation of variables II: Functional equations and form factors}

\vspace{50pt}

\begin{large}
{\bf D.~Levy-Bencheton}\footnote{Univ. de Bourgogne; IMB, UMR 5584 du CNRS, Dijon, France; damien.levybencheton@ens-lyon.fr},
{\bf G. Niccoli}\footnote{ENS Lyon; CNRS; Laboratoire de Physique, UMR 5672, Lyon, France; giuliano.niccoli@ens-lyon.fr}
{\bf and V.~Terras}\footnote{Univ. Paris Sud; CNRS; LPTMS, UMR 8626, Orsay 91405, France; veronique.terras@lptms.u-psud.fr}
\end{large}

\vspace{.5cm}
 
\vspace{2cm}

\today

\end{center}

\vspace{1cm}

\begin{abstract}
We pursue our study of the antiperiodic dynamical 6-vertex model using Sklyanin's separation of variables approach, allowing in the model new possible global shifts of the dynamical parameter. We show in particular that the spectrum and eigenstates of the antiperiodic transfer matrix are completely characterized by a system of discrete equations. We prove the existence of different reformulations of this characterization in terms of functional equations of Baxter's type.
We notably consider the homogeneous functional $T$-$Q$ equation which is the continuous analog of the aforementioned discrete system and show, in the case of a model with an even number of sites, that the complete spectrum and eigenstates of the antiperiodic transfer matrix  can equivalently be described in terms of a particular class of its $Q$-solutions, hence leading to a complete system of Bethe equations.
Finally, we compute the form factors of local operators for which we obtain determinant representations in finite volume.
\end{abstract}

\vspace{1cm}


\newpage

\tableofcontents
\newpage

\section{Introduction}

This paper is a continuation of a previous work of one of the authors \cite{Nic13a}, in which the study of the antiperiodic dynamical 6-vertex model was initiated by means of the Separation of Variable (SOV) method introduced by Sklyanin \cite{Skl85,Skl90,Skl92,Skl95}.

The dynamical 6-vertex model is constructed from an $R$-matrix satisfying the dynamical (or modified) Yang-Baxter equation \cite{GerN84,Fel95}. In the context of exactly solvable models of statistical mechanics, it is the archetype of {\em interaction-round-faces} (IRF) models \cite{Bax82L}, which describe the interactions of a local variable around faces of a two-dimensional square lattice. This is notably the case of the exactly solvable solid-on-solid (SOS) model \cite{Bax73a,AndBF84,DatJMO86,KunY88,PeaS88}, which models the growth of a surface (for instance in the context of a crystal-vapor interface) with respect to a flat reference surface. In this context, a height variable is attached to each site of the lattice, and the local Boltzmann weights describing the probability of a height configuration around each face of the lattice correspond to the non-zero entries of the dynamical $R$-matrix.
This model plays a crucial role in Baxter's solution \cite{Bax73a} of the famous eight-vertex model.

The fact that the $R$-matrix of the model depends on an extra parameter, the so-called {\em dynamical} parameter (related to the height variable in the language of SOS model), results into a modification of the algebraic structure underlying integrability compared to what happens in usual vertex models such as the six-vertex model \cite{Fel95,FelV96a,FelV96b}. 
In particular, as mentioned above, the $R$-matrix of the model no longer satisfies the usual quantum Yang-Baxter equation, but instead a modified version of this equation, in which the dynamical parameter gets shifted by an element of the Cartan. 
As a consequence, the corresponding Yang-Baxter algebra incorporates some additional operator structure acting on the dynamical variable \cite{FelV96a,FelV96b}.
In practice, the appearance of the dynamical shifts may be a problem for actual computations of physical quantities of the model. For instance, the partition function of the model with domain wall boundary conditions does not seem to be expressible in the form of a single determinant \cite{Ros09,PakRS08} as in the six-vertex case \cite{Ize87}. As a consequence, the study of the form factors and correlation functions of the periodic model in the algebraic Bethe Ansatz (ABA) framework happens to be slightly more complicated than in the six-vertex case \cite {LevT13a,LevT13b,LevT14a}. In fact, the latter relies on the use of a compact formula, preferably in the form of a single determinant, for the scalar products of Bethe states. The problem is that there does not exist, in the ABA framework, a model-independent clear procedure to construct (or even guess) such a representation, either for the Bethe states scalar products or for the aforementioned partition function, two quantities which are intimately related.

In the antiperiodic SOS model, which can be solved by SOV \cite{FelS99,Nic13a}, the situation is somehow simpler. First, the space of states of the (finite size) model is finite dimensional for generic crossing parameter, contrary to what happens in the periodic case: this means in particular that, when performing a change of basis between the canonical basis of the space of states and the eigenstate basis, we only have to deal with finite sums. Second, due to the eigenstate representation in terms of separate variables, determinant formulas for the corresponding scalar products appear in a much more natural way, somehow intrinsic to the method \cite{Nic13a}. We shall see in the present paper that such formulas can quite naturally be extended to the form factors.

The purpose of this paper is three-fold. First, we revisit the study of \cite{Nic13a} so as to slightly generalize it to cases that correspond to different versions of the SOS solvable model (see \cite{KunY88, PeaS88}), related to different possible values of the global shift of the dynamical parameter of the model, and that will be useful for our further study of the eight-vertex model \cite{NicT15b}.
Second, we pursue the study of the spectrum, and discuss the reformulation of the discrete SOV characterization in terms  of solutions of some homogeneous (respectively inhomogeneous) functional equation of Baxter's type (and hence of Bethe-type equations), a question that we expect to be of primary importance for the consideration of the homogeneous and thermodynamic limits: we notably show that, in the case of a model with an even number of sites, the SOV characterization of the spectrum is completely equivalent with a description in terms of particular classes of solutions of the usual (homogeneous) $T$-$Q$ equation.
Finally, we explicitly write determinant representations for the form factors of local spin and height operators.

The article is organized as follows.
In Section~\ref{sec-model}, we introduce the model and recall the algebraic framework for its resolution.
In Section~\ref{sec-SOVbasis}, we construct a SOV basis of the representation space of the model.
In Section~\ref{sec-diag-t}, we diagonalize the commuting antiperiodic transfer matrices on the subspace of the representation space which corresponds to the actual space of states of the antiperiodic SOS model. We notably discuss the characterization of the spectrum and eigenstates in terms of solutions of functional equations.
Then, in Section~\ref{sec-inv}, we solve the quantum inverse problem for local operators of the model in terms of the antiperiodic monodromy matrix elements, and in particular in terms of the antiperiodic transfer matrix.
This enables us, in Section~\ref{sec-ff}, to write explicit determinant representations for the form factors of local spin and local height in the eigenstate basis of the transfer matrix.

\section{The dynamical 6-vertex model}
\label{sec-model}

The dynamical 6-vertex model is associated with a dynamical $R$-matrix of the form
\begin{equation}\label{mat-R}
  R(\lambda|t)=
  \begin{pmatrix}
  a(\lambda) & 0 & 0 & 0 \\
  0 & e^{i \mathsf{y}\eta}\,b(\lambda|t) & e^{i\mathsf{y}\lambda}\, c(\lambda|t) & 0 \\
  0 & e^{-i\mathsf{y}\lambda}\, c(\lambda|-t) & e^{-i\mathsf{y}\eta}\, b(\lambda|-t) & 0\\
  0 & 0 & 0 & a(\lambda)
  \end{pmatrix} \in \mathrm{End}(V_1\otimes V_2),
\end{equation}
with $V_i\simeq \mathbb{C}^2$. Throughout this paper, $\mathsf{y}\in\{0,1\}$ is fixed. The R-matrix \eqref{mat-R} depends on two parameters, a spectral parameter $\lambda\in\mathbb{C}$ and a dynamical parameter $t\in t_0+\eta\mathbb{Z}$, where $t_0$ will be specified later.
The functions $a,b,c$ are given as
\begin{equation}\label{def-abc}
   a(\lambda)=\theta(\lambda+\eta), \quad 
   b(\lambda|t)=\frac{\theta(\lambda)\,\theta(t+\eta)}{\theta(t)}, \quad
   c(\lambda|t)=\frac{\theta(\eta)\,\theta(t+\lambda)}{\theta(t)},
\end{equation}
where $\theta(\lambda)\equiv\theta_1(\lambda|\omega)$ denotes the usual theta-function (see Appendix~\ref{app-theta}) with quasi-periods $\pi$ and $\pi\omega$ ($\Im\omega>0$). $\eta\in\mathbb{C}$ is the crossing parameter of the model, which is supposed to be generic throughout this paper.
The $R$-matrix \eqref{mat-R} with dynamical parameter $t$ is solution of the quantum dynamical Yang-Baxter equation \cite{GerN84,Fel95} on $V_1\otimes V_2\otimes V_3$,
\begin{equation}\label{DYBE}
  R_{1,2}(\lambda_{12}|t+\eta\sigma_3^z)\, R_{1,3}(\lambda_{13}|t)\,
  R_{2,3}(\lambda_{23}|t+\eta\sigma_1^z)
  =R_{2,3}(\lambda_{23}|t)\, R_{1,3}(\lambda_{13}|t+\eta\sigma_2^z)\,
  R_{1,2}(\lambda_{12}|t).
\end{equation}
Here and in the following, the indices indicate as usual the space of the tensor product on which the corresponding operator acts non-trivially, and $\sigma_i^\alpha$ ($\alpha=x,y,z$) denotes the usual Pauli matrix acting on $V_i\simeq\mathbb{C}^2$. Also we have used the shorthand notation $\lambda_{ij}\equiv \lambda_i-\lambda_j$.

As mentioned in Introduction, the matrix elements of the $R$-matrix \eqref{mat-R} can be understood as the local Boltzmann weights of an exactly solvable solid-on-solid (SOS) model (see Figure~\ref{6faces}). In this framework, the dynamical parameter $t$ corresponds to the height variable, and the dynamical Yang-Baxter equation \eqref{DYBE} is simply the star-triangle relation of IRF type \cite{Bax82L}. The shift ($\in\{\pm\eta\}$) of the height between two neighboring sites can be understood in terms of a spin variable ($\in\{\pm 1\}$) on the corresponding link.
\begin{figure}[h]
\centering
\begin{tikzpicture}
    \draw(0,5.5) ;
    \draw(0,-2) ;   

    \draw (0,3.5) -- (0,4.5) node[above left]{$t$} ; 
    \draw (0,4.5) -- (1,4.5) node[above right]{$t+\eta$} ;
    \draw (1,4.5) -- (1,3.5) node[below right]{$t+2\eta$} ;
    \draw (1,3.5) -- (0,3.5) node[below left]{$t+\eta$} ;
    \draw (0.5,4.5) node[above]{$+$} ;
    \draw (0.5,3.5) node[below]{$+$} ;
    \draw (0,4) node[left]{$+$} ;
    \draw (1,4) node[right]{$+$} ;
    
    \draw (0.5,2) node[above]{$a(\lambda)$} ;

    \draw (0,0) -- (0,1) node[above left]{$t$} ; 
    \draw (0,1) -- (1,1) node[above right]{$t-\eta$} ;
    \draw (1,1) -- (1,0) node[below right]{$t-2\eta$} ;
    \draw (1,0) -- (0,0) node[below left]{$t-\eta$} ;
    \draw (0.5,1) node[above]{$-$} ;
    \draw (0.5,0) node[below]{$-$} ;
    \draw (0,0.5) node[left]{$-$} ;
    \draw (1,0.5) node[right]{$-$} ;  
    
    \draw (0.5,-1.5) node[above]{$a(\lambda)$} ;

    \draw (4,3.5) -- (4,4.5) node[above left]{$t$} ; 
    \draw (4,4.5) -- (5,4.5) node[above right]{$t+\eta$} ;
    \draw (5,4.5) -- (5,3.5) node[below right]{$t$} ;
    \draw (5,3.5) -- (4,3.5) node[below left]{$t-\eta$} ;
    \draw (4.5,4.5) node[above]{$+$} ;
    \draw (4.5,3.5) node[below]{$+$} ;
    \draw (4,4) node[left]{$-$} ;
    \draw (5,4) node[right]{$-$} ;
    
    \draw (4.5,2) node[above]{$e^{i\mathsf{y}\eta}\,{b}(\lambda|t)$} ;

    \draw (4,0) -- (4,1) node[above left]{$t$} ; 
    \draw (4,1) -- (5,1) node[above right]{$t-\eta$} ;
    \draw (5,1) -- (5,0) node[below right]{$t$} ;
    \draw (5,0) -- (4,0) node[below left]{$t+\eta$} ;
    \draw (4.5,1) node[above]{$-$} ;
    \draw (4.5,0) node[below]{$-$} ;
    \draw (4,0.5) node[left]{$+$} ;
    \draw (5,0.5) node[right]{$+$} ;  
    
    \draw (4.5,-1.5) node[above]{$e^{-i\mathsf{y}\eta}\,{b}(\lambda|-t)$} ;

    \draw (8,3.5) -- (8,4.5) node[above left]{$t$} ; 
    \draw (8,4.5) -- (9,4.5) node[above right]{$t-\eta$} ;
    \draw (9,4.5) -- (9,3.5) node[below right]{$t$} ;
    \draw (9,3.5) -- (8,3.5) node[below left]{$t-\eta$} ;
    \draw (8.5,4.5) node[above]{$-$} ;
    \draw (8.5,3.5) node[below]{$+$} ;
    \draw (8,4) node[left]{$-$} ;
    \draw (9,4) node[right]{$+$} ;
    
    \draw (8.5,2) node[above]{$e^{i\mathsf{y}\lambda}\,{c}(\lambda|t)$} ;

    \draw (8,0) -- (8,1) node[above left]{$t$} ; 
    \draw (8,1) -- (9,1) node[above right]{$t+\eta$} ;
    \draw (9,1) -- (9,0) node[below right]{$t$} ;
    \draw (9,0) -- (8,0) node[below left]{$t+\eta$} ;
    \draw (8.5,1) node[above]{$+$} ;
    \draw (8.5,0) node[below]{$-$} ;
    \draw (8,0.5) node[left]{$+$} ;
    \draw (9,0.5) node[right]{$-$} ;  
    
    \draw (8.5,-1.5) node[above]{$e^{-i\mathsf{y}\lambda}\,{c}(\lambda|-t)$} ;     
\end{tikzpicture}\vspace{-5mm}
\caption{\label{6faces}The 6 different local configurations around a face and their associated local statistical weights.}
\end{figure}

\begin{rem}\label{rem-y=1}
The case $\mathsf{y}=1$, which was not considered in \cite{Nic13a}, corresponds to a diagonal dynamical gauge transformation of the $R$-matrix of \cite{Nic13a} of the form:
\begin{equation}
  R_{12}(\lambda_{12}|t)=G_2(\lambda_2|t)\, G_1(\lambda_1|t+\eta\sigma_2^z)\,
  \big[ R_{12}(\lambda_{12}|t) \big]_{\mathsf{y}=0}\,
  G_2(\lambda_2|t+\eta\sigma_1^z)^{-1}\,
  G_1(\lambda_1|t)^{-1},
\end{equation}
with
\begin{equation}
   G(\lambda|t)=e^{-i\mathsf{y}\frac{t}{2}}
   \begin{pmatrix} e^{i\mathsf{y}\frac{\lambda}{2}} & 0\\
   0 & e^{-i\mathsf{y}\frac{\lambda}{2}} \end{pmatrix}.
\end{equation}
%
This case is interesting since it enables us to consider a model in which the dynamical parameter is shifted by half of the imaginary quasi-period (i.e. $t_0\in\mathbb{R}+\mathsf{y}\frac{\pi}{2}\omega$, see Section~\ref{sec-LRrep}) so as to recover for instance the Boltzmann weights considered in \cite{KunY88,PeaS89}.
It is also useful  for the consideration of particular quasi-periodic boundary conditions for the 8-vertex model obtained from this SOS model by vertex-IRF transformations \cite{NicT15b}.
\end{rem}

\begin{rem}\label{rem-trig-lim}
It may be interesting to consider the trigonometric limit of \eqref{mat-R}-\eqref{def-abc}, which corresponds to the limit $\omega\to +i\infty$. If $t_0$ is of the form $t_0=\tilde{t}_0+\mathsf{y}\frac{\pi}{2}\omega$, with $\tilde{t}_0$ independent of $\omega$ (as considered in this paper, see \eqref{Dynamical-spectrum-2}), then one obtains different limits according to whether $\mathsf{y}=0$ or $\mathsf{y}=1$. Up to normalization, the trigonometric limit of the case $\mathsf{y}=0$ corresponds to the trigonometric {\em dynamical} 6-vertex model, with 
\begin{equation}\label{trig-abc}
  a(\lambda )=\sin (\lambda +\eta ),\quad 
  b(\lambda | t )=\frac{\sin \lambda \, \sin ( t +\eta )}{\sin t },\quad 
  c(\lambda | t )=\frac{\sin \eta \, \sin ( t +\lambda )}{\sin t },
\end{equation}
whereas, in the case $\mathsf{y}=1$, one simply recovers the $R$-matrix of the {\em usual} 6-vertex (or XXZ) model\footnote{We use the fact that $\frac{1}{2} e^{-i\pi\frac{\omega}{4}}\theta_1(u|\omega)\underset{\omega\to +i\infty}{\longrightarrow} \sin u$, and that $\frac{\theta_1(u\pm\epsilon\pi\omega|\omega)}{\theta_1(v\pm\epsilon \pi\omega|\omega)}\underset{\omega\to +i\infty}{\longrightarrow} e^{\pm i(v-u)}$ for $u,v\in\mathbb{C}$ and $0<\epsilon<1$.}.
\end{rem}

\subsection{The dynamical Yang-Baxter algebra}
\label{sec-DYBA}

The elliptic quantum group associated with the dynamical $R$-matrix \eqref{mat-R} was introduced in \cite{Fel95}, and its representations were studied in \cite{FelV96a}, leading in \cite{FelV96b} to an algebraic Bethe Ansatz for the corresponding periodic SOS model. 
In this framework, the key object is the monodromy matrix, which provides a representation of the corresponding Yang-Baxter algebra.
In the case of a tensor product of fundamental representations that we consider in this paper, it is defined as the following ordered product of $R$-matrices,
\begin{align}
\mathsf{M}_{0}(\lambda |t )
&\equiv 
R_{0,\mathsf{N}}(\lambda -\xi _{\mathsf{N}}|t +\eta \sum_{a=1}^{\mathsf{N}-1}\sigma _{a}^{z})\cdots R_{0,1}(\lambda -\xi _{1}|t)\quad
\in\mathrm{End}(V_0\otimes V_1\otimes \ldots\otimes V_{\mathsf{N}}) \nonumber\\
&\equiv 
\begin{pmatrix}
\mathsf{A}(\lambda |t ) & \mathsf{B}(\lambda |t ) \\ 
\mathsf{C}(\lambda |t ) & \mathsf{D}(\lambda |t )
\end{pmatrix}_{\! [0]}
. \label{monodromy}
\end{align}
Here $\mathsf{N}$ is the size of the model, $V_n\simeq\mathbb{C}^2$ for  $n\in \{0,1,\ldots,{\mathsf{N}}\}$, and $\xi _{n}$, $n\in \{1,\ldots,{\mathsf{N}}\}$, are inhomogeneity parameters.
In this context, $V_0$ is usually called the auxiliary space and $\mathbb{V}_{\mathsf{N}}\equiv \otimes_{n=1}^\mathsf{N} V_n$ the quantum space.
Commutation relations for the entries $\mathsf{A}(\lambda |t )$, $\mathsf{B}(\lambda |t )$, $\mathsf{C}(\lambda |t )$, $\mathsf{D}(\lambda |t )$ $\in\mathrm{End}(\mathbb{V}_{\mathsf{N}})$ of the monodromy matrix \eqref{monodromy} are given by the following quadratic relation, which is a consequence of \eqref{DYBE},
\begin{equation}\label{RTT}
R_{0,0'}(\lambda _{0 0'}|t+\eta \mathsf{S})\,
\mathsf{M}_{0}(\lambda _{0}|t )\,
\mathsf{M}_{0'}(\lambda _{0'}|t +\eta \sigma _{0}^{z})
=\mathsf{M}_{0'}(\lambda _{0'}| t )\,
\mathsf{M}_{0}(\lambda _{0}|t+\eta \sigma _{0'}^{z})\,
R_{0,0'}(\lambda _{0 0'}|t ),
\end{equation}
where $\mathsf{S}$ is\ the total $z$-component of the spin:
\begin{equation}
  \mathsf{S}=\sum_{n=1}^{\mathsf{N}}\sigma _{n}^{z}.  \label{Def-S}
\end{equation}
To handle the shifts of the dynamical parameter induced by the relation \eqref{RTT}, it is convenient to introduce, as in  \cite{Nic13a}, some dynamical operators $\tau $ and $\mathsf{T}_{\tau }^{\pm }$ which commute with local spin operators $\sigma_n^\alpha$ and such that
\begin{equation}
\mathsf{T}_{\tau }^{\pm }\tau =(\tau \pm \eta )\mathsf{T}_{\tau }^{\pm }.
\label{Dyn-op-comm}
\end{equation}
This enables us to define a new monodromy matrix incorporating these dynamical operators,
\begin{equation}\label{mon-op}
\mathcal{M}_{0}(\lambda )
\equiv \mathsf{M}_{0}(\lambda |\tau )\, \mathsf{T}_{\tau }^{\sigma_{0}^{z}}
\equiv \begin{pmatrix}
\mathcal{A}(\lambda ) & \mathcal{B}(\lambda  ) \\ 
\mathcal{C}(\lambda ) & \mathcal{D}(\lambda )
\end{pmatrix}_{\! [0]},
\end{equation}
where $\mathsf{M}_{0}(\lambda |\tau )$ corresponds to the monodromy matrix \eqref{monodromy} in which we have substituted the dynamical parameter $t$ by the operator $\tau$, and where
\begin{equation}
\mathsf{T}_{\tau }^{\pm \sigma _{0}^{z}}
\equiv
\begin{pmatrix}
\mathsf{T}_{\tau }^{\pm } & 0 \\ 
0 & \mathsf{T}_{\tau }^{\mp }
\end{pmatrix}_{\! [0]}.
\end{equation}
%
The operator entries $\mathcal{A},\mathcal{B},\mathcal{C},\mathcal{D}$ of the monodromy matrix \eqref{mon-op} act on the space
\begin{equation}\label{space}
\mathbb{D}_{\mathsf{(6VD)},\mathsf{N}}
\equiv \mathbb{V}_{\mathsf{N}}\otimes \mathbb{D},
\end{equation}
where $\mathbb{D}$\ is a representation space of the dynamical operators.
Their commutation relations are now given by the quadratic relation
\begin{equation}\label{RTT-op}
R_{0,0'}(\lambda _{00'}|\tau +\eta \mathsf{S})\,
\mathcal{M}_{0}(\lambda _{0} )\,\mathcal{M}_{0'}(\lambda _{0'} )
 =\mathcal{M}_{0'}(\lambda _{0'} )\,\mathcal{M}_{0}(\lambda _{0} )\,R_{0,0'}(\lambda _{00'}|\tau ),
\end{equation}
which follows from \eqref{RTT} and from the zero-weight property of the $R$-matrix \eqref{mat-R}:
\begin{equation}
\mathsf{T}_{\tau }^{-\sigma_0^z}\,\mathsf{T}_{\tau }^{-\sigma_{0'}^z }\,
R_{0,0'}(\lambda _{00'}|\tau )\, \mathsf{T}_{\tau}^{\sigma _{0}^{z}}\,\mathsf{T}_{\tau }^{\sigma _{0'}^{z}}
=R_{0,0'}(\lambda _{00'}|\tau ).  \label{Comm-R12}
\end{equation}
We also recall the inversion formula for the monodromy matrix:
\begin{theorem}\label{th-inv}
The inverse of the monodromy matrix \eqref{mon-op} is given by the relation
\begin{equation}\label{inv-mon}
  \mathcal{M}_0(\lambda)\cdot \sigma_0^y\,\mathcal{M}_0(\lambda-\eta)^{t_0}\,\sigma_0^y
  =e^{-i\mathsf{y}\eta\mathsf{S}}\frac{\theta(\tau)}{\theta(\tau+\eta\mathsf{S})}\,\mathrm{det}_q M(\lambda),
\end{equation}
in terms of the quantum determinant $\mathrm{det}_q M(\lambda)$. The latter is a central element of the dynamical Yang-Baxter algebra defined as
\begin{align}
\mathrm{det}_{q}\mathsf{M}(\lambda ) =\text{\textsc{a}}(\lambda )\, \text{\textsc{d}}(\lambda -\eta )
&= e^{i\mathsf{y}\eta\mathsf{S}}\,\frac{\theta (\tau +\eta \mathsf{S})}{\theta (\tau )}
     \big( \mathcal{A}(\lambda  )\,\mathcal{D}(\lambda -\eta  )
     -\mathcal{B}(\lambda )\, \mathcal{C}(\lambda -\eta  )\big)
       \notag\\
&=e^{i\mathsf{y}\eta\mathsf{S}}\,\frac{\theta (\tau +\eta \mathsf{S})}{\theta (\tau )}
    \big( \mathcal{D}(\lambda )\,\mathcal{A}(\lambda -\eta )
    -\mathcal{C}(\lambda  )\, \mathcal{B}(\lambda -\eta )\big) ,
    \label{q-det}
\end{align}
with
\begin{equation}\label{a-d}
   \mathsc{a}(\lambda)\equiv\prod_{n=1}^\mathsf{N} a(\lambda-\xi_n),
   \qquad
   \mathsc{d}(\lambda)\equiv \mathsc{a}(\lambda-\eta).
\end{equation}
\end{theorem}

The algebraic Bethe Ansatz \cite{FelV96b} consists in diagonalizing the transfer matrices of the periodic model, i.e. the traces of the monodromy matrix \eqref{mon-op},
\begin{equation}\label{transfer}
\mathcal{T}(\lambda)=\mathcal{A}(\lambda)+\mathcal{D}(\lambda),
\end{equation}
on the subspace of the representation space \eqref{space} associated with the zero eigenvalue of the spin operator $\mathsf{S}$ \eqref{Def-S}. In this paper, we apply instead the SOV approach to the antiperiodic model (or more generally to a $\kappa$-twisted antiperiodic model for some parameter $\kappa\in\mathbb{C}\setminus\{0\}$), and therefore define the following $\kappa$-twisted {\em antiperiodic} monodromy matrices:
\begin{equation}
\mathsf{\overline{M}}_{0}^{(\kappa)}(\lambda |\tau )
\equiv X_0^{(\kappa)}\sigma _{0}^{x}\, \mathsf{M}_{0}(\lambda |\tau ),\qquad
\mathcal{\overline{M}}_{0}^{(\kappa)}(\lambda  )
\equiv \mathsf{\overline{M}}_{0}^{(\kappa)}(\lambda |\tau )\,\mathsf{T}_{\tau }^{\sigma _{0}^{z}},
\label{anti-p-6vD-M}
\end{equation}
with $X^{(\kappa)}\equiv \mathrm{diag}(\kappa,\kappa^{-1})$.
It is easy to see that, thanks to the following properties of the R-matrix \eqref{mat-R},
\begin{align}
   &\sigma _{1}^{x}\, \sigma _{2}^{x}\, R_{1,2}(\lambda|t)\,\sigma _{1}^{x}\, \sigma _{2}^{x}
=R_{1,2}(\lambda|-t+\mathsf{y}\pi\omega), \\
   &X^{(\kappa)}_1\, X^{(\kappa)}_2\, R_{1,2}(\lambda|t)\, X^{(\kappa^{-1})}_1\, X^{(\kappa^{-1})}_2
   =R_{1,2}(\lambda|t),
\end{align}
the monodromy matrices \eqref{anti-p-6vD-M}
satisfy the quadratic relations
\begin{multline}
R_{0,0'}(\lambda _{00'}|-\tau -\eta \mathsf{S}+\mathsf{y}\pi\omega)\,
\mathsf{\overline{M}}_{0}^{(\kappa)}(\lambda _{0}|\tau )\, \mathsf{T}_{\tau}^{\sigma _{0}^{z}}\,
\mathsf{\overline{M}}_{0'}^{(\kappa)}(\lambda _{0'}|\tau )\, \mathsf{T}_{\tau }^{-\sigma _{0}^{z}}\\
 = \mathsf{\overline{M}}_{0'}^{(\kappa)}(\lambda _{0'}|\tau )\, \mathsf{T}_{\tau }^{\sigma _{0'}^{z}}\,
    \mathsf{\overline{M}}_{0}^{(\kappa)}(\lambda _{0}|\tau )\, \mathsf{T}_{\tau}^{-\sigma _{0'}^{z}}\,
    R_{0,0'}(\lambda _{00'}|\tau ),
\label{D-YB-op-0}
\end{multline}
\begin{equation}
R_{0,0'}(\lambda _{00'}|-\tau -\eta \mathsf{S} +\mathsf{y}\pi\omega)\,
\overline{\mathcal{M}}_{0}^{(\kappa)}(\lambda _{0} )\,
\overline{\mathcal{M}}_{0'}^{(\kappa)}(\lambda _{0'} )
 = \overline{\mathcal{M}}_{0'}^{(\kappa)}(\lambda_{0'} )\,
 \overline{\mathcal{M}}_{0}^{(\kappa)}(\lambda_{0} )\,R_{0,0'}(\lambda _{00'}|\tau ).
\label{YBECalDyn}
\end{equation}
In the following, we will show how to diagonalize, in the SOV framework, the $\kappa$-twisted antiperiodic transfer matrices,
\begin{equation}\label{anti-transfer}
\overline{\mathcal{T}}^{(\kappa)}(\lambda)
=\kappa^{-1}\mathcal{B}(\lambda)+\kappa\,\mathcal{C}(\lambda),
\end{equation}
on some subspace of the representation space \eqref{space} associated with a particular eigenvalue of the operator
\begin{equation}\label{S-tau}
\mathsf{S}_{\tau }\equiv \eta \mathsf{S}+2\tau.
\end{equation}

\subsection{Left and right representation spaces}
\label{sec-LRrep}

In the following, we may particularize by a subscript $\mathcal{L}$ (respectively $\mathcal{R}$) the left (respectively the right) representation spaces for the spin and dynamical operators. Hence, for instance, the operator entries of the monodromy matrix \eqref{mon-op} act\footnote{These notations may be confusing since $\mathbb{D}_{\mathsf{(6VD)},\mathsf{N}}^{\mathcal{R}}$ corresponds in fact to a left-module (and $\mathbb{D}_{\mathsf{(6VD)},\mathsf{N}}^{\mathcal{L}}$ to a right-module) of the elliptic quantum group, but they agree with some more physical convention also used in \cite{Nic13a}.} to the right on the ``ket'' space $\mathbb{D}_{\mathsf{(6VD)},\mathsf{N}}^{\mathcal{R}}\equiv \mathbb{D}_{\mathsf{(6VD)},\mathsf{N}}$ and to the left on its restricted dual space, the ``bra'' space $\mathbb{D}_{\mathsf{(6VD)},\mathsf{N}}^{\mathcal{L}}$.

$\mathbb{D}^{\mathcal{L}/\mathcal{R}}$ is an infinite dimensional representation space of the dynamical operators algebra \eqref{Dyn-op-comm} with  $\tau $-eigenbasis left (covectors) and right (vectors) respectively defined as
\begin{equation}
\langle t(a)|\equiv \langle t(0)|\mathsf{T}_{\tau }^{-a},\qquad
|t(a)\rangle \equiv \mathsf{T}_{\tau }^{a}|t(0)\rangle ,\qquad
\forall a\in \mathbb{Z},  \label{t-dyn-sp}
\end{equation}
such that
\begin{equation} \label{Dynamical-spectrum-1}
\langle t(a)|\tau =t(a)\langle t(a)|,\qquad
\tau |t(a)\rangle =t(a)|t(a)\rangle ,
 \qquad 
t(a) \equiv -\eta a+t_0,\quad
\forall a\in \mathbb{Z},
\end{equation}
with the normalization $\langle t(a)|t(b)\rangle=\delta _{a,b},\ \forall a,b\in \mathbb{Z}$.
In this paper, we fix the value of $t_0$ to be
\begin{equation}\label{Dynamical-spectrum-2}
 t_0=  -\frac{\eta }{2} \mathsf{N} +\mathsf{x}\frac{\pi }{2}+\mathsf{y}\frac{\pi}{2}\omega,
 \qquad
 \text{with}\qquad
 \mathsf{x}\in \{0,1\},
\end{equation}
such that $(\mathsf{x},\mathsf{y})\not=(0,0)$ when $\mathsf{N}$ is even.
We denote the corresponding representation space as $\mathbb{D}^{\mathcal{L}/\mathcal{R}}\equiv \mathbb{D}_{(\mathsf{x},\mathsf{y}),\mathsf{N}}^{\mathcal{L}/\mathcal{R}}$.
We respectively denote the left and right spin basis in $V_n^{\mathcal{L}/\mathcal{R}}$ as
\begin{equation}
\langle n,h_{n}|\sigma _{n}^{z}=(1-2h_{n})\langle n,h_{n}|,\quad
\sigma _{n}^{z}|n,h_{n}\rangle =(1-2h_{n})|n,h_{n}\rangle ,\quad
h_{n}\in \{0,1\},
\end{equation}
with $\langle n,h_{n}|n,h_{n}^{\prime }\rangle =\delta _{h_{n},h_{n}^{\prime}}$ for any $n\in \{1,\ldots,\mathsf{N}\}$.
Hence the states
\begin{equation}\label{dyn-spin-basis}
(\otimes _{n=1}^{\mathsf{N}}\langle n,h_{n}|)\otimes \langle t(a)|,\quad
(\otimes _{n=1}^{\mathsf{N}}|n,h_{n}\rangle )\otimes |t(a)\rangle ,
\end{equation}
obtained by tensoring  common eigenstates of the commuting operators $\tau $ and $\sigma _{n}^{z}$, $1\le n\le \mathsf{N}$, provide left and right dynamical-spin basis in $\mathbb{D}_{\mathsf{(6VD)},\mathsf{N}}^{\mathcal{L}}$ and $\mathbb{D}_{\mathsf{(6VD)},\mathsf{N}}^{\mathcal{R}}$, respectively.
The following scalar product is thereby
naturally induced in the linear space $\mathbb{D}_{\mathsf{(6VD)},\mathsf{N}}^{\mathcal{R}}$:
\begin{equation}
\big(\otimes _{n=1}^{\mathsf{N}}|n,h_{n}\rangle \otimes |t(a)\rangle ,
\otimes_{n=1}^{\mathsf{N}}|n,h_{n}^{\prime }\rangle \otimes |t(a^{\prime })\rangle\big)
=\delta _{a,a^{\prime }}\prod_{n=1}^{\mathsf{N}}\delta_{h_{n},h_{n}^{\prime }}.
\end{equation}

\begin{rem}
In the case $(\mathsf{x},\mathsf{y})=(0,0)$ in \eqref{Dynamical-spectrum-2}, the corresponding representation of the dynamical Yang-Baxter algebra is well defined only when $\mathsf{N}$ is odd (to avoid singularities in \eqref{def-abc}); this case was analyzed by SOV in the paper \cite{Nic13a} (see also \cite{FelS99}).
The analysis that we present here will address both the case $\mathsf{N}$ odd for $(\mathsf{x},\mathsf{y})=(0,0)$ and the case $\mathsf{N}$ even and odd for $(\mathsf{x},\mathsf{y})\not=(0,0)$.
This enables us to consider more general SOS models (such as for instance as in \cite{KunY88,PeaS88}), and will be useful for our study of the eight-vertex model with different types of quasi-periodic boundary conditions \cite{NicT15b}.
\end{rem}

The operator $\mathsf{S}_{\tau }$ \eqref{S-tau} defines a natural grading on $\mathbb{D}_{\mathsf{(6VD)},\mathsf{N}}^{\mathcal{L}/\mathcal{R}}$:
\begin{equation}
\mathbb{D}_{\mathsf{(6VD)},\mathsf{N}}^{\mathcal{L}/\mathcal{R}}
=\oplus _{r=-\infty }^{\infty }\mathbb{\bar{D}}_{\mathsf{(6VD)},\mathsf{N}}^{(r,\mathcal{L}/\mathcal{R})},
\end{equation}
where $\mathbb{\bar{D}}_{\mathsf{(6VD)},\mathsf{N}}^{(r,\mathcal{L}/\mathcal{R}) }$ is the $2^{\mathsf{N}}$-dimensional linear eigenspace corresponding
to the eigenvalue $2r\eta +\mathsf{x}\pi+\mathsf{y}\pi\omega $ of $\mathsf{S}_{\tau }$.
In terms of the dynamical-spin basis the linear (covector) space $\mathbb{\bar{D}}_{\mathsf{(6VD)},\mathsf{N}}^{(r,\mathcal{L})}$ and the linear (vector) space $\mathbb{\bar{D}}_{\mathsf{(6VD)},\mathsf{N}}^{(r,\mathcal{R})}$ are respectively generated by the elements
\begin{equation}
\big( \otimes _{n=1}^{\mathsf{N}}\langle n,h_{n}|\big) \otimes \langle t_{r,\mathbf{h}}|
\quad \text{with}\quad 
\langle t_{r,\mathbf{h}}|\equiv \langle t_{0,\mathbf{h}}|\mathsf{T}_{\tau }^{r}\, ,
\label{DyS-basis-L}
\end{equation}
and
\begin{equation}
\big(\otimes _{n=1}^{\mathsf{N}}|n,h_{n}\rangle \big) \otimes |t_{r,\mathbf{h}}\rangle
\quad \text{with}\quad 
|t_{r,\mathbf{h}}\rangle \equiv \mathsf{T}_{\tau }^{-r}|t_{0,\mathbf{h}}\rangle .
\label{DyS-basis-R}
\end{equation}
Here we have set
\begin{equation}
t_{r,\mathbf{h}}\equiv -\frac{\eta }{2}\mathsf{s}_{\mathbf{h}}+\mathsf{x}\frac{\pi }{2}+\mathsf{y}\frac{\pi}{2}\omega+r\eta ,
\quad\text{with}\quad 
\mathsf{s}_{\mathbf{h}}\equiv\sum_{k=1}^{\mathsf{N}}(1-2h_{k})\quad\text{and}\quad
 \mathbf{h}\equiv (h_{1},\ldots,h_{\mathsf{N}}).
\label{DyS-basis-L+}
\end{equation}

\begin{proposition}\label{prop-Dr}
For each $r\in\mathbb{Z}$, the finite-dimensional vector spaces
$\mathbb{\bar{D}}_{\mathsf{(6VD)},\mathsf{N}}^{(r,\mathcal{L}/\mathcal{R})}$
are invariant under the action of the operators
\begin{equation}
\mathsf{A}(\lambda |\tau ),\text{ \ }\mathsf{D}(\lambda |\tau ),\text{ \ }
\mathcal{B}(\lambda  ),\text{ \ }\mathcal{C}(\lambda  ).
\end{equation}
\end{proposition}

\begin{proof}
The commutation relations
\begin{alignat}{2}
&[ \mathsf{A}(\lambda |\tau ),\mathsf{S}] =[\mathsf{A}(\lambda |\tau ),\tau ]
  =  &\ &[\mathsf{D}(\lambda |\tau ),\mathsf{S}] =[\mathsf{D}(\lambda |\tau ),\tau ]=0, \displaybreak[0]\\
&[\mathcal{B}(\lambda  ),\mathsf{S}] =2\mathcal{B}(\lambda  ),
& &[\mathcal{B}(\lambda  ),\tau ]=-\eta \mathcal{B}(\lambda ), \displaybreak[0]\\\
&[ \mathcal{C}(\lambda ),\mathsf{S}] =-2\mathcal{C}(\lambda  ),
& &[\mathcal{C}(\lambda  ),\tau ]=\eta \mathcal{C}(\lambda ),
\end{alignat}
imply that
\begin{equation}
 [ \mathsf{A}(\lambda |\tau ),\mathsf{S}_{\tau }]=[\mathsf{D}(\lambda|\tau ),\mathsf{S}_{\tau }]
 =[\mathcal{B}(\lambda  ),\mathsf{S}_{\tau}]=[\mathcal{C}(\lambda  ),\mathsf{S}_{\tau }]=0,
\end{equation}
which means that $\mathbb{\bar{D}}_{\mathsf{(6VD)},\mathsf{N}}^{(r,\mathcal{L}/\mathcal{R})}$ are invariant under the action of these operators.
\end{proof}

For a SOS model with free boundary conditions, the physical space of states corresponds to the whole representation space $\mathbb{D}_{\mathsf{(6VD)},\mathsf{N}}$. If we impose different types of boundary conditions, the physical space of states will correspond only to a subspace of $\mathbb{D}_{\mathsf{(6VD)},\mathsf{N}}$. For instance, in the case of periodic boundary conditions, it corresponds to the subspace of \eqref{space} associated with the zero eigenvalue of $\mathsf{S}$, whereas in the case of antiperiodic boundary conditions that we consider in this paper, it corresponds to the subspace of \eqref{space} associated with the eigenvalue $\mathsf{x}\pi+\mathsf{y}\pi\omega$ of $\mathsf{S}_\tau$, i.e. to $\mathbb{\bar{D}}_{\mathsf{(6VD)},\mathsf{N}}^{(0)}$. In fact, we shall see in Section~\ref{sec-diag-t} that the commutation of the {\em antiperiodic} transfer matrices is ensured on this subspace only. Nevertheless, the construction of the SOV basis that we present in the next section holds in the whole representation space $\mathbb{D}_{\mathsf{(6VD)},\mathsf{N}}$.

\begin{rem}
In the case of periodic boundary conditions, the space of states of the (finite-size) model is usually infinite dimensional, which may be a technical inconvenience for the study of the model. In fact, due to this reason, one usually deals with restricted models for which the crossing parameter is rational and the space of states is finite dimensional (such as for instance in the original paper \cite{Bax73a}). Let us stress that, for antiperiodic boundary conditions, we do not have this problem since the space of states $\mathbb{\bar{D}}_{\mathsf{(6VD)},\mathsf{N}}^{(0)}$ has dimension $2^\mathsf{N}$. 
\end{rem}

\section{SOV basis in left and right representation spaces}
\label{sec-SOVbasis}

For usual vertex models such as the 6-vertex model \cite{Skl90,Skl92,NicT10,GroN12,Nic13,Nic13b,NicT15}, the SOV approach to diagonalize the transfer matrix is based on the construction of a basis of the space of states of the model which explicitly diagonalizes the action of one particular generator of the Yang-Baxter algebra. In the dynamical case, however, the corresponding SOV basis for the antiperiodic transfer matrix \cite{FelS99,Nic13a} only partially diagonalizes the operator $\mathcal{D}(\lambda)$ (or $\mathcal{A}(\lambda)$). In other words, the latter still acts as a shift operator on some ``dynamical" part of the corresponding basis, whereas the operators $\mathcal{B}(\lambda)$ and $\mathcal{C}(\lambda)$ act as a sum of shift operators on the ``spin" part of the basis. In this section, we present the explicit construction of this SOV basis in the left and right representation spaces of the model.

In each of the subspaces $\mathbb{\bar{D}}_{\mathsf{(6VD)},\mathsf{N}}^{(r,\mathcal{L}/\mathcal{R})}$, we define the following {\em reference states}:
%
\begin{equation}\label{ref-states}
\langle r,\mathbf{0}|
\equiv \frac{1}{\mathsc{n}}
\big(\otimes _{n=1}^{\mathsf{N}}\langle n,h_{n}=0| \big)\otimes \langle t_{r,\mathbf{0}}|,
\qquad
|\mathbf{1},r\rangle
 \equiv \frac{1}{\mathsc{n}}
 \big( \otimes _{n=1}^{\mathsf{N}}|n,h_{n}=1\rangle \big)
\otimes |t_{r,\mathbf{1}}\rangle ,
\end{equation}
where we have used the notations 
$\mathbf{0}\equiv (h_{1}=0,\ldots,h_{\mathsf{N}}=0)$ and 
$\mathbf{1}\equiv (h_{1}=1,\ldots,h_{\mathsf{N}}=1)$, and where $\mathsc{n}$ is a fixed normalization constant that will be specified latter.
It is easy to see that these references states are $\mathsf{D}(\lambda |\tau )$
and $\mathsf{A}(\lambda |\tau )$-eigenstates with eigenvalues given in terms of $\mathsc{a}(\lambda)$ and $\mathsc{d}(\lambda)$  (see \eqref{a-d}),
\begin{alignat}{2}
 &\langle r,\mathbf{0}|\, \mathsf{A}(\lambda |\tau )=\text{\textsc{a}}(\lambda )\, \langle r,\mathbf{0}|,
 &\qquad
 &\langle r,\mathbf{0}|\, \mathsf{D}(\lambda |\tau )
 =e^{-i\mathsf{N}\mathsf{y}\eta}\,\frac{\theta(t_{ r,\mathbf{0}}-\eta )}{\theta (t_{ r,\mathbf{1}}-\eta )}\, \text{\textsc{d}}(\lambda )\, \langle r,\mathbf{0}|,
\\
 &\mathsf{D}(\lambda |\tau )\, |\mathbf{1},r\rangle = \text{\textsc{a}}(\lambda )\, |\mathbf{1},r\rangle,
 &\quad
 &\mathsf{A}(\lambda |\tau)\, |\mathbf{1},r\rangle =e^{-i\mathsf{N}\mathsf{y}\eta}\,\frac{\theta(t_{- r,\mathbf{0}}-\eta )}{\theta (t_{- r,\mathbf{1}}-\eta )}\, \text{\textsc{d}}(\lambda )\, |\mathbf{1},r\rangle,\end{alignat}
%
and are annihilated by the action of the operators $\mathcal{B}(\lambda )$.
Then, for each $\mathsf{N}$-tuple $\mathbf{h}\equiv (h_{1},\ldots,h_{\mathsf{N}})\in \{0,1\}^{\mathsf{N}}$, we construct a state $\langle r,\mathbf{h}|\in\mathbb{D}_{\mathsf{(6VD)},\mathsf{N}}^{\mathcal{L}}$ and a state $| \mathbf{h},r\rangle \in\mathbb{D}_{\mathsf{(6VD)},\mathsf{N}}^{\mathcal{R}}$ as
\begin{align}
 &\langle r,\mathbf{h}|\equiv 
 \langle r,\mathbf{0}| \prod_{n=1}^{\mathsf{N}}
 \left( \frac{\mathcal{C}(\xi_{n} )}{\text{\textsc{d}}(\xi _{n}-\eta )}\right)^{h_{n}},
\label{D-left-eigenstates}\\
 &|\mathbf{h},r\rangle 
 \equiv \prod_{n=1}^{\mathsf{N}}
 \left( \frac{\mathcal{C}(\xi _{n}-\eta  )}{\text{\textsc{d}}(\xi _{n}-\eta )}\right) ^{(1-h_{n})}|\mathbf{1},r\rangle .
 \label{D-right-eigenstates}
\end{align}
%

\begin{rem}
It is easy to see that, $\forall\, \mathbf{h}\in\{0,1\}^{\mathsf{N}}$,
\begin{alignat}{2}
&\langle r,\mathbf{h}|\, \tau =t_{r,\mathbf{h}}\, \langle r,\mathbf{h}|,
& \qquad
&\tau\, |\mathbf{h},r\rangle = t_{r,\mathbf{h}}\, |\mathbf{h},r\rangle,   \label{act-tau}\\
&\langle r,\mathbf{h}|\, \mathsf{S}=\mathsf{s}_{\mathbf{h}}\, \langle r,\mathbf{h}|,
&\qquad
&\mathsf{S}\, |\mathbf{h},r\rangle = \mathsf{s}_{\mathbf{h}}\, |\mathbf{h},r\rangle,
\label{act-S}
\end{alignat}
with $t_{r,\mathbf{h}}$ and $\mathsf{s}_{\mathbf{h}}$ given by \eqref{DyS-basis-L+},
which also implies that $\langle r,\mathbf{h}|\in \mathbb{\bar{D}}_{\mathsf{(6VD)},\mathsf{N}}^{(r,\mathcal{L})}$ and $|\mathbf{h},r\rangle \in \mathbb{\bar{D}}_{\mathsf{(6VD)},\mathsf{N}}^{(r,\mathcal{R})}$. 
\end{rem}

We have the following result:

\begin{theorem}\label{thm-D-basis}
Let $\Gamma=\pi\mathbb{Z}+\pi\omega\mathbb{Z}$ and let the inhomogeneity parameters $\xi _{1},\ldots,\xi _{\mathsf{N}}\in \mathbb{C}$ be such that
\begin{equation}
\forall\epsilon\in  \{-1,0,1\}, \quad \xi _{a}- \xi_{b}+\epsilon\eta \notin \Gamma
 \quad \text{if}\  a\neq b.
    \label{cond-inh}
\end{equation}
Then the set of states $\{ \langle r,\mathbf{h}|,\mathbf{h}\in\{0,1\}^{\mathsf{N}}\}$ defined by \eqref{D-left-eigenstates} form a basis of $\mathbb{\bar{D}}_{\mathsf{(6VD)},\mathsf{N}}^{(r,\mathcal{L})}$.
Similarly, the set of states $\{ | \mathbf{h},r\rangle,\mathbf{h}\in\{0,1\}^{\mathsf{N}}\}$ defined by \eqref{D-right-eigenstates} form a basis of $\mathbb{\bar{D}}_{\mathsf{(6VD)},\mathsf{N}}^{(r,\mathcal{R})}$.
Moreover, the action of  $\mathcal{D}(\lambda  )$ on the states $ \langle r,\mathbf{h}|$ in $\mathbb{D}_{\mathsf{(6VD)},\mathsf{N}}^{\mathcal{L}}$ (respectively  on the states $ |\mathbf{h},r\rangle$ in $\mathbb{D}_{\mathsf{(6VD)},\mathsf{N}}^{\mathcal{R}}$) is given by
\begin{align}
&\langle r,\mathbf{h}|\, \mathcal{D}(\lambda  )
=\mathsf{d}_{r-1,\mathbf{h}}(\lambda )\, \langle r-1,\mathbf{h}|,
  \label{L-eigen-1}\\
&\mathcal{D}(\lambda  )\, |\mathbf{h},r\rangle 
=\mathsf{d}_{r+1,\mathbf{h}}(\lambda )\, |\mathbf{h},r+1\rangle ,
  \label{R-eigen-1}
\end{align}
where
\begin{equation}
\mathsf{d}_{r,\mathbf{h}}(\lambda )\equiv e^{-i\mathsf{y}\eta\frac{s_\mathbf{h}-s_\mathbf{1}}{2}}\,
\frac{\theta(t_{r,\mathbf{h}})}{\theta (t_{r,\mathbf{1}})}\, 
\prod_{n=1}^{\mathsf{N}}\theta (\lambda -\xi _{n}^{(h_{n})}),
\qquad\text{with}\quad
\xi _{a}^{(h_{a})}=\xi _{a}-\eta h_{a}.
\label{EigenValue-D}
\end{equation}
\end{theorem}

\begin{proof}The action of the operator $\mathcal{D}(\lambda  )$ on the states $\langle r,\mathbf{h}|$ and $|\mathbf{h},r\rangle$ can respectively be computed by means of the following dynamical 6-vertex commutation relations issued from \eqref{RTT-op}:
\begin{multline}
\mathcal{C}(\mu  )\,\mathcal{D}(\lambda )
 =\Big[ \mathcal{D}(\lambda  )\,\mathcal{C}(\mu  )\,\theta (\lambda -\mu +\eta )\theta(\tau )
   \\
 -\mathcal{D}(\mu  )\,\mathcal{C}(\lambda  )\,
       e^{i\mathsf{y}(\lambda-\mu)}\theta (\eta)\theta (\tau +\lambda -\mu )\Big]  
       \frac{e^{i\mathsf{y}\eta}}{\theta (\lambda -\mu )\theta (\tau -\eta )},
\label{C-D-Comm-Left} 
\end{multline}
\begin{multline}
\mathcal{D}(\lambda  )\,\mathcal{C}(\mu  ) 
=\frac{e^{-i\mathsf{y}\eta}}{\theta (\mu-\lambda ) \theta (\tau +\eta )}
  \Big[\theta (\mu -\lambda +\eta )\theta (\tau )\,
         \mathcal{C}(\mu  )\, \mathcal{D}(\lambda  ) 
          \\
 -e^{-i\mathsf{y}(\mu-\lambda)}\theta (\eta )\theta (\tau -\mu +\lambda )\,
        \mathcal{C}(\lambda  )\, \mathcal{D}(\mu  )\Big]. 
        \label{C-D-Comm-Right}
\end{multline}
Hence, using the fact that $\text{\textsc{d}}(\xi_n )= \text{\textsc{a}}(\xi_n-\eta )=0$, $n=1,\ldots, N$, so that only the direct terms contribute in the above relations, 
we obtain \eqref{L-eigen-1} and \eqref{R-eigen-1}.

It remains to prove that the states $\langle r,\mathbf{h}|$, $\mathbf{h}\in\{0,1\}^\mathsf{N}$, form
a set of $2^{\mathsf{N}}$ independent states, i.e. a basis of $\mathbb{\bar{D}}_{\mathsf{(6VD)},\mathsf{N}}^{(r,\mathcal{L})}$ (the fact that the states $|\mathbf{h},r\rangle $, $\mathbf{h}\in\{0,1\}^\mathsf{N}$,  form a basis of $\mathbb{\bar{D}}_{\mathsf{(6VD)},\mathsf{N}}^{(r,\mathcal{L})}$ can be proven similarly).
Let us suppose that we have a relation of the form
\begin{equation}
\sum_{\mathbf{h}\in\{0,1\}^\mathsf{N}}\! c_{\mathbf{h}}\, \langle r,\mathbf{h}|=0,
 \label{LC-0}
\end{equation}
for some set of complex coefficients $c_{\mathbf{h}}\in\mathbb{C}$, $\mathbf{h}\in\{0,1\}^\mathsf{N}$.
Let $\mathbf{k}=(k_{1},\ldots,k_{\mathsf{N}})$ be any given $\mathsf{N}$-tuple in $\{0,1\}^{\mathsf{N}}$.
Then, by applying the product
\begin{equation}
\prod_{n=1}^{\mathsf{N}}\mathcal{D}(\xi _{n}^{( \bar{k}_{n})} )
\qquad \text{with}\quad \bar{k}_{n}\in\{0,1\},\quad \bar{k}_{n}=k_{n}+1\mod 2,
\end{equation}
to \eqref{LC-0}, we get the identity
\begin{equation}
c_\mathbf{k}\, \prod_{n=1}^{\mathsf{N}}\mathsf{d}_{r-n,\mathbf{k}}(\xi _{n}^{( \bar{k}_{n}) })\,
 \langle r-\mathsf{N},\mathbf{k}| = 0.
\end{equation}
Since $\mathsf{d}_{r-n,\mathbf{k}}(\xi _{n}^{( \bar{k}_{n}) })\not=0$, $1\le n\le \mathsf{N}$, and $ \langle r-\mathsf{N},\mathbf{k}| \not=0$, this implies that $c_{\mathbf{k}}=0$, which ends the proof of Theorem~\ref{thm-D-basis}. 
\end{proof}

We want to determine the action of the remaining generators of the Yang-Baxter algebra on the basis elements \eqref{D-left-eigenstates} and \eqref{D-right-eigenstates}, and in particular of the operators $\mathcal{B}(\lambda)$ and $\mathcal{C}(\lambda)$ constituting the antiperiodic transfer matrix. To do this, we shall rely on the following lemma:

\begin{lemma}\label{lem-period}
$\mathcal{D}(\lambda  )$, $\widetilde{\mathcal{B}}(\lambda  )\equiv e^{-i\mathsf{y}\lambda}\mathcal{B}(\lambda)$ and $\widetilde{\mathcal{C}}(\lambda  )\equiv e^{i\mathsf{y}\lambda}\mathcal{C}(\lambda )$ are 
entire functions of $\lambda$ which satisfy the following quasi-periodicity properties:
\begin{xalignat}{2}
&\mathcal{D}(\lambda +\pi  )=( -1)^{\mathsf{N}}\, \mathcal{D}(\lambda ),
&\quad
&\mathcal{D}(\lambda +\pi w )=\left(-e^{-2i\lambda-i\pi\omega}\right)^{\!\mathsf{N}}
e^{2i\sum_{n=1}^{\mathsf{N}}(\xi_{n}-\eta/2)+i\eta \mathsf{S}}\, \mathcal{D}(\lambda  ),
\label{Characteristic-D}\\
&\widetilde{\mathcal{B}}(\lambda +\pi  ) =( -1) ^{\mathsf{N}}\, \widetilde{\mathcal{B}}(\lambda  ),
&\ \
&\widetilde{\mathcal{B}}(\lambda +\pi w )=\left(-e^{-2i\lambda-i\pi\omega}\right)^{\!\mathsf{N}}
e^{2i\sum_{n=1}^{\mathsf{N}}(\xi_{n}-\eta /2)-i\mathsf{S}_{\tau }}\,\widetilde{\mathcal{B}}(\lambda  ),
\label{Characteristic-B} 
\\
&\widetilde{\mathcal{C}}(\lambda +\pi  )=( -1) ^{\mathsf{N}}\, \widetilde{\mathcal{C}}(\lambda  ),
&\quad
&\widetilde{\mathcal{C}}(\lambda +\pi w )=\left(-e^{-2i\lambda-i\pi\omega}\right)^{\!\mathsf{N}}
e^{2i\sum_{n=1}^{\mathsf{N}}(\xi_{n}-\eta /2)+i\mathsf{S}_{\tau }}\,\widetilde{\mathcal{C}}(\lambda ).
\label{Characteristic-C}
\end{xalignat}
\end{lemma}

\begin{proof}
The fact that $\mathcal{D}(\lambda )$, $\mathcal{B}(\lambda  )$
and $\mathcal{C}(\lambda  )$ are entire functions of $\lambda$ is a simple consequence of the
definition of the dynamical 6-vertex monodromy matrix in terms of the
dynamical 6-vertex $R$-matrix \eqref{mat-R}.
The quasi-periodicity properties \eqref{Characteristic-D} of $\mathcal{D}(\lambda)$ follows directly from the explicit form  \eqref{L-eigen-1}, \eqref{R-eigen-1}, \eqref{EigenValue-D} of its action on the basis \eqref{D-left-eigenstates} or \eqref{D-right-eigenstates} and from the quasi-periodicity properties of the usual theta function $\theta(\lambda)\equiv\theta_1(\lambda|\omega)$.
The quasi-periodicity properties \eqref{Characteristic-C} of $\widetilde{\mathcal{C}}(\lambda)$ can be deduced from the quasi-periodicity properties \eqref{Characteristic-D} of $\mathcal{D}(\lambda)$ by comparing the quasi-periodicity properties in $\lambda$ of the two members of the commutation relation \eqref{C-D-Comm-Left} rewritten in the form
\begin{equation*}
\mathcal{C}(\lambda)\,\mathcal{D}(\mu)
= e^{-i\mathsf{y}(\lambda-\mu-\eta)}\,
\frac{\theta(\lambda-\mu)\, \theta(\tau+\eta)}{\theta(\lambda-\mu+\tau)\, \theta(\eta)}\,  
\mathcal{D}(\lambda  )\,\mathcal{C}(\mu ) 
- e^{-i\mathsf{y}(\lambda-\mu)}\, 
\frac{\theta (\lambda-\mu-\eta )\,  \theta (\tau)}{\theta(\lambda-\mu+\tau)\,\theta(\eta)}\,
         \mathcal{C}(\mu  )\, \mathcal{D}(\lambda ) .
\end{equation*}
Similarly, the quasi-periodicity properties of $\widetilde{\mathcal{B}}(\lambda)$ can be obtained from \eqref{Characteristic-D} and from the commutation relation
\begin{equation*}
  \mathcal{B}(\lambda)\,\mathcal{D}(\mu)
= e^{i\mathsf{y}(\lambda-\mu)}\,
\frac{\theta(\lambda-\mu+\eta)\, \theta(\tau+\eta\mathsf{S})}{\theta(\tau+\eta\mathsf{S}-\lambda+\mu)\, \theta(\eta)}\,  
\mathcal{B}(\mu  )\,\mathcal{D}(\lambda  ) 
- e^{i\mathsf{y}(\lambda-\mu-\eta)}\, 
\frac{\theta (\lambda-\mu )\,  \theta (\tau+\eta\mathsf{S}-\eta)}{\theta(\tau+\eta\mathsf{S}-\lambda+\mu)\,\theta(\eta)}\,
         \mathcal{D}(\lambda  )\, \mathcal{B}(\mu  ) .
\end{equation*}
\end{proof}

\begin{theorem}\label{thm-action}
Under the hypothesis \eqref{cond-inh} of Theorem~\ref{thm-D-basis}, the action of the operators $\mathsf{A}(\lambda |\tau )$, $\mathsf{D}(\lambda |\tau )$, $\mathcal{B}(\lambda  )$ and $\mathcal{C}(\lambda  )$ on the basis elements $\langle r,\mathbf{h} |$ of $\mathbb{\bar{D}}_{\mathsf{(6VD)},\mathsf{N}}^{(r,\mathcal{L})}$ (respectively on the basis elements $| \mathbf{h} ,r\rangle$ of $\mathbb{\bar{D}}_{\mathsf{(6VD)},\mathsf{N}}^{(r,\mathcal{R})}$) take the following form:

\textsf{I)} \underline{Left representations:}
\begin{align}
&\langle r,\mathbf{h} | \, \mathsf{D}(\lambda |\tau )
=\mathsf{d}_{r-1,\mathbf{h}}(\lambda )\,
\left( \langle r,\mathbf{0}|
\prod_{n=1}^{\mathsf{N}}\left(\frac{ \mathsf{T}_{\tau }^{-}\,\mathcal{C}(\xi _{n} )\, \mathsf{T}_{\tau }^{+}}{\mathsc{d}(\xi _{n}-\eta )}\right)
^{h_{n}}\right) \, ,  
\label{D-L-EigenV}
\\
&\langle r,\mathbf{h}|\,\mathcal{C}(\lambda  ) 
=\sum_{a=1}^{\mathsf{N}}
e^{i\mathsf{y}(\xi _{a}^{(h_{a})}-\lambda )}\,
\frac{\theta (t_{r,\mathbf{h}}-\lambda +\xi _{a}^{(h_{a})})}{\theta (t_{r,\mathbf{h}})}
  \prod_{b\neq a}\frac{\theta (\lambda -\xi_{b}^{(h_{b})})}{\theta (\xi _{a}^{(h_{a})}-\xi _{b}^{(h_{b})})}\,
  \text{\textsc{d}}(\xi _{a}^{(1-h_{a})})\,
  \langle r,\mathsf{T}_{a}^{+} \mathbf{h}|\, , 
 \label{C-SOV_D-left}
  \\
&\langle r,\mathbf{h}|\, \mathcal{B}(\lambda  )
 =\sum_{a=1}^{\mathsf{N}}e^{i\mathsf{y}(\xi _{a}^{(h_{a})}-\lambda )}\,
   \frac{\theta (t_{-r,\mathbf{h}}-\lambda +\xi_{a}^{(h_{a})})}{\theta (t_{-r,\mathbf{h}})} 
   \prod_{b\neq a}\!\frac{\theta (\lambda -\xi _{b}^{(h_{b})})}{\theta (\xi _{a}^{(h_{a})}-\xi_{b}^{(h_{b})})}
   \mathsc{a}_{\mathsf{x},\mathsf{y},r,\mathbf{h}}^{(-)}(\xi_{a}^{(1-h_{a})})\,
   \langle r,\mathsf{T}_{a}^{-}\mathbf{h}|\, , \hspace{-1mm}
\label{B-SOV_D-left}
\end{align}
and the action of $\mathsf{A}(\lambda |\tau )$ is completely determined by the quantum
determinant relations \eqref{q-det}.

\textsf{II)} \underline{Right representations:} 
\begin{align}
&\mathsf{D}(\lambda |\tau+\eta )\,  |\mathbf{h},r\rangle
=\mathsf{d}_{r+1,\mathbf{h}}(\lambda )\,
\left( \,
\prod_{n=1}^{\mathsf{N}}
\left( \frac{\mathsf{T}_{\tau }^{+}\,\mathcal{C}(\xi _{n}-\eta  )\, \mathsf{T}_{\tau }^-}{\mathsc{d}(\xi _{n}-\eta )}\right) ^{(1-h_{n})}
|\mathbf{1},r\rangle \right) \, ,  
\label{D-R-EigenV}
\\
&\mathcal{C}(\lambda  )\, |\mathbf{h},r\rangle 
 =\sum_{a=1}^{\mathsf{N}} e^{i\mathsf{y}(\xi _{a}^{(h_{a})}-\lambda )}\,
\frac{\theta (t_{r,\mathbf{h}}-\lambda +\xi _{a}^{(h_{a})})}{\theta(t_{r,\mathbf{h}})}
\prod_{b\neq a}\frac{\theta (\lambda -\xi_{b}^{(h_{b})})}{\theta (\xi _{a}^{(h_{a})}-\xi _{b}^{(h_{b})})}\,\text{\textsc{d}}(\xi _{a}^{(h_{a})})\,
|\mathsf{T}_{a}^{-}\mathbf{h},r\rangle\,  ,  
\label{C-SOV_D-right} 
 \\
&\mathcal{B}(\lambda  )\, |\mathbf{h},r\rangle  
=\sum_{a=1}^{\mathsf{N}} e^{i\mathsf{y}(\xi _{a}^{(h_{a})}-\lambda )}\,
\frac{\theta (t_{-r,\mathbf{h}}-\lambda +\xi _{a}^{(h_{a})})}{\theta(t_{-r,\mathbf{h}})}
\prod_{b\neq a}\frac{\theta (\lambda -\xi_{b}^{(h_{b})})}{\theta (\xi _{a}^{(h_{a})}-\xi _{b}^{(h_{b})})}\,
\mathsc{a}_{\mathsf{x},\mathsf{y},r,\mathbf{h}}^{(+)}(\xi _{a}^{(h_{a})})\,
 |\mathsf{T}_{a}^{+}\mathbf{h},r\rangle \, ,
\label{B-SOV_D-right}
\end{align}
and the action of $\mathsf{A}(\lambda |\tau )$ is uniquely determined by the quantum
determinant relation \eqref{q-det}.

In the above expressions, we have used the definition  \eqref{EigenValue-D} of $\mathsf{d}_{r,\mathbf{h}}(\lambda )$ and $\xi _{a}^{(h_{a})}$, and we have set
\begin{align}
&\mathsf{T}_{a}^{\pm }(h_{1},\ldots,h_{\mathsf{N}})
=(h_{1},\ldots,h_{a}\pm 1,\ldots,h_{\mathsf{N}}),
\\
&\mathsc{a}_{\mathsf{x},\mathsf{y},r,\mathbf{h}}^{(\pm )}(\lambda )
=(-1)^{\mathsf{x}+\mathsf{y}+\mathsf{x}\mathsf{y}}\, e^{2i\mathsf{y}r\eta}\,
\frac{\theta (t_{r,\mathbf{h}}\pm \eta )}{\theta (t_{-r,\mathbf{h}}\pm \eta  )}\,\mathsc{a}(\lambda ).
\end{align}
\end{theorem}

\begin{proof}
The action \eqref{D-L-EigenV}, \eqref{D-R-EigenV} of $\mathsf{D}(\lambda |\tau )=\mathcal{D}(\lambda
 )\,\mathsf{T}_{\tau }^{+}$ and $\mathsf{D}(\lambda |\tau+\eta )=\mathsf{T}_{\tau }^{+}\,\mathcal{D}(\lambda
 )$ on $\langle r,\mathbf{h}|$ and $|\mathbf{h},r\rangle$ respectively follow directly from the formulae  \eqref{L-eigen-1}, \eqref{R-eigen-1} and from the commutation relations \eqref{Dyn-op-comm}. 

To determine the action of $\mathcal{B}(\lambda)$ and $\mathcal{C}(\lambda)$ on a left state $\langle r,\mathbf{h}|$ (respectively on a right state $|\mathbf{h},r\rangle$) associated with a given $\mathsf{N}$-tuple $\mathbf{h}=(h_1,\ldots,h_\mathsf{N})$,  we first compute these actions at the $\mathsf{N}$ special points $\lambda=\xi _{n}^{( h_{n}) }$, $n=1,\ldots,\mathsf{N}$.
From the definitions \eqref{D-left-eigenstates}, \eqref{D-right-eigenstates}, and the fact that the product $\mathcal{C}(\xi_n)\,\mathcal{C}(\xi_n-\eta)$ vanishes when evaluated at one of the inhomogeneity parameters $\xi_n$ (see \eqref{cancel-inh}), we get
\begin{equation}\label{act-C}
  \langle r,\mathbf{h}|\,\mathcal{C}(\xi_n^{(h_n)})
  =\mathsc{d}(\xi_n^{(1-h_n)}) \,  \langle r,\mathsf{T}_n^{+}\mathbf{h}| \, ,
  \qquad
  \mathcal{C}(\xi_n^{(h_n)})\, |\mathbf{h},r\rangle
  =\mathsc{d}(\xi_n^{(h_n)}) \, | \mathsf{T}_n^-\mathbf{h},r\rangle \, .
\end{equation}
The action of $\mathcal{B}(\xi _{n}^{\left( h_{n}\right) } )$ on the left and right
states \eqref{D-left-eigenstates} and \eqref{D-right-eigenstates} can be computed by means of the dynamical Yang-Baxter commutation relations \eqref{RTT-op} and the quantum determinant relations \eqref{q-det}. We obtain
\begin{align}
   \langle r, \mathbf{h} |\,\mathcal{B}(\xi_n^{(h_n)})
   &=-\mathsc{a}(\xi_n^{(1-h_n)})\,  \langle r, \mathsf{T}_n^-\mathbf{h} |\,
   e^{-i\mathsf{y}\eta\mathsf{S}}\,\frac{\theta(\tau)}{\theta(\tau+\eta\mathsf{S})},
   \notag\\
   &=(-1)^{\mathsf{x}+\mathsf{y}+\mathsf{x}\mathsf{y}}\,
   \mathsc{a}(\xi_n^{(1-h_n)})\, e^{2i\mathsf{y}r\eta}\,
   \frac{\theta(t_{r,\mathbf{h}}-\eta)}{\theta(t_{-r,\mathbf{h}}-\eta)}\,
   \langle r, \mathsf{T}_n^-\mathbf{h} |\, ,
   \label{act-C-xi_n}
\end{align}
and
\begin{align}
   \mathcal{B}(\xi_n^{(h_n)})\,|\mathbf{h},r\rangle 
   &=-\mathsc{a}(\xi_n^{(h_n)})\, e^{-i\mathsf{y}\eta\mathsf{S}}\frac{\theta(\tau)}{\theta(\tau+\eta\mathsf{S})}
          \,|\mathsf{T}_n^+\mathbf{h},r\rangle \, ,
          \notag\\
   &=   (-1)^{\mathsf{x}+\mathsf{y}+\mathsf{x}\mathsf{y}}\,
   \mathsc{a}(\xi_n^{(h_n)})\, e^{2i\mathsf{y}r\eta}\,
   \frac{\theta(t_{r,\mathbf{h}}+\eta)}{\theta(t_{-r,\mathbf{h}}+\eta)}\,
   | \mathsf{T}_n^+\mathbf{h} , r\rangle \, ,    
   \label{act-B-xi_n}
\end{align}
in which we have used the quasi-periodicity properties of the theta function.

Finally, the action of $\mathcal{B}(\lambda )$ and $\mathcal{C}(\lambda )$ to the left \eqref{C-SOV_D-left}-\eqref{B-SOV_D-left} and to the right 
\eqref{C-SOV_D-right}-\eqref{B-SOV_D-right} are obtained from \eqref{act-C-xi_n} and \eqref{act-B-xi_n} and from Lemma~\ref{lem-period} by means of interpolation formulae (see \eqref{interpolation}).
\end{proof}

\begin{rem}
It is easy to check that the representations of the Yang-Baxter generators in the SOV basis coincide with their representations in the so-called $F$-basis \cite{MaiS00} as obtained in \cite{AlbBFPR00}.
Such a relation between the SOV construction and the $F$-basis was first noticed in \cite{Ter99}.
\end{rem}

The basis elements \eqref{D-left-eigenstates} and \eqref{D-right-eigenstates} have been constructed independently in the left and right representation spaces respectively. It is nevertheless not difficult to compute their scalar product \cite{Nic13a}:

\begin{proposition}\label{prop-sc}
Under the hypothesis  \eqref{cond-inh}  of Theorem~\ref{thm-D-basis}, let us define, for each $\mathbf{h}\in\{0,1\}^{\mathsf{N}}$ and $r\in\mathbb{Z}$, the $\mathsf{N}\times\mathsf{N}$ matrices $\Theta^{(r,\mathbf{h})}$ of elements
\begin{equation}
  \big[ \Theta^{(r,\mathbf{h})} \big]_{ij}=\vartheta_{j-1}(\xi_i^{(h_i)}-\bar\xi_r ),
   \qquad
   \text{with}
   \quad
   \bar\xi_r=\frac{1}{\mathsf{N}}\left(\sum_{k=1}^\mathsf{N}\xi_k+t_{r,\mathbf{0}}\right),
\end{equation}
where the functions $\vartheta_j$ are defined as
\begin{equation}\label{def-theta_j}
  \vartheta_j(\lambda)=\sum_{n\in\mathbb{Z}} 
  e^{i\pi\mathsf{N}\omega(n+\frac12-\frac{j}{\mathsf{N}})^2+2i\mathsf{N}(n+\frac12-\frac{j}{\mathsf{N}})(\lambda-\frac{\pi}{2})},
  \quad 0\le j\le \mathsf{N}-1.
\end{equation}
Then, one can fix the normalization constant in \eqref{ref-states} such that, $\forall \mathbf{h},\tilde{\mathbf{h}}\in\{0,1\}^\mathsf{N}$, $\forall r,\tilde{r}\in\mathbb{Z}$,
\begin{equation}\label{sc-rh}
   \langle r,\mathbf{h} | \tilde{\mathbf{h}},\tilde{r} \rangle 
   = \delta_{r,\tilde{r}}\,\delta_{\mathbf{h},\tilde{\mathbf{h}}}\, 
   \frac{\theta(t_{r,\mathbf{1}})}{\theta(t_{0,\mathbf{1}})}
   \frac{e^{-i\mathsf{y}\eta \sum_{k=1}^\mathsf{N}h_k}}
           {\det_{\mathsf{N}}\big[ \Theta^{(r,\mathbf{h})} \big]}.
\end{equation}
\end{proposition}

\begin{proof}
The fact that $\langle r,\mathbf{h} |$ and $| \tilde{\mathbf{h}},\tilde{r} \rangle$ are both eigenstates of $\mathsf{S}_\tau$ with respective eigenvalues $2r\eta+\mathsf{x}\pi+\mathsf{y}\pi\omega$ and $2\tilde{r}\eta+\mathsf{x}\pi+\mathsf{y}\pi\omega$ implies that $ \langle r,\mathbf{h} | \tilde{\mathbf{h}},\tilde{r} \rangle $ vanishes if $r\not=\tilde{r}$.
Similarly, from the consideration of the matrix element
\begin{equation}\label{el-D}
  \langle r+1,\mathbf{h} |\,\mathcal{D}(\lambda)\, | \tilde{\mathbf{h}}, r \rangle 
\end{equation}
and the explicit expression \eqref{L-eigen-1}, \eqref{R-eigen-1} and \eqref{EigenValue-D} of the action of $\mathcal{D}(\lambda)$ to the left and to the right, one gets that  $ \langle r,\mathbf{h} | \tilde{\mathbf{h}},\tilde{r} \rangle $  vanishes if $\mathbf{h}\not=\tilde{\mathbf{h}}$. One also gets that
\begin{equation}\label{ratio1}
   \frac{\langle r+1,\mathbf{h} | \mathbf{h},r+1\rangle }{\langle r,\mathbf{h} | \mathbf{h},r\rangle }
   = \frac{\theta(t_{r,\mathbf{h}})\,\theta(t_{r+1,\mathbf{1}})}{\theta(t_{r,\mathbf{1}})\,\theta(t_{r+1,\mathbf{h}})}.
\end{equation}
Also, from the consideration of the matrix element
\begin{equation}\label{el-C}
  \langle r,\mathbf{h} |\,\mathcal{C}(\xi_a)\, | \mathsf{T}_a^+{\mathbf{h}}, r \rangle 
\end{equation}
for any $\mathsf{N}$-tuple $\mathbf{h}$ such that $h_a=0$, and the explicit expressions \eqref{act-C}, \eqref{C-SOV_D-right} of the action of $ \mathcal{C}(\xi_a)$ to the left and to the right, one gets that
\begin{equation}\label{ratio2}
   \frac{\langle r,\mathsf{T}_a^+\mathbf{h} | \mathsf{T}_a^+\mathbf{h},r\rangle }{\langle r,\mathbf{h} | \mathbf{h},r\rangle }
   = e^{-i\mathsf{y}\eta}\,\frac{\theta(t_{r,\mathbf{h}})}{\theta(t_{r,\mathsf{T}_a^+\mathbf{h}})}\,
   \prod_{b\neq a}\frac{\theta(\xi_a^{(0)}-\xi_b^{(h_b)})}{\theta(\xi_a^{(1)}-\xi_b^{(h_b)})}.
\end{equation}
Since $\{\vartheta_{j-1}\}_{1\le j\le \mathsf{N}}$ is a basis of the space of theta functions of order $\mathsf{N}$ and of norm $0$ (see Appendix~\ref{app-theta}), we can use \eqref{det-thetaj} to express the determinant of the matrix $\Theta^{(r,\mathbf{h})}$ as
%
\begin{equation}\label{id-det}
  \det \big[ \Theta^{(r,\mathbf{h})}  \big]
  = c_\mathsf{N}\, \theta(t_{r,\mathbf{h}})\,\prod_{i<j}\theta\big(\xi_i^{(h_i)}-\xi_j^{(h_j)}\big),
\end{equation} 
where $c_\mathsf{N}$ is a constant (i.e. it does not depend on  the variables $\xi_i$, on $r$ nor on $\mathbf{h}$).
Hence, it follows from \eqref{ratio1}, \eqref{ratio2} that
\begin{equation}
  \langle r,\mathbf{h} | \mathbf{h},r\rangle 
  =\tilde{c}_\mathsf{N}\,
   \frac{e^{-i\mathsf{y}\eta\sum_{k=1}^\mathsf{N} h_k}\, \theta(t_{r,\mathbf{1}})}
           {\det_{\mathsf{N}}\big[ \Theta^{(r,\mathbf{h})} \big]},
\end{equation}
where $\tilde{c}_\mathsf{N}$ is a constant. 
\end{proof}

It follows from Proposition~\ref{prop-sc} that
\begin{equation}\label{Id}
  \mathbb{I}\equiv \sum_{\substack{r\in\mathbb{Z}\\ \mathbf{h}\in\{0,1\}^{\mathsf{N}} }}
    \frac{|\mathbf{h},r \rangle \langle r, \mathbf{h}|}{\langle r, \mathbf{h} | \mathbf{h},r\rangle},
\end{equation}
with normalization $\langle r,\mathbf{h}| \mathbf{h},r\rangle $ given by \eqref{sc-rh}, provides a decomposition of the identity on the whole representation space $\mathbb{D}_{\mathsf{(6VD)},\mathsf{N}}^{\mathcal{L}/\mathcal{R}}$. Similarly, for any given $r\in\mathbb{Z}$,
\begin{equation}\label{Idr}
  \mathbb{I}_r\equiv \sum_{ \mathbf{h}\in\{0,1\}^{\mathsf{N}} }
    \frac{|\mathbf{h},r \rangle \langle r, \mathbf{h}|}{\langle r, \mathbf{h} | \mathbf{h},r\rangle},
\end{equation}
provides a decomposition of the identity on the subspace $\mathbb{\bar{D}}_{\mathsf{(6VD)},\mathsf{N}}^{(r,\mathcal{L}/\mathcal{R})}$ of $\mathbb{D}_{\mathsf{(6VD)},\mathsf{N}}^{\mathcal{L}/\mathcal{R}}$.

\section{Diagonalization of commuting antiperiodic transfer matrices}
\label{sec-diag-t}

In this section, we diagonalize the $\kappa$-twisted antiperiodic transfer matrices \eqref{anti-transfer} in the space of states $\mathbb{\bar{D}}_{\mathsf{(6VD)},\mathsf{N}}^{(0,\mathcal{L}/\mathcal{R})}$ of the antiperiodic model. As usual in the SOV framework, the spectrum and eigenstates are completely characterized by a discrete system of equations involving the inhomogeneity parameters of the model. We explain how to rewrite this discrete system in terms of functional $T$-$Q$ equations of Baxter's type, so as to obtain a new characterization of the spectrum and eigenstates in terms of solutions of Bethe-type equations.

\subsection{The SOV discrete characterization of the spectrum and eigenstates}

The diagonalization of the antiperiodic transfer matrices in $\mathbb{\bar{D}}_{\mathsf{(6VD)},\mathsf{N}}^{(0,\mathcal{L}/\mathcal{R})}$ has been performed in \cite{FelS99,Nic13a} in the case $(\mathsf{x},\mathsf{y})=(0,0)$ with $\mathsf{N}$ odd. The corresponding SOV procedure can easily be extended to the more general cases that we consider here, and we have the following result:

\begin{theorem}\label{thm-eigen-t}
For any fixed $\kappa \in \mathbb{C}\setminus\{0\}$, the $\kappa $-twisted antiperiodic dynamical 6-vertex transfer matrix $\overline{\mathcal{T}}^{(\kappa)}(\lambda )$ \eqref{anti-transfer} defines a one-parameter family of commuting operators on $\mathbb{\bar{D}}_{\mathsf{(6VD)},\mathsf{N}}^{(0,\mathcal{L}/\mathcal{R})}$.\vspace{-1.3mm}
All these families are isospectral, i.e. the set of the eigenvalues of $\overline{\mathcal{T}}^{(\kappa)}(\lambda )$ is the same for all the values of $\kappa \in \mathbb{C}\setminus\{0\}$ and we can denote it with $\Sigma _{\overline{\mathcal{T}}}$.
Moreover, for any fixed $\mathsf{N}$-tuple of inhomogeneities $(\xi _{1},\ldots,\xi _{\mathsf{N}})\in 
\mathbb{C}^{\mathsf{N}}$ satisfying \eqref{cond-inh}, the spectrum of $\overline{\mathcal{T}}^{(\kappa)}(\lambda )$ in $\mathbb{\bar{D}}_{\mathsf{(6VD)},\mathsf{N}}^{(0,\mathcal{L}/\mathcal{R})}$ is simple and $\Sigma _{{\overline{\mathcal{T}}}}$ coincides with the set of functions of the form
\begin{equation}
\bar{\mathsf{t}}(\lambda )
=\sum_{a=1}^{\mathsf{N}} e^{i\mathsf{y}(\xi_a-\lambda)}\,
\frac{\theta(t_{0,\mathbf{0}}-\lambda +\xi _{a})}{\theta (t_{0,\mathbf{0}})}
\prod_{b\neq a}\frac{\theta (\lambda -\xi _{b})}{\theta (\xi _{a}-\xi _{b})}\,
\bar{\mathsf{t}}(\xi _{a}) ,
\qquad \big(\bar{\mathsf{t}}(\xi _1),\ldots,\bar{\mathsf{t}}(\xi _{\mathsf{N}})\big)\in\mathbb{C}^\mathsf{N},
\label{set-t}
\end{equation}
which satisfy the discrete system of equations
\begin{equation}
\bar{\mathsf{t}}(\xi _{a})\,\bar{\mathsf{t}}(\xi_{a}-\eta)=( -1) ^{\mathsf{x}+\mathsf{y}+\mathsf{x}\mathsf{y}}\,
\mathsc{a}(\xi _{a})\,\mathsc{d}(\xi _{a}-\eta),
\qquad
\forall a\in \{1,\ldots,\mathsf{N}\}.
\label{syst-t}
\end{equation}
The right $\overline{\mathcal{T}}^{(\kappa)}(\lambda)$-eigenstate $|\Psi_{\bar{\mathsf{t}}}^{(\kappa)}\rangle\in\mathbb{\bar{D}}_{\mathsf{(6VD)},\mathsf{N}}^{(0,\mathcal{R})}$ and the left $\overline{\mathcal{T}}^{(\kappa)}(\lambda  )$-eigenstate $\langle \Psi_{\bar{\mathsf{t}}}^{(\kappa)}|\in\mathbb{\bar{D}}_{\mathsf{(6VD)},\mathsf{N}}^{(0,\mathcal{L})}$ associated with the eigenvalue $\bar{\mathsf{t}}(\lambda )\in \Sigma_{\overline{\mathcal{T}}}$ are respectively given by
\begin{align}
& |\Psi_{\bar{\mathsf{t}}}^{( \kappa) }\rangle
=\sum_{\mathbf{h}\in\{0,1\}^\mathsf{N}}
\prod_{a=1}^{\mathsf{N}}
\left[\bigg( \kappa^{-1}e^{i\mathsf{y}\eta}\,\frac{\mathsc{a}_\mathsf{x,y}(\xi_a)}{\mathsc{d}(\xi_a-\eta)}\bigg)^{\! h_a}\,  \mathsf{q}_{\bar{\mathsf{t}},a}^{(h_a)} \right]\,
\det_{\mathsf{N}}\big[\Theta^{(0, \mathbf{h}) }\big]\, |\mathbf{h},0\rangle , 
 \label{eigenT-r}\\
&  \langle \Psi_{\bar{\mathsf{t}}}^{(\kappa)}|
  =\sum_{\mathbf{h}\in\{0,1\}^\mathsf{N}}
\prod_{a=1}^{\mathsf{N}}
\left[ \big( \kappa\, e^{i\mathsf{y}\eta}\big)^{h_a}\, \mathsf{q}_{\bar{\mathsf{t}},a}^{(h_{a})}\right]\,
\det_{\mathsf{N}}\big[\Theta^{(0, \mathbf{h})}\big] \, 
\langle 0,\mathbf{h} |, 
 \label{eigenT-l}
\end{align}
where the coefficients $\mathsf{q}_{\bar{\mathsf{t}},a}^{(h_a)}$ are (up to an overall normalization) characterized by
\begin{align}
&\frac{ \mathsf{q}_{\bar{\mathsf{t}},a}^{(1)} }{  \mathsf{q}_{\bar{\mathsf{t}},a}^{(0)}  }
=\frac{\mathsc{d}(\xi_a-\eta)}{\bar{\mathsf{t}}(\xi_a-\eta)}
=
\frac{\bar{\mathsf{t}}(\xi_a)}{\mathsc{a}_\mathsf{x,y}(\xi_a)},
\qquad
\text{with}
\quad
\mathsc{a}_\mathsf{x,y}(\lambda)\equiv (-1)^\mathsf{x+y+xy}\mathsc{a}(\lambda).
\label{t-Q-relation}
\end{align}
%
\end{theorem}

\begin{proof}
Let us first recall that, from Proposition~\ref{prop-Dr}, the action of the twisted antiperiodic transfer matrix $\overline{\mathcal{T}}^{(\kappa)}(\lambda  )$ preserves the subspace $\mathbb{\bar{D}}_{\mathsf{(6VD)},\mathsf{N}}^{(0,\mathcal{L}/\mathcal{R})}$ of $\mathbb{D}_{\mathsf{(6VD)},\mathsf{N}}^{\mathcal{L}/\mathcal{R}}$.
It follows from \eqref{YBECalDyn} that
\begin{align}
  \overline{\mathcal{T}}^{(\kappa)}(\lambda_1  )\,
  \overline{\mathcal{T}}^{(\kappa)}(\lambda_2  )
  &=\mathrm{tr}_{12} \!\left[ R_{1,2}(\lambda _{12}|-\tau -\eta \mathsf{S} +\mathsf{y}\pi\omega)^{-1}\,\overline{\mathcal{M}}_{2}^{(\kappa)}(\lambda_{2} )\,
 \overline{\mathcal{M}}_{1}^{(\kappa)}(\lambda_{1} )\,R_{1,2}(\lambda _{12}|\tau) \right]
 \notag\\
 &\hspace{-1.5cm}=\mathrm{tr}_{12} \!\left[ R_{1,2}(\lambda _{12}|-\tau -\eta \mathsf{S} +\mathsf{y}\pi\omega)^{-1}\,\overline{\mathsf{M}}_{2}^{(\kappa)}(\lambda_{2}|\tau+\eta\sigma_2^z )\,
 \overline{\mathsf{M}}_{1}^{(\kappa)}(\lambda_{1}|\tau )\,
 \mathsf{T}_{\tau }^{\sigma _{1}^{z}}\mathsf{T}_\tau^{\sigma_2^z}\,R_{1,2}(\lambda _{12}|\tau) \right]
 \notag\\
 &\hspace{-1.5cm}=\mathrm{tr}_{12} \!\left[ R_{1,2}(\lambda _{12}|\tau) R_{1,2}(\lambda _{12}|-\tau -\eta \mathsf{S} +\mathsf{y}\pi\omega)^{-1}\,\overline{\mathcal{M}}_{2}^{(\kappa)}(\lambda_{2} )\,
 \overline{\mathcal{M}}_{1}^{(\kappa)}(\lambda_{1} )\right] \! ,
  \label{comm-transfer}
\end{align}  
where we have used both the cyclic property of the trace and the zero-weight property \eqref{Comm-R12} of the $R$-matrix. Hence, \eqref{comm-transfer} acts as
\begin{equation}
 \mathrm{tr}_{12} \!\left[ \overline{\mathcal{M}}_{2}^{(\kappa)}(\lambda_{2} )\,
 \overline{\mathcal{M}}_{1}^{(\kappa)}(\lambda_{1} )\right] 
 = \overline{\mathcal{T}}^{(\kappa)}(\lambda_2  )\, \overline{\mathcal{T}}^{(\kappa)}(\lambda_1  )
\end{equation}
on the left subspace $\mathbb{\bar{D}}_{\mathsf{(6VD)},\mathsf{N}}^{(0,\mathcal{L})}$ of $\mathbb{D}_{\mathsf{(6VD)},\mathsf{N}}^{\mathcal{L}}$ associated with the eigenvalue $\mathsf{x}\pi+\mathsf{y}\pi\omega$ of $\mathsf{S}_\tau=\eta\mathsf{S}+2\tau$.
Commutativity on the right subspace  $\mathbb{\bar{D}}_{\mathsf{(6VD)},\mathsf{N}}^{(0,\mathcal{R})}$ of $\mathbb{D}_{\mathsf{(6VD)},\mathsf{N}}^{\mathcal{R}}$ follows by inserting the decomposition of the identity $\mathbb{I}_0$ \eqref{Idr}.
Hence, for any fixed $\kappa\in\mathbb{C}$, the $\kappa$-twisted transfer matrices $\overline{\mathcal{T}}^{(\kappa )}(\lambda)$ define a one-parameter family of commuting operators on $\mathbb{\bar{D}}_{\mathsf{(6VD)},\mathsf{N}}^{(0,\mathcal{L/R})}$.

It follows from the quasi-periodicity properties of the operators $\mathcal{B}(\lambda)$ and $\mathcal{C}(\lambda)$ (Lemma~\ref{lem-period}) that the restriction $\widetilde{\mathcal{T}}^{(\kappa)}(\lambda)$ of $e^{i\mathsf{y}\lambda}\,\overline{\mathcal{T}}^{(\kappa )}(\lambda)$ on $\mathbb{\bar{D}}_{\mathsf{(6VD)},\mathsf{N}}^{(0,\mathcal{L/R})}$ is a theta function of order $\mathsf{N}$ and of norm $\alpha_{\bar{\mathsf{t}}}\equiv \sum_{k=1}^\mathsf{N}\xi_k+t_{0,\mathbf{0}}$ (see Appendix~\ref{app-theta}),
which means that the action of  $\overline{\mathcal{T}}^{(\kappa )}(\lambda)$ on any vector of  $\mathbb{\bar{D}}_{\mathsf{(6VD)},\mathsf{N}}^{(0,\mathcal{L/R})}$ is completely determined by its action at $\mathsf{N}$ independent points with respect to $\alpha_{\bar{\mathsf{t}}}$ (see \eqref{interpolation}).
It also means that any eigenvalue function $\bar{\mathsf{t}}^{(\kappa)}(\lambda)$ of  $\overline{\mathcal{T}}^{(\kappa )}(\lambda)$ on $\mathbb{\bar{D}}_{\mathsf{(6VD)},\mathsf{N}}^{(0,\mathcal{L/R})}$ takes the form
\begin{equation}
\bar{\mathsf{t}}^{(\kappa )}(\lambda )
=\sum_{a=1}^{\mathsf{N}} e^{i\mathsf{y}(\xi_a-\lambda)}\,
\frac{\theta(t_{0,\mathbf{0}}-\lambda +\xi _{a})}{\theta (t_{0,\mathbf{0}})}
\prod_{b\neq a}\frac{\theta (\lambda -\xi _{b})}{\theta (\xi _{a}-\xi _{b})}\,
\bar{\mathsf{t}}^{(\kappa )}(\xi _{a}) ,
\label{t-alpha}
\end{equation}
in terms of some $\mathsf{N}$-tuple of complex numbers $\big(\bar{\mathsf{t}}^{(\kappa )}(\xi _1),\ldots,\bar{\mathsf{t}}^{(\kappa )}(\xi _\mathsf{N})\big)$.
The condition for a function $\bar{\mathsf{t}}^{(\kappa )}(\lambda )$ of the form \eqref{t-alpha} to be an eigenvalue  of $\overline{\mathcal{T}}^{(\kappa )}(\lambda)$ on $\mathbb{\bar{D}}_{\mathsf{(6VD)},\mathsf{N}}^{(0,\mathcal{L})}$ is equivalent to the fact that there exists a non-zero vector
\begin{equation}
  \langle \Psi_{\bar{\mathsf{t}}}^{( \kappa )}| = \sum_{\mathbf{h}\in\{0,1\}^{\mathsf{N}}}
  \psi _{\bar{\mathsf{t}}}^{( \kappa ) }(\mathbf{h})
  \frac{\langle 0,\mathbf{h}|}{\langle 0,\mathbf{h} | \mathbf{h},0\rangle} \
  \in \mathbb{\bar{D}}_{\mathsf{(6VD)},\mathsf{N}}^{(0,\mathcal{L})},
\end{equation}
with $\psi _{\bar{\mathsf{t}}}^{( \kappa ) }(\mathbf{h}) \equiv \langle \Psi_{\bar{\mathsf{t}}}^{( \kappa )}|\mathbf{h},0\rangle$, such that
\begin{equation}\label{t-eigen}
  \forall\mathbf{h}\in\{0,1\}^{\mathsf{N}}, \qquad
  \langle \Psi_{\bar{\mathsf{t}}}^{( \kappa )}|\, \overline{\mathcal{T}}^{(\kappa )}(\lambda)\, |\mathbf{h},0\rangle
  = \bar{\mathsf{t}}^{( \kappa) }(\lambda )\, \langle \Psi_{\bar{\mathsf{t}}}^{( \kappa )}|\mathbf{h},0\rangle .
\end{equation}
By computing the action of $\overline{\mathcal{T}}^{(\kappa )}(\lambda)$ on $|\mathbf{h},0\rangle$
at the $\mathsf{N}$ independent points $\xi_n^{(h_n)}$ by means of \eqref{act-C}, \eqref{act-B-xi_n},
one obtains that the condition \eqref{t-eigen} is equivalent to the system of equations
\begin{equation}
\bar{\mathsf{t}}^{(\kappa)}(\xi _{n}^{(h_{n})})\,
\Psi_{\bar{\mathsf{t}}}^{( \kappa) }(\mathbf{h})
=\kappa^{-1} \, (-1)^{\mathsf{x}+\mathsf{y}+\mathsf{x}\mathsf{y}}\, \mathsc{a}(\xi _{n}^{(h_{n})})\,
\Psi_{\bar{\mathsf{t}}}^{( \kappa) }(\mathsf{T}_{n}^{+}\mathbf{h})
+\kappa\, \mathsc{d}(\xi _{n}^{(h_{n})})\,\Psi_{\bar{\mathsf{t}}}^{( \kappa) }(\mathsf{T}_{n}^{-}\mathbf{h}),  \label{FarXYZSOVBax1}
\end{equation}
for any $n\in \{1,\ldots,\mathsf{N}\}$ and $\mathbf{h} \in \{0,1\}^{\mathsf{N}}$.
Taking into account the fact that
$\mathsc{a}(\xi _{n}^{(1)})=\mathsc{d}(\xi _{n}^{(0)})=0$,
one can rewrite this system of equations as the following system of homogeneous equations: 
\begin{equation}
\begin{pmatrix}
\bar{\mathsf{t}}^{(\kappa)}(\xi _{n}^{(0)}) 
& -\kappa^{-1}\,(-1)^{\mathsf{x}+\mathsf{y}+\mathsf{x}\mathsf{y}}\, \mathsc{a}(\xi _{n}^{(0)}) \\ 
-\kappa\,\mathsc{d}(\xi _{n}^{(1)})
& \bar{\mathsf{t}}^{(\kappa)}(\xi _{n}^{(1)})
\end{pmatrix}
\begin{pmatrix}
\psi^{(\kappa)}_{\bar{\mathsf{t}}}(\mathbf{h}) \\  \psi^{(\kappa)}_{\bar{\mathsf{t}}}(\mathsf{T}_n^+\mathbf{h})
\end{pmatrix} 
=
\begin{pmatrix} 0 \\ 0 \end{pmatrix} ,  \label{FarXYZhomo-system}
\end{equation}
for any $n\in \{1,\ldots,\mathsf{N}\}$ and any $\mathbf{h}\in \{0,1\}^{\mathsf{N}}$ such that $h_n=0$.
It follows that a function $\bar{\mathsf{t}}^{(\kappa)}(\lambda)$ of the form  \eqref{t-alpha} is an eigenvalue of $\overline{\mathcal{T}}^{(\kappa )}(\lambda)$  on $\mathbb{\bar{D}}_{\mathsf{(6VD)},\mathsf{N}}^{(0,\mathcal{L})}$ if and only if this system admits a non-zero solution, i.e. if and only if the
determinants of all $2\times 2$ matrices in \eqref{FarXYZhomo-system} vanish:
\begin{equation}\label{condition-talpha}
   \bar{\mathsf{t}}^{(\kappa)}(\xi _{n}^{(0)}) \, \bar{\mathsf{t}}^{(\kappa)}(\xi _{n}^{(1)}) 
   = (-1)^{\mathsf{x}+\mathsf{y}+\mathsf{x}\mathsf{y}}\, \mathsc{a}(\xi _{n}^{(0)})\, \mathsc{d}(\xi _{n}^{(1)}),
   \qquad
   \forall n\in\{1,\ldots,\mathsf{N}\}.
\end{equation}
The condition for a function $\bar{\mathsf{t}}^{(\kappa )}(\lambda )$ of the form \eqref{t-alpha} to be an eigenvalue  of $\overline{\mathcal{T}}^{(\kappa )}(\lambda)$ on $\mathbb{\bar{D}}_{\mathsf{(6VD)},\mathsf{N}}^{(0,\mathcal{R})}$ can be written similarly and one obtains the same system of equations \eqref{condition-talpha}. 
The announced isospectrality is then a trivial consequence of the fact that the conditions \eqref{t-alpha} and \eqref{condition-talpha} are the same for all the values of $\kappa \in \mathbb{C}\setminus\{0\}$.
Hence, we can omit the upper indices $( \kappa ) $ and denote the
eigenvalue functions of $\overline{\mathcal{T}}^{(\kappa)}(\lambda )$ simply by $\bar{\mathsf{t}}(\lambda )$.

Finally, it is easy to see that a given function $\bar{\mathsf{t}}(\lambda)$ of the form \eqref{t-alpha} satisfying \eqref{condition-talpha} corresponds to a unique eigenvector  $\langle\Psi_{ \bar{\mathsf{t}}}^{( \kappa )}|$ of $\overline{\mathcal{T}}^{(\kappa )}(\lambda)$, so that the spectrum of $\overline{\mathcal{T}}^{(\kappa)}(\lambda )$ in $\mathbb{\bar{D}}_{\mathsf{(6VD)},\mathsf{N}}^{(0,\mathcal{L})}$ is simple. Indeed, since
$\mathsc{a}(\xi _{n}^{(0)}), \mathsc{d}(\xi_{n}^{(1)})\neq 0$, the solution of \eqref{FarXYZhomo-system} is uniquely fixed (up to an overall normalization) by the requirement
\begin{equation}
\frac{\psi _{\bar{\mathsf{t}}}^{( \kappa ) }(\mathsf{T}_n^+\mathbf{h})}
        {\psi _{\bar{\mathsf{t}}}^{( \kappa ) }(\mathbf{h})}
 =(-1)^{\mathsf{x}+\mathsf{y}+\mathsf{x}\mathsf{y}}
     \frac{\kappa\,\bar{\mathsf{t}}(\xi _{n}^{(0)}) }{\mathsc{a}(\xi _{n}^{(0)})},
\end{equation}
for any$\,n\in \{1,\ldots,\mathsf{N}\}$ and any $\mathbf{h}\in \{0,1\}^{\mathsf{N}}$ such that $h_n=0$.
In other words, it means that this eigenvector is given by the factorized formula \eqref{eigenT-l}- \eqref{t-Q-relation}.
The proof for the right eigenstates is similar.
\end{proof}

As usual within the SOV approach, the eigenstates of the transfer matrix are obtained as {\em separate} states on the SOV basis, so that their scalar product can straightforwardly be expressed as a determinant issued from \eqref{sc-rh}:
\begin{equation}\label{sp-eigen}
    \langle \Psi_{\bar{\mathsf{t}}}^{(\kappa)}|\Psi_{\bar{\mathsf{t}}'}^{( \kappa) }\rangle
    =\det_\mathsf{N}\left[ \mathcal{F}_{\bar{\mathsf{t}},\bar{\mathsf{t}}'}\right],
\end{equation}
where $\mathcal{F}_{\bar{\mathsf{t}},\bar{\mathsf{t}}'}$ denotes the $\mathsf{N}\times\mathsf{N}$ matrix of elements
\begin{equation}\label{mat-F}
   \left[ \mathcal{F}_{\bar{\mathsf{t}},\bar{\mathsf{t}}'}\right]_{a,b}
   =
   \sum_{h=0}^{1}\left( e^{i\mathsf{y}\eta }\,
        \frac{\mathsc{a}_\mathsf{x,y}(\xi_{a})}{\mathsc{d}(\xi_{a}-\eta)}\right)^{\! h}\,
        \mathsf{q}_{\bar{\mathsf{t}},a}^{(h)}\,
        \mathsf{q}_{\bar{\mathsf{t}}^{\prime },a}^{(h)}\ 
        \vartheta _{b-1}(\xi _{a}^{(h)}-\bar{\xi}_{0}).
\end{equation}
Such a representation can be obtained for any $\bar{\mathsf{t}}(\lambda)$, $\bar{\mathsf{t}}'(\lambda)\in\Sigma_{\overline{\mathcal{T}}}$. Note that it does not depend on the value of the twist $\kappa$.

\subsection{On the reformulation of the SOV characterization of the spectrum in terms of solutions of functional $T$-$Q$ equations: homogeneous versus inhomogeneous equations}
\label{sec-hom-inhom}

It follows from the previous study that the spectrum $\Sigma _{\overline{\mathcal{T}}}$  of $\overline{\mathcal{T}}^{(\kappa)}(\lambda )$ is given as the set of all entire functions $\bar{\mathsf{t}}(\lambda)$ satisfying the quasi-periodicity properties
\begin{align}
 & \bar{\mathsf{t}}(\lambda+\pi)=(-1)^{\mathsf{N}+\mathsf{y}}\, \bar{\mathsf{t}}(\lambda),
 \label{periodt-1}\\
 &\bar{\mathsf{t}}(\lambda+\pi\omega)
 = (-e^{-2i\lambda-i\pi\omega})^\mathsf{N}\, e^{2i [\sum_{k=1}^\mathsf{N}\xi_k-\frac{\mathsf{N}}{2}\eta+\mathsf{x}\frac{\pi}{2}]} \, \bar{\mathsf{t}}(\lambda),
 \label{periodt-2}
\end{align}
and 
such that, for each $n\in\{1,\ldots,\mathsf{N}\}$, the matrix
\begin{equation}\label{mat-Dn}
   D_{\bar{\mathsf{t}},n}\equiv
   \begin{pmatrix}
    \bar{\mathsf{t}}(\xi_n^{(0)}) & -\mathsc{a}_\mathsf{x,y}(\xi_n^{(0)})\\
    -\mathsc{d}(\xi_n^{(1)}) & \bar{\mathsf{t}}(\xi_n^{(1)})
    \end{pmatrix}
\end{equation}
is of rank one, i.e. that there exists a non-zero vector
\begin{equation}\label{q-vect}
   \mathbf{q}_{\bar{\mathsf{t}},n}\equiv
   \begin{pmatrix} \mathsf{q}_{\bar{\mathsf{t}},n}^{(0)} \vspace{1mm}\\  \mathsf{q}_{\bar{\mathsf{t}},n}^{(1)}
   \end{pmatrix}
\end{equation}
which satisfies $D_{\bar{\mathsf{t}},n}\cdot \mathbf{q}_{\bar{\mathsf{t}},n}=0$ and which can be used to construct the corresponding eigenstates.
In other words, the system of quadratic  equations \eqref{syst-t} which completely characterizes, together with the functional form \eqref{set-t} of the eigenvalues, the transfer matrix spectrum $\Sigma _{\overline{\mathcal{T}}}$, is equivalent to the following condition:
\begin{multline}\label{dis-T-Q}
   \forall n\in\{1,\ldots,\mathsf{N}\}, \qquad
   \exists\ (\mathsf{q}_{\bar{\mathsf{t}},n}^{(0)},\mathsf{q}_{\bar{\mathsf{t}},n}^{(1)})\not=(0,0)
   \qquad \text{such that}\\
   \forall h_n\in\{0,1\},\qquad
   \bar{\mathsf{t}}(\xi_n^{(h_n)})\, \mathsf{q}_{\bar{\mathsf{t}},n}^{(h_n)}
   = \mathsc{a}_\mathsf{x,y}(\xi_n^{(h_n)})\, \mathsf{q}_{\bar{\mathsf{t}},n}^{(h_n+1)}
   + \mathsc{d}(\xi_n^{(h_n)})\, \mathsf{q}_{\bar{\mathsf{t}},n}^{(h_n-1)}.
\end{multline}
Hence the system of equations \eqref{dis-T-Q} corresponds to a discrete version of Baxter's famous functional $T$-$Q$ equation \cite{Bax82L}. However, in its present form, this characterization of the spectrum in terms of a discrete set of equations which strongly depends on the inhomogeneities of the system (subject to the condition \eqref{cond-inh}) does not seem very convenient for the study of physical quantities of the model, since it does not allow for an easy determination of the homogeneous and thermodynamic limits. On the contrary, Baxter's $T$-$Q$ equation in its usual (i.e. functional) form is smooth with respect to the homogeneous limit and, thanks to the equivalence with a system of Bethe equations, makes it possible to use some standard techniques to study the thermodynamic properties of the model under consideration. 
It is therefore important to be able to pass from the discrete to the continuous picture or, in other words, to find an  equivalent reformulation of the SOV discrete characterization of the transfer matrix spectrum and eigenstates in terms of some particular class of solutions on the whole complex plane $\mathbb{C}$ of a functional  $T$-$Q$ equation of Baxter's type. Note that the existence of such a reformulation has been already proven for several integrable quantum models solved by SOV \cite{Nic10a,Nic11,GroN12}, notably for the antiperiodic XXZ spin chain \cite{NicT15} which constitutes a limiting case of the present model (see Remark~\ref{rem-trig-lim}).

Hence, the problem one wants to solve can be formulated as follows: does it exist, for each $\bar{\mathsf{t}}(\lambda)\in\Sigma _{\overline{\mathcal{T}}}$, a function $Q(\lambda)$ on $\mathbb{C}$, in a class of analytic functions that has to be precisely determined, such that $\bar{\mathsf{t}}(\lambda)$ and $Q(\lambda)$ satisfy the continuous (functional) version of \eqref{dis-T-Q}:
\begin{equation}\label{hom-eq}
   \bar{\mathsf{t}}(\lambda)\, Q(\lambda )
   =\mathsc{a}_{\mathsf{x,y}}(\lambda)\,  Q(\lambda -\eta )+\mathsc{d}(\lambda)\, Q(\lambda +\eta ).
\end{equation}
If moreover, for each $n\in\{1,\ldots,\mathsf{N}\}$, $\big( Q(\xi_n), Q(\xi_n-\eta)\big)\not=(0,0)$, then this solution $Q(\lambda)$ provides, through the identification $\mathsf{q}_{\bar{\mathsf{t}},n}^{(h_n)}\equiv Q(\xi_n^{(h_n)})$, all the vectors \eqref{q-vect} satisfying the condition \eqref{dis-T-Q} and enabling ones to construct the corresponding eigenstates through \eqref{eigenT-l}-\eqref{eigenT-r}.
Provided $Q(\lambda)$ can be factorized in a generic form, the determination of all eigenvalues $\bar{\mathsf{t}}(\lambda)\in\Sigma _{\overline{\mathcal{T}}}$ can therefore be reduced to the determination of the set of roots of the corresponding solution $Q(\lambda)$, i.e. of the solution of a system of Bethe-type equations.

For any given $\bar{\mathsf{t}}(\lambda)\in\Sigma _{\overline{\mathcal{T}}}$, the equation \eqref{hom-eq} is a second-order finite-difference equation which may in principle admit two independent solutions $Q(\lambda)$. The whole problem is therefore to determine what could be the functional form of these solutions, and whether this form does or not depend on the particular $\bar{\mathsf{t}}(\lambda)\in\Sigma _{\overline{\mathcal{T}}}$ we consider (problem of the {\em completeness} of the associated system of Bethe equations). In general, this may not be an easy task, since the functional form of the $Q$-solutions to this finite-difference equation may be quite different from the functional form of its coefficients: for instance, in the present case, it is obvious that \eqref{hom-eq} cannot admit, for $\bar{\mathsf{t}}(\lambda)\in\Sigma_{\bar{\mathsf{T}}}$, any solution of the type 
\begin{equation}\label{Q-form}
  Q(\lambda)=c_Q\, e^{\alpha\lambda}\prod_{j=1}^\mathsf{M}\theta(\lambda-\lambda_j),
  \qquad c_Q\not=0,\quad
  \mathsf{M}\in\mathbb{N},\quad
  \alpha,\lambda_1,\ldots,\lambda_{\mathsf{M}}\in\mathbb{C},
\end{equation}
since the terms $\bar{\mathsf{t}}(\lambda)\, Q(\lambda )$, $\mathsc{a}(\lambda)\,  Q(\lambda -\eta )$ and $\mathsc{d}(\lambda)\, Q(\lambda +\eta )$ would all have different quasi-periodicity properties and would therefore be linearly independent.
A possible (and quite usual) way to solve the problem would be to explicitly construct the so-called {\em Q-operator} \cite{Bax72,Bax82L} (see also for instance \cite{PasG92,BazLZ97,AntF97,BazLZ99,Der99}) and to determine the functional form of its eigenvalues.
This procedure may however be quite involved and, to our knowledge, has never been performed in the case of the antiperiodic dynamical 6-vertex model.

It was recently suggested in the context of the so-called ``off-diagonal Bethe Ansatz'' \cite{CaoYSW13a,CaoCYSW14} that one may avoid these difficulties by considering, instead of \eqref{hom-eq}, a generalized functional equation.
The idea is to allow some freedom in the rewriting of the discrete equations into the continuous one, and in particular the presence of an inhomogeneous (``off-diagonal") term, so as to force the latter to admit solutions of the form \eqref{Q-form}. The equivalence of the discrete SOV characterization of the spectrum and the solutions of such generalized functional $T$-$Q$ equation is then quite simple to prove,  and it is probably the easiest way, in the context of SOV, to obtain a complete system of Bethe-type equations (see for instance \cite{KitMN14,NicT15}).
Concretely, in the present case, such a reformulation would rely  on two main ideas. The first one is that one can in fact allow, in the rewriting of \eqref{dis-T-Q} into some functional finite-difference equation, some modification of the coefficients $\mathsc{a}(\lambda)$ and $\mathsc{d}(\lambda)$ by a gauge transformation of the form
\begin{equation}\label{def-ad}
\bar{\mathsc{a}}(\lambda )\equiv f(\lambda )\,\mathsc{a}(\lambda),
\qquad
\bar{\mathsc{d}}(\lambda )\equiv\frac{\mathsc{d}(\lambda )}{f(\lambda +\eta )},
\end{equation}
which leaves the product $\mathsc{a}(\lambda)\, \mathsc{d}(\lambda-\eta)$ (and therefore the set of equations \eqref{syst-t}) unchanged.
One can therefore choose the function $f(\lambda)$ such that, for functions $Q(\lambda)$ of the form \eqref{Q-form}, the three terms $\bar{\mathsf{t}}(\lambda )\,Q(\lambda )$, $\bar{\mathsc{a}}(\lambda)\, Q(\lambda -\eta )$ and $\bar{\mathsc{d}}(\lambda)\,Q(\lambda +\eta )$ obey the same quasi-periodicity properties with respect to the periods $\pi$ and $\pi\omega$.
The second idea is that, in the rewriting of the discrete finite-difference equations into a continuous (functional) one, one can in fact allow the presence of an extra term as long as the later vanishes at each of the points $\xi_n^{(h)}$, $ n\in\{1,\ldots, \mathsf{N} \}$, $h\in\{0,1\}$, so that the discretized version of this equation effectively coincides with a gauge transformed variant of \eqref{dis-T-Q}. The resulting equation will then be of the form
\begin{equation}\label{inhom}
\bar{\mathsf{t}}(\lambda )\,Q(\lambda )
= f(\lambda )\,\mathsc{a}_{\mathsf{x,y}}(\lambda)\,Q(\lambda -\eta )
+\frac{\mathsc{d}(\lambda )}{f(\lambda +\eta )}\,  Q(\lambda +\eta )-\mathsc{a}(\lambda)\mathsc{d}(\lambda )F(\lambda).
\end{equation}
In fact, there exist many possible rewriting of this form, depending on the function $f(\lambda)$ and the degree $\mathsf{M}$ of the class of solutions $Q(\lambda)$. A consistent choice for $f(\lambda)$ and $Q(\lambda)$ is for example:
\begin{equation}
f(\lambda )\equiv f^{(\beta)}_{\mu}(\lambda)
=\beta^{-1} e^{-i\mathsf{y}\lambda }\,\frac{\theta (\lambda -\mu+(\mathsf{M}-\mathsf{N})\eta )}
          {\theta (\lambda -\mu +t_{0,\mathbf{0}})},
\qquad \qquad 
Q(\lambda )=\prod_{j=1}^{\mathsf{M}}\theta (\lambda -\lambda _j),
\label{def-f-Q}
\end{equation}
in terms of some arbitrary complex parameters $\beta$ and $\mu$ and of some roots $\lambda_1,\ldots,\lambda_\mathsf{M}$. With such a choice with fixed parameters $\beta$ and $\mu$, it is natural to expect that the reformulation of the transfer matrix spectrum in terms of the solutions $Q(\lambda)$ of degree $\mathsf{M}$ of \eqref{inhom} is complete as soon as  $\mathsf{M}\geq \mathsf{N}$ 
\footnote{\label{foot-N-1}Note that it is a priori possible to lower the degree $\mathsf{M}$ of the considered class of solutions $Q(\lambda)$ by allowing some unknown parameters in $f(\lambda)$, as far as the total number of unknowns (i.e. of Bethe roots) is still at least equal to $\mathsf{N}$. For instance, with the choice \eqref{def-f-Q}, one may consider an equation \eqref{inhom} for $\mathsf{N}$ unknowns which would be the set of the $\mathsf{M}-1$ roots of $Q(\lambda)$ {\em and} the parameter $\mu$ of $f(\lambda)$. We would then obtain  Bethe equations of slightly different form coupling these two subsets of unknowns.}.
In the appendix~\ref{app-inhom} we explicitly detail  the case $\mathsf{M}=\mathsf{N}$, whose result can be summarized as follows:
\begin{itemize}
\item[] 
if $\mathsf{M}=\mathsf{N}$, for any\footnote{\label{ft-restr-beta-mu}In fact, if $\eta\in\mathbb{R}$, they may be some additional restrictions on the set of parameters $\beta$ for which this assertion is valid (see Theorem~\ref{th-SOV-Baxter}). Moreover, we also suppose that $\mu$ is such that $\mu+(\mathsf{N}-\mathsf{M})\eta-\xi_j$, $\mu+(\mathsf{N}-\mathsf{M}-1)\eta-\xi_j$, $\mu-t_{0,\mathbf{0}}-\xi_j$, $\mu+\eta-t_{0,\mathbf{0}}-\xi_j\notin\Gamma$, $\forall j\in\{1,\ldots,\mathsf{N}\}$. } 
fixed $\mu\in\mathbb{C} $ and $\beta \in 
\mathbb{C}\setminus\{ 0\}$, then $\bar{\mathsf{t}}(\lambda )\in \Sigma _{\overline{\mathcal{T}}}$ if and only if there exists a function $Q(\lambda )$ of the form \eqref{def-f-Q},
solution with $\bar{\mathsf{t}}(\lambda )$ of the inhomogeneous functional equation \eqref{inhom} for $f(\lambda)\equiv f^{(\beta)}_{\mu}(\lambda)$ given by \eqref{def-f-Q}, such that $\big( Q(\xi_n), Q(\xi_n-\eta)\big)\not=(0,0)$ for each $n\in\{1,\ldots,\mathsf{N}\}$ (see Theorem~\ref{th-SOV-Baxter}).
\end{itemize}
The function $F(\lambda)\equiv F^{(\beta)}_{\mu,Q}(\lambda)$ in \eqref{inhom} is completely determined in terms of $f(\lambda)\equiv  f^{(\beta)}_{\mu}(\lambda)$ and $Q(\lambda)$ by imposing that the r.h.s. of \eqref{inhom} is an elliptic polynomial.
 
Such a functional equation reformulation of the transfer matrix spectrum has the clear advantage to be completely smooth with respect to the homogeneous limit in which all parameters $\xi_n$ tend to the same value (which was obviously not the case for the initial formulation). It moreover enables us to rewrite the transfer matrix eigenstates obtained by SOV in a form very similar to a Bethe vector, i.e. as some multiple action of the operators $\mathcal{D}$, evaluated at the roots of the $Q$-function (the ``Bethe" roots), on some specified pseudo-vacuum state (see Appendix~\ref{app-inhom}, Corollary~\ref{Eigen-SOV-Bethe}).
Nevertheless, the possibility to use it as an efficient tool to analyze also the thermodynamic limit is still an open question. Indeed, there is an important  price to pay for the apparent simplicity in the deriving of the functional equation: the presence of the inhomogeneous term $F(\lambda)$ makes the resulting Bethe-type equations a priori much more difficult to solve than their homogeneous analog. 

Hence, in the remaining part of this section, we turn back to the study of the homogeneous functional equation \eqref{hom-eq}. 
We shall in particular explain how one can infer the expected form of the $Q$-solutions to \eqref{hom-eq} from the quasi-periodicity properties of the coefficients, and we shall prove the completeness of these solutions for $\bar{\mathsf{t}}(\lambda)\in\Sigma_{\bar{\mathsf{T}}}$ in the case where $\mathsf{N}$ is even.

\subsection{On the reformulation of the SOV characterization of the spectrum in terms of solutions of an homogeneous functional $T$-$Q$ equation: a few preliminary considerations}

Before turning to more precise considerations about the analytic properties of the $Q$-solutions to \eqref{hom-eq}, let us precisely formulate at which conditions  the knowledge of some function $Q(\lambda)$ defines, through the equation \eqref{hom-eq}, a function $\bar{\mathsf{t}}(\lambda)$ which is an eigenvalue of the transfer matrix.

\begin{theorem}\label{th-sens1}
   Let us suppose that the inhomogeneity parameters $\xi_1,\ldots,\xi_\mathsf{N}$ of the model satisfy \eqref{cond-inh}, and let $\bar{\mathsf{t}}(\lambda)$ be an entire function such that the following conditions are satisfied:
   \begin{enumerate}
   \item\label{cond1}
   there exists a function $Q(\lambda)$ such that $\bar{\mathsf{t}}(\lambda)$ can be expressed as
\begin{equation}\label{def-t-Q}
   \bar{\mathsf{t}}(\lambda)\equiv 
   \frac{\mathsc{a}_{\mathsf{x},\mathsf{y}}(\lambda)\, Q(\lambda-\eta) 
           +\mathsc{d}(\lambda)\, Q(\lambda+\eta)}{Q(\lambda)};
\end{equation}
   \item\label{cond2}
   the function $Q(\lambda)$ of {\it \ref{cond1}.} is such that, for each $n\in\{1,\ldots,\mathsf{N}\}$, $\big( Q(\xi_n), Q(\xi_n-\eta)\big)\not=(0,0)$;
   \item\label{cond3}
   $\bar{\mathsf{t}}(\lambda)$ satisfies the quasi-periodicity properties \eqref{periodt-1} and \eqref{periodt-2}.
   \end{enumerate}
   Then $\bar{\mathsf{t}}(\lambda)$ is an eigenvalue of the $\kappa$-twisted antiperiodic transfer matrix (i.e. $\bar{\mathsf{t}}(\lambda)\in\Sigma_{\overline{\mathcal{T}}}$), with corresponding eigenstates given in terms of the function $Q(\lambda)$ of {\it \ref{cond1}.} as
\begin{align}
& |\Psi_{\bar{\mathsf{t}}}^{( \kappa) }\rangle
=\sum_{\mathbf{h}\in\{0,1\}^\mathsf{N}}
\prod_{a=1}^{\mathsf{N}}\left[   \left( \kappa^{-1} e^{i\mathsf{y}\eta}\,\frac{\mathsc{a}_\mathsf{x,y}(\xi_a)}{\mathsc{d}(\xi_a-\eta)}\right)^{\! h_a} Q(\xi_a^{(h_{a})})\right]\,
\det_{\mathsf{N}}\big[\Theta^{(0, \mathbf{h}) }\big]\, |\mathbf{h},0\rangle , 
 \label{eigenT-r-Q}\\
&  \langle \Psi_{\bar{\mathsf{t}}}^{(\kappa)}|
  =\sum_{\mathbf{h}\in\{0,1\}^\mathsf{N}}
\prod_{a=1}^{\mathsf{N}}\left[ \big(\kappa\, e^{i\mathsf{y}\eta}\big)^{ h_a}\, Q(\xi_a^{(h_{a})})\right]\,
\det_{\mathsf{N}}\big[\Theta^{(0, \mathbf{h})}\big] \, 
\langle 0,\mathbf{h} |.
 \label{eigenT-l-Q}
\end{align}
\end{theorem}

As discussed in the previous subsection, this theorem is a direct corollary of the SOV characterization of the antiperiodic transfer matrix spectrum and eigenstates of Theorem~\ref{thm-eigen-t}, which can easily be proven by particularizing the relation \eqref{def-t-Q} at the $2\mathsf{N}$ points $\xi_n,\ \xi_n-\eta$, $1\le n\le \mathsf{N}$, so as to recover \eqref{dis-T-Q}. The whole problem is therefore
\begin{enumerate}
   \item[(a)] to characterize the functional form of the function $Q(\lambda)$, preferably in a completely factorized form, so as to rewrite the entireness condition for $\bar{\mathsf{t}}(\lambda)$ as a system of Bethe equations for the roots of $Q(\lambda)$;
   \item[(b)] to prove that, for each $\bar{\mathsf{t}}(\lambda)\in \Sigma _{\overline{\mathcal{T}}}$, there indeed exists such a $Q$-solution to the corresponding $T$-$Q$ equation, i.e. to prove the completeness of the aforementioned characterization in terms of solutions to Bethe equations.
\end{enumerate}
Before turning to these more delicate points that we shall partially solve in the next subsections, we would like to make a few remarks about the formulation of Theorem~\ref{th-sens1}, which  can in fact be rewritten in slightly different equivalent forms. For instance, it is easy to see that the condition {\it \ref{cond3}.} can be replaced by some equivalent conditions on the function $Q(\lambda)$:

\begin{proposition}\label{lem-per-t-W}
For some given function $Q(\lambda)$, we define a function $\bar{\mathsf{t}}(\lambda)$ by the relation \eqref{def-t-Q}. Then
\begin{enumerate}
    \item[(i)] the function $\bar{\mathsf{t}}(\lambda)$ defined by \eqref{def-t-Q} satisfies the quasi-periodicity property 
    \eqref{periodt-1} if and only if the function
    \begin{equation}\label{W-1}
       W_Q^{(1)}(\lambda)\equiv Q(\lambda+\pi)\, Q(\lambda-\eta)-(-1)^\mathsf{y}\, Q(\lambda+\pi-\eta)\, Q(\lambda)
    \end{equation}
    satisfies the relation
    \begin{equation}\label{rel-W-1}
       \mathsc{d}(\lambda)\, W_Q^{(1)}(\lambda+\eta)
       =(-1)^\mathsf{x+xy}\, \mathsc{a}(\lambda)\, W_Q^{(1)}(\lambda);
    \end{equation}
    \item[(ii)] the function $\bar{\mathsf{t}}(\lambda)$ defined by \eqref{def-t-Q} satisfies the quasi-periodicity property 
    \eqref{periodt-2} if and only if the function
    \begin{equation}\label{W-2}
       W_Q^{(2)}(\lambda)\equiv Q(\lambda+\pi\omega)\, Q(\lambda-\eta)-(-1)^\mathsf{x}\, e^{-i\mathsf{N}\eta}\, Q(\lambda+\pi\omega-\eta)\, Q(\lambda)
    \end{equation}
    satisfies the relation
    \begin{equation}\label{rel-W-2}
       \mathsc{d}(\lambda)\, W_Q^{(2)}(\lambda+\eta)
       =(-1)^\mathsf{y+xy}\, e^{-i\mathsf{N}\eta}\, \mathsc{a}(\lambda)\, W_Q^{(2)}(\lambda).
    \end{equation}
\end{enumerate}
Hence the condition {\it \ref{cond3}.} in Theorem~\ref{th-sens1} can be replaced by the condition:
\begin{enumerate}
 \item[3'.] the function $Q(\lambda)$ of {\it \ref{cond1}.}  satisfies the relations \eqref{rel-W-1} and \eqref{rel-W-2}.
\end{enumerate}
\end{proposition}

Another remark comes from the fact that, if $Q(\lambda)$ satisfies \eqref{hom-eq} for some function $\bar{\mathsf{t}}(\lambda)$ satisfying \eqref{periodt-1} and \eqref{periodt-2}, then $Q(\lambda+\pi)$ and $Q(\lambda+\pi\omega)$ satisfy respectively the following equations:
\begin{align*}
&\bar{\mathsf{t}}(\lambda)\, Q(\lambda+\pi )
   =(-1)^\mathsf{y}\, \mathsc{a}_{\mathsf{x},\mathsf{y}}(\lambda)\,  Q(\lambda+\pi -\eta )
   +(-1)^\mathsf{y}\, \mathsc{d}(\lambda)\, Q(\lambda+\pi +\eta ),
   \\
 &\bar{\mathsf{t}}(\lambda)\, Q(\lambda+\pi\omega )
   =(-1)^\mathsf{x}\, e^{-i\mathsf{N}\eta}\,\mathsc{a}_{\mathsf{x},\mathsf{y}}(\lambda)\,  Q(\lambda+\pi \omega-\eta )
   +(-1)^\mathsf{x}\,  e^{i\mathsf{N}\eta}\,\mathsc{d}(\lambda)\, Q(\lambda+\pi\omega +\eta ),
\end{align*}
which means that we can in fact slightly relax the admissibility condition {\it \ref{cond2}.} according to the following proposition:

\begin{proposition}\label{prop-Qinh}
If $Q(\lambda)$ is a solution to \eqref{hom-eq}  for some function $\bar{\mathsf{t}}(\lambda)$ satisfying \eqref{periodt-1} and \eqref{periodt-2}, then $e^{\pm i\pi \mathsf{y}\frac{\lambda}{\eta}}\, Q(\lambda+\pi)$, $e^{ i(\mathsf{N}\pm\frac{\pi \mathsf{x}}{\eta})\lambda}\, Q(\lambda+\pi\omega)$ and $e^{ i(\mathsf{N}+\epsilon_1\frac{\pi \mathsf{x}}{\eta}+\epsilon_2\frac{\pi\mathsf{y}}{\eta})\lambda}\, Q(\lambda+\pi+\pi\omega)$  (with $\epsilon_1,\epsilon_2=\pm 1$) are also solutions of  \eqref{hom-eq} for the same function $\bar{\mathsf{t}}(\lambda)$.
Hence the condition {\it \ref{cond2}.} in Theorem~\ref{th-sens1} can be replaced by the condition:
\begin{enumerate}
 \item[2'.] the function $Q(\lambda)$ of {\it \ref{cond1}.}  is such that,  for each $n\in\{1,\ldots,\mathsf{N}\}$, there exists $(\alpha_n,\beta_n)\in\{0,1\}^2$ such that $\big( Q(\xi_n+\alpha_n\pi+\beta_n\pi\omega), Q(\xi_n+\alpha_n\pi+\beta_n\pi\omega-\eta)\big)\not=(0,0)$,
\end{enumerate}
with corresponding eigenstates given by 
\begin{align}
& |\Psi_{\bar{\mathsf{t}}}^{( \kappa) }\rangle
=\sum_{\mathbf{h}\in\{0,1\}^\mathsf{N}}
\prod_{a=1}^{\mathsf{N}}\left[ 
\left( \frac{e^{i\mathsf{y}\eta}\mathsc{a}_\mathsf{x,y}(\xi_a)}{\kappa\,\mathsc{d}(\xi_a-\eta)}\right)^{\! h_a} 
e^{i [ \alpha_a\frac{\pi\mathsf{y}}{\eta}+\beta_a(\mathsf{N}+\frac{\pi\mathsf{x}}{\eta}) ]\xi_a^{(h_a)}}Q(\xi_a^{(h_{a})}+\alpha_a\pi+\beta_a\pi\omega)\right]\, 
 \nonumber\\ 
 &\hspace{11cm}\times
\det_{\mathsf{N}}\big[\Theta^{(0, \mathbf{h}) }\big]\, |\mathbf{h},0\rangle , 
 \label{eigenT-r-Qgen}\\
&  \langle \Psi_{\bar{\mathsf{t}}}^{(\kappa)}|
  =\sum_{\mathbf{h}\in\{0,1\}^\mathsf{N}}
\prod_{a=1}^{\mathsf{N}}\left[ \big(\kappa\, e^{i\mathsf{y}\eta }\big)^{h_a}\, e^{i [ \alpha_a\frac{\pi\mathsf{y}}{\eta}+\beta_a(\mathsf{N}+\frac{\pi\mathsf{x}}{\eta}) ]\xi_a^{(h_a)}}Q(\xi_a^{(h_{a})}+\alpha_a\pi+\beta_a\pi\omega)\right]\,
\det_{\mathsf{N}}\big[\Theta^{(0, \mathbf{h})}\big] \, 
\langle 0,\mathbf{h} |,
 \label{eigenT-l-Qgen}
\end{align}
instead of \eqref{eigenT-r-Q}, \eqref{eigenT-l-Q}.
\end{proposition}

Note finally that the relation \eqref{rel-W-1} (respectively the relation \eqref{rel-W-2}) can be understood as some ``wronskian-type'' identity for the two solutions $Q(\lambda)$ and $e^{\pm i\pi \mathsf{y}\frac{\lambda}{\eta}}\, Q(\lambda+\pi)$ (respectively $Q(\lambda)$ and $e^{ i(\mathsf{N}\pm\frac{\pi \mathsf{x}}{\eta})\lambda}\, Q(\lambda+\pi\omega)$) of the functional $T$-$Q$ equation \eqref{hom-eq}, as a particular case of the following very general property:

\begin{proposition}\label{lem-W12}
Let $Q_1(\lambda)$ and $Q_2(\lambda)$ be two solutions of the functional equation \eqref{hom-eq} for some given function $\bar{\mathsf{t}}(\lambda)$.
Then their quantum Wronskian,
\begin{equation}\label{W12}
  W_{12}(\lambda)=Q_1(\lambda-\eta)\, Q_2(\lambda)- Q_1(\lambda)\, Q_2(\lambda-\eta),
\end{equation}
satisfies the relation
\begin{align}
  \mathsc{d}(\lambda)\, W_{12}(\lambda+\eta)
  &=(-1)^{\mathsf{x}+\mathsf{y}+\mathsf{x}\mathsf{y}}\, \mathsc{a}(\lambda)\, W_{12}(\lambda).
  \label{rel-W1}
\end{align}
Hence
\begin{equation}\label{rel-W2}
   W_{12}(\lambda)=f_{12}(\lambda)\,\mathsc{d}(\lambda),
   \qquad \text{with}\quad 
   f_{12}(\lambda+\eta)=(-1)^{\mathsf{x}+\mathsf{y}+\mathsf{xy}} f_{12}(\lambda).
\end{equation}
\end{proposition}

\begin{rem}\label{rem-shift-eta}
If $Q(\lambda)$ satisfies the homogeneous functional equation \eqref{hom-eq} then, for each $k\in\mathbb{Z}$, $e^{2\pi i k\frac{\lambda}{\eta}}\, Q(\lambda)$ also provides a solution to the functional  \eqref{hom-eq}, which is however not independent from $Q(\lambda)$ in the sense that their quantum Wronskian \eqref{W12} is identically zero.
\end{rem}

\subsection{Study of the homogeneous $T$-$Q$ functional equation: an Ansatz for the $Q$-solutions}

As just announced, the periodicity properties of the coefficients of \eqref{hom-eq} for $\bar{\mathsf{t}}(\lambda)\in\Sigma_{\overline{\mathcal{T}}}$ enable one to make a reasonable guess about the functional  form of its possible entire $Q$-solutions, i.e. the solutions that are susceptible to lead to a system of Bethe equations.
The idea is to transform this equation, similarly as in  \cite{KriLWZ97}, into a difference equation with elliptic coefficients: 
\begin{equation}\label{eq-hom-2}
   \widetilde{Q}(\lambda+\eta)+f_1(\lambda)\,\frac{P(\lambda)}{P(\lambda+\eta)}\, \widetilde{Q}(\lambda) + f_2(\lambda)\, \frac{P(\lambda-\eta)}{P(\lambda+\eta)}\, \widetilde{Q}(\lambda-\eta)=0.
\end{equation}
Here we have set
\begin{equation}
  f_1(\lambda)=-\frac{\bar{\mathsf{t}}(\lambda)}{\mathsc{d}(\lambda)},
  \qquad
  f_2(\lambda)=\frac{\mathsc{a}_{\mathsf{x,y}}(\lambda)}{\mathsc{d}(\lambda)},
  \qquad\text{and}\quad
  \widetilde{Q}(\lambda)=\frac{Q(\lambda)}{P(\lambda)},
\end{equation}
where $P(\lambda)$ has to be chosen appropriately so as to ensure that the coefficients of \eqref{eq-hom-2} are doubly-periodic.
Since the functions $f_i(\lambda)$ exhibit different quasi-periodicity properties according to the values of $\mathsf{x}$ and $\mathsf{y}$, we shall now consider the three different cases $\mathsf{x=0}$, $\mathsf{y=0}$ and $\mathsf{x=y}$ separately.
To this aim, it is convenient to introduce shorthand notations for different variants of the theta function that we shall use within the study of these different cases. We therefore define the following odd functions of $\lambda$,
\begin{align}
  &{\theta}_\mathsf{x=0}(\lambda)\equiv\theta_1\Big(\frac{\lambda}{2}\,\Big|\,\frac{\omega}{2}\Big),
  \label{theta-x0}\\
  &\theta_\mathsf{y=0}(\lambda)\equiv\theta_1(\lambda\,|\,2\omega),
  \label{theta-y0}\\
  &\theta_\mathsf{x=y}(\lambda)\equiv e^{i\frac{\lambda}{2}}\,\theta_1\Big(\frac{\lambda}{2}\,\Big|\,\omega\Big)\, \theta_1\Big(\frac{\lambda+\pi+\pi\omega}{2}\,\Big|\,\omega\Big),
  \label{theta-xy}
\end{align}
which satisfy respectively the quasi-periodicity properties:
\begin{xalignat}{2}
   &\theta_\mathsf{x=0}(\lambda+2\pi)=-\theta_\mathsf{x=0}(\lambda), &
   &\theta_\mathsf{x=0}(\lambda+\pi\omega)=-e^{-i\lambda-i\frac{\pi\omega}{2}}\,\theta_\mathsf{x=0}(\lambda),
   \\
   &\theta_\mathsf{y=0}(\lambda+\pi)=-\theta_\mathsf{y=0}(\lambda), &
   &\theta_\mathsf{y=0}(\lambda+2\pi\omega)=-e^{-2i\lambda-2i\pi\omega}\,\theta_\mathsf{y=0}(\lambda),
   \displaybreak[0]\\
   &\theta_\mathsf{x=y}(\lambda+2\pi)=-\theta_\mathsf{x=y}(\lambda), &
   &\theta_\mathsf{x=y}(\lambda+2\pi\omega)=-e^{-2i\lambda-2i\pi\omega}\,\theta_\mathsf{x=y}(\lambda),
   \\
   &\theta_\mathsf{x=y}(\lambda\pm\pi+\pi\omega)=\pm i e^{-i\lambda-i\frac{\pi\omega}{2}}\,\theta_\mathsf{x=y}(\lambda).
\end{xalignat}

\paragraph{The case $\mathsf{x}=0$.}  Let us choose $P(\lambda)$ in the form
  \begin{equation}\label{P-form-x0}
     P(\lambda)=\prod_{j=1}^\mathsf{N}\theta_{\mathsf{x=0}}(\lambda-z_j),
  \end{equation}
  for arbitrary $z_j$, $1\le j\le \mathsf{N}$.
  Then the coefficients $f_1(\lambda)\,\frac{P(\lambda)}{P(\lambda+\eta)}$ and $f_2(\lambda)\, \frac{P(\lambda-\eta)}{P(\lambda+\eta)}$ for the equation \eqref{eq-hom-2} for $\widetilde{Q}(\lambda)$ are elliptic functions of periods $2\pi$ and $\pi\omega$.
Following \cite{KriLWZ97}, we may therefore 
look for the  double-Bloch solutions  of \eqref{eq-hom-2}, i.e. for the
meromorphic solutions  $\widetilde{Q}(\lambda)=\frac{Q(\lambda)}{P(\lambda)}$ such that
  \begin{equation}\label{double-bloch}
     \widetilde{Q}(\lambda+2\pi)=B_1\, \widetilde{Q}(\lambda),
     \qquad\text{and}\quad
     \widetilde{Q}(\lambda+\pi\omega)=B_2\, \widetilde{Q}(\lambda),
  \end{equation}
  for some Bloch multipliers $B_1, B_2$.
  It is easy to see \cite{KriWZ98} that any double-Bloch function of the form \eqref{double-bloch}  can be written in the form
  \begin{equation}\label{Qtilde}
     \widetilde{Q}(\lambda)=c_Q\, e^{\alpha \lambda}
     \prod_{j=1}^\mathsf{M}\frac{\theta_{\mathsf{x=0}}(\lambda-\lambda_j)}
                                                     {\theta_{\mathsf{x=0}}(\lambda-\mu_j)}
  \end{equation}
  for some integer $\mathsf{M}$, some complex parameter $\alpha$ and some sets of zeroes and poles $\lambda_j$ and $\mu_j$, $1\le j \le \mathsf{M}$.
  If $\widetilde{Q}(\lambda)$ of the form \eqref{Qtilde} is solution of \eqref{eq-hom-2}, then $Q(\lambda)=P(\lambda)\,\widetilde{Q}(\lambda)$ is solution of \eqref{hom-eq}, and we expect the latter to be entire\footnote{It is easy to see that, if $Q(\lambda)$ solution to \eqref{hom-eq} admits some pole $\mu_j$, then it should have an infinite number of poles of the form $\mu_j+k\eta$ for an infinite number of $k\in\mathbb{Z}$, which is clearly not compatible with the form \eqref{Qtilde} of $\tilde{Q}(\lambda)$ (we suppose here that $\eta\notin\pi\mathbb{Q}+\pi\omega\mathbb{Q}$, and that the inhomogeneity parameters satisfy \eqref{cond-inh}).}, which means that it can be written in the form
  %
  \begin{equation}\label{formQ-x0}
     Q(\lambda)=c_Q\, e^{\alpha \lambda} \prod_{j=1}^\mathsf{N}\theta_{\mathsf{x=0}}(\lambda-\lambda_j),
  \end{equation}
  for some complex parameter $\alpha$ and some set of roots $\lambda_j$, $1\le j \le \mathsf{N}$.
  
Note that, in the obtention of the Ansatz \eqref{formQ-x0} for $Q(\lambda)$, we have only used a weaker periodicity property than \eqref{periodt-1} for $\bar{\mathsf{t}}(\lambda)$. Imposing that $\bar{\mathsf{t}}(\lambda)$ strictly satisfies \eqref{periodt-1} is equivalent\footnote{Indeed, if $Q(\lambda)$ of the form \eqref{formQ-x0} is a solution to \eqref{hom-eq}, $e^{i\frac{\pi\mathsf{y}}{\eta}\lambda}Q(\lambda+\pi)$ is another solution to \eqref{hom-eq} that we expect here to be independent from the first one.} to the additional constraint \eqref{rel-W-1} for the corresponding function $W_Q^{(1)}(\lambda)$ defined in term of $Q(\lambda)$ as in \eqref{W-1}. It means that
  \begin{equation}\label{W1-d}
    W_Q^{(1)}(\lambda)=g^{(1)}(\lambda)\, \mathsc{d}(\lambda),
  \end{equation}
  where $g^{(1)}(\lambda)$ is $\eta$-periodic: $g^{(1)}(\lambda+\eta)=g^{(1)}(\lambda)$.
  On the other hand, for $Q(\lambda)$  of the form \eqref{formQ-x0}, it is easy to see that $W_Q^{(1)}(\lambda)$ satisfies the quasi-periodicity properties
  \begin{align}
     &W_Q^{(1)}(\lambda+\pi)=(-1)^{\mathsf{N}+\mathsf{y}+1}\, e^{2\pi\alpha}\, W_Q^{(1)}(\lambda),\\
     &W_Q^{(1)}(\lambda+\pi\omega)= (-e^{-2i\lambda-i\pi\omega})^\mathsf{N}\,  e^{2i\sum\lambda_j+i\mathsf{N}\eta+2\pi\omega\alpha}\,W_Q^{(1)}(\lambda),
  \end{align}
  which have to be compared to the quasi-periodicity properties of $\mathsc{d}(\lambda)$. Hence $g^{(1)}(\lambda)$ also satisfies the quasi-periodicity properties:
  \begin{align}
     &g^{(1)}(\lambda+\pi)
     = (-1)^{\mathsf{y}+1}\, e^{2\pi\alpha}\, g^{(1)}(\lambda),\\
     &g^{(1)}(\lambda+\pi\omega)
     = e^{2i\big[\sum\lambda_j+\frac{\mathsf{N}}{2}\eta-\sum\xi_k\big]+2\pi\omega\alpha}\,g^{(1)}(\lambda).
  \end{align}
  It means that there exist two constants $c_W^{(1)}$ and $\alpha_1$ such that 
  \begin{equation}
    g^{(1)}(\lambda)=c_W^{(1)}\, e^{\alpha_1\lambda},
  \end{equation}
  and that there exist $k_1,k_2,k_3\in\mathbb{Z}$ such that
  \begin{equation}
    \left\{
    \begin{aligned}
       & \alpha_1\eta = 2\pi i k_1,\\
       & \alpha_1\pi = i\pi (1-\mathsf{y})+2\pi\alpha +2\pi i k_2,\\
       & \alpha_1 \pi\omega =2i\Big[\sum\lambda_j+\frac{\mathsf{N}}{2}\eta-\sum\xi_k\Big] +2\pi\omega\alpha+2\pi i k_3.
     \end{aligned}
    \right.
  \end{equation}
  Hence, if $c_W^{(1)}\not=0$, one obtains the following sum rules:
  \begin{equation}\label{sum-rule1}
    \left\{
     \begin{aligned}
     &\alpha= i\left[ \frac{\mathsf{y}-1}{2}+k_1\frac{\pi}{\eta}-k_2\right],\\
     &\sum_{j=1}^\mathsf{N}\lambda_j
     =\sum_{k=1}^\mathsf{N}\xi_k-\frac{\mathsf{N}}{2}\eta+(1-\mathsf{y})\frac{\pi\omega}{2}+k_2 \pi\omega-k_3\pi,
     \end{aligned}
    \right.
  \end{equation}
  for some $k_1,k_2,k_3\in\mathbb{Z}$.
  
  Note that, if there indeed exists such a solution, it means that there exists such a solution for which $k_2=0$ and $k_1,k_3\in\{0,1\}$ (this is due to Remark~\ref{rem-shift-eta}, and to the quasi-periodicity properties of the function $\theta_\mathsf{x=0}$ with respect to a shift of one the roots $\lambda_j$ by $2\pi$ or $\pi\omega$).
  In fact, one can even be more precise, and formulate the following Ansatz:
  \begin{itemize}
  \item if $\mathsf{(x,y)=(0,1)}$, we expect two independent solutions of the form
  \begin{equation}\label{Q-01}
      Q(\lambda)=\prod_{j=1}^\mathsf{N}\theta_\mathsf{x=0}(\lambda-\lambda_j),
  \end{equation}
  with
  \begin{equation}\label{sum-01}
      \sum_{j=1}^\mathsf{N}\lambda_j=\sum_{k=1}^\mathsf{N}\xi_k-\frac{\mathsf{N}}{2}\eta+k\pi,\qquad
      k\in\{0,1\},
  \end{equation}
  and
  \begin{equation}\label{Qbis-01}
      \widehat{Q}(\lambda)=e^{i\frac{\pi}{\eta}\lambda}\, Q(\lambda+\pi);
  \end{equation}
  \item if $\mathsf{(x,y)=(0,0)}$ and $\mathsf{N}$ odd, we expect two independent solutions of the form
  \begin{equation}\label{Q-00-1}
      Q(\lambda)=e^{i\big[\frac{k\pi}{\eta}-\frac{1}{2}\big]\lambda}\,\prod_{j=1}^\mathsf{N}\theta_\mathsf{x=0}(\lambda-\lambda_j),
      \qquad k\in\{0,1\},
  \end{equation}
  with
  \begin{equation}\label{sum-00-1}
      \sum_{j=1}^\mathsf{N}\lambda_j=\sum_{k=1}^\mathsf{N}\xi_k-\frac{\mathsf{N}}{2}\eta+\frac{\pi\omega}{2},
  \end{equation}
  and
  \begin{equation}\label{Qbis-00-1}
      \widehat{Q}(\lambda)=Q(\lambda+\pi).
  \end{equation}
  \end{itemize}

\paragraph{The case $\mathsf{y}=0$.}
A similar reasoning can be made in this case, by choosing $P(\lambda)$ of the form
\begin{equation}\label{P-form-y0}
     P(\lambda)=\prod_{j=1}^\mathsf{N}\theta_\mathsf{y=0}(\lambda-z_j),
\end{equation}
with arbitrary $z_j$, $1\le j\le \mathsf{N}$, so as to obtain elliptic coefficients with periods $\pi$ and $2\pi\omega$ in \eqref{eq-hom-2}. It leads to the following Ansatz for $Q(\lambda)$:
\begin{equation}\label{formQ-y0}
     Q(\lambda)=c_Q\, e^{\alpha \lambda} \prod_{j=1}^{\mathsf{N}}\theta_\mathsf{y=0} (\lambda-\lambda_j  ),
\end{equation}
for some complex parameter $\alpha$ and some set of roots $\lambda_j$, $1\le j \le \mathsf{N}$. The condition \eqref{rel-W-2}, which is a necessary and sufficient condition for the corresponding function \eqref{def-t-Q} to satisfy \eqref{periodt-2}, results into the following sum rules for \eqref{formQ-y0}:
\begin{equation}\label{sum-rule2}
    \left\{
     \begin{aligned}
     &\alpha= i\left[ k_1\frac{\pi}{\eta}-k_2\right],\\
     &\sum_{j=1}^\mathsf{N}\lambda_j
     =\sum_{k=1}^\mathsf{N}\xi_k-\frac{\mathsf{N}}{2}\eta+(1-\mathsf{x})\frac{\pi}{2}+k_2 \pi\omega-k_3\pi,
     \end{aligned}
    \right.
\end{equation}
for some $k_1,k_2,k_3\in\mathbb{Z}$.

As previously, we can be more precise. Indeed, if such a solution exists, then it is easy to see, from Remark~\ref{rem-shift-eta} and considerations about the quasi-periodicity properties of the function $\theta_\mathsf{y=0}$ with respect to a shift of one of the roots $\lambda_j$ by $\pi$ or $2\pi\omega$, that there also exists a solution for which $k_3=0$ and $k_1,k_2\in\{0,1\}$. Then
one can formulate the following Ansatz:
  \begin{itemize}
  \item if $\mathsf{(x,y)=(1,0)}$, we expect two independent solutions of the form
  \begin{equation}\label{Q-10}
      Q(\lambda)=e^{-i k \lambda}\prod_{j=1}^\mathsf{N}\theta_\mathsf{y=0}(\lambda-\lambda_j),
      \qquad k\in\{0,1\},
  \end{equation}
  with
  \begin{equation}\label{sum-10}
      \sum_{j=1}^\mathsf{N}\lambda_j=\sum_{k=1}^\mathsf{N}\xi_k-\frac{\mathsf{N}}{2}\eta+k\pi\omega,
  \end{equation}
  and
  \begin{equation}\label{Qhat-10}
      \widehat{Q}(\lambda)=e^{i(\mathsf{N}+\frac{\pi}{\eta})\lambda}\, Q(\lambda+\pi\omega);
  \end{equation}

  \item if $\mathsf{(x,y)=(0,0)}$ and $\mathsf{N}$ odd, we expect two independent solutions of the form
  \begin{equation}\label{Q-00-2}
      Q(\lambda)=e^{i\frac{k\pi}{\eta}\lambda}\,\prod_{j=1}^\mathsf{N}\theta_\mathsf{y=0}(\lambda-\lambda_j),
      \qquad k\in\{0,1\},
  \end{equation}
  with
  \begin{equation}\label{sum-00-2}
      \sum_{j=1}^\mathsf{N}\lambda_j=\sum_{k=1}^\mathsf{N}\xi_k-\frac{\mathsf{N}}{2}\eta+\frac{\pi}{2},
  \end{equation}
  and
  \begin{equation}\label{Qbis-00-2}
      \widehat{Q}(\lambda)=e^{i\mathsf{N}\lambda}Q(\lambda+\pi\omega).
  \end{equation}
  \end{itemize}

\paragraph{The case $\mathsf{x=y}$.}
In that case, the choice
\begin{equation}\label{P-form-gen}
     P(\lambda)=\prod_{j=1}^{\mathsf{N}}\theta_\mathsf{x=y}(\lambda-z_j),
\end{equation}
with arbitrary $z_j$, $1\le j\le \mathsf{N}$,
leads to an equation \eqref{eq-hom-2} which has elliptic coefficients with periods $\pi\omega+\pi$ and $\pi\omega-\pi$, and to the following Ansatz for the solution $Q(\lambda)$ of \eqref{hom-eq}:
\begin{equation}\label{formQx=y}
     Q(\lambda)=c_Q\, e^{\alpha \lambda} \prod_{j=1}^{\mathsf{N}}\theta_\mathsf{x=y} (\lambda-\lambda_j  ),
\end{equation}
for some complex parameter $\alpha$ and some set of roots $\lambda_j$, $1\le j \le \mathsf{N}$. 
It is then easy to see that the conditions \eqref{rel-W-1} and \eqref{rel-W-2} result into the same\footnote{In fact, for $Q(\lambda)$ of the form \eqref{formQx=y}, one has 
$W_Q^{(2)}(\lambda+\pi)=(-1)^{\mathsf{x}+1}\,(ie^{-i\lambda-i\frac{\pi\omega}{2}})^\mathsf{N}\,
e^{i\sum\lambda_j+(\pi+\pi\omega)\alpha} \, W_Q^{(1)}(\lambda)$.} sum rules for the parameters entering \eqref{formQx=y}, which can be written as
\begin{equation}\label{rule-x=y}
    \left\{
     \begin{aligned}
     &\alpha=  i\left[ \frac{\mathsf{y}-1}{2}+k_1\frac{\pi}{\eta}-k_2\right],\\
     &\sum_{j=1}^{\mathsf{N}}\lambda_j
     =\sum_{k=1}^\mathsf{N}\xi_k-\frac{\mathsf{N}}{2} \eta+(\pi+ \pi\omega)\left(\frac{1-\mathsf{y}}{2}+k_2\right)+k_3\pi,
     \end{aligned}
    \right.
\end{equation}
for $k_1,k_2,k_3\in\mathbb{Z}$.

Here again one can be more precise: if such a solution exists, it is easy to see that, due to Remark~\ref{rem-shift-eta} and considerations about the quasi-periodicity properties of the function $\theta_\mathsf{x=y}$ with respect to a shift of one of the roots $\lambda_j$ by $\pi+\pi\omega$ or $2\pi$, it exists for $k_2=0$ and $k_1,k_3\in\{0,1\}$. Due to the existence of a second solution of the form
\begin{equation}
e^{i\frac{\mathsf{y}\pi}{\eta}\lambda}\, Q(\lambda+\pi)
                                           \propto e^{i(\mathsf{N}+\frac{\mathsf{x}\pi}{\eta})\lambda}\, Q(\lambda+\pi\omega),
\end{equation}
one can finally formulate the following Ansatz:
  \begin{itemize}
  \item if $\mathsf{(x,y)=(1,1)}$, we expect two independent solutions of the form
  \begin{equation}\label{Q-11}
      Q(\lambda)=\prod_{j=1}^\mathsf{N}\theta_\mathsf{x=y}(\lambda-\lambda_j),
  \end{equation}
  with
  \begin{equation}\label{sum-11}
      \sum_{j=1}^\mathsf{N}\lambda_j=\sum_{k=1}^\mathsf{N}\xi_k-\frac{\mathsf{N}}{2}\eta+k\pi,\qquad
      k\in\{0,1\},
  \end{equation}
  and
  \begin{equation}\label{Qbis-11}
      \widehat{Q}(\lambda)=e^{i\frac{\pi}{\eta}\lambda}\, Q(\lambda+\pi)
                                           \propto e^{i(\mathsf{N}+\frac{\pi}{\eta})\lambda}\, Q(\lambda+\pi\omega);
  \end{equation}
  \item if $\mathsf{(x,y)=(0,0)}$ and $\mathsf{N}$ odd, we expect two independent solutions of the form
  \begin{equation}\label{Q-00-3}
      Q(\lambda)=e^{i\big[\frac{k\pi}{\eta}-\frac{1}{2}\big]\lambda}\,\prod_{j=1}^\mathsf{N}\theta_\mathsf{x=y}(\lambda-\lambda_j),
      \qquad k\in\{0,1\},
  \end{equation}
  with
  \begin{equation}\label{sum-00-3}
      \sum_{j=1}^\mathsf{N}\lambda_j=\sum_{k=1}^\mathsf{N}\xi_k-\frac{\mathsf{N}}{2}\eta+\frac{\pi+\pi\omega}{2},
  \end{equation}
  and
  \begin{equation}\label{Qbis-00-3}
      \widehat{Q}(\lambda)=Q(\lambda+\pi)\propto e^{i\mathsf{N}\lambda}\, Q(\lambda+\pi\omega).
  \end{equation}
  \end{itemize}

\begin{rem}
 The case $\mathsf{(x,y)}=(0,0)$ with $\mathsf{N}$ odd can be considered from three different viewpoints (as a particular case of $\mathsf{x}=0$, of $\mathsf{y}=0$ or of $\mathsf{x=y}$), hence leading to three possible different systems of independent solutions (\eqref{Q-00-1} and \eqref{Qbis-00-1}, \eqref{Q-00-2} and \eqref{Qbis-00-2}, or \eqref{Q-00-3} and \eqref{Qbis-00-3}).
\end{rem}

In the previous study, we have exhibited possible forms for the entire $Q$-solutions of \eqref{hom-eq}, hence susceptible to lead to a system of Bethe-type equations. 
One can in fact be slightly more precise in our argument by remarking that, if  $Q(\lambda)$ is a solution of the form \eqref{formQ-x0}, \eqref{formQ-y0} or \eqref{formQx=y} to the homogeneous functional equation \eqref{hom-eq}  for some function $\bar{\mathsf{t}}(\lambda)$ satisfying \eqref{periodt-1} and \eqref{periodt-2}, then at least one of the two functions \eqref{W-1} or \eqref{W-2} should be non identically vanishing. This means on the one hand that the sum rules \eqref{sum-rule1}, \eqref{sum-rule2} or \eqref{rule-x=y} should indeed be satisfied, so that we can indeed refine the form of $Q(\lambda)$ as in \eqref{Q-01}, \eqref{Q-00-1}, \eqref{Q-10}, \eqref{Q-00-2}, \eqref{Q-11} or \eqref{Q-00-3} according to the case. This means on the other hand that the condition {\it 2'.} of Proposition~\ref{prop-Qinh} is automatically satisfied. This is due to the following proposition and corollary:

\begin{proposition}\label{prop-W1-W2}
If  $Q(\lambda)$ is a non-zero solution of \eqref{hom-eq} of the form \eqref{formQ-x0},  \eqref{formQ-y0} or \eqref{formQx=y} to \eqref{hom-eq}  for some function $\bar{\mathsf{t}}(\lambda)$ satisfying \eqref{periodt-1} and \eqref{periodt-2}, then at least one of the two functions $W_Q^{(1)}(\lambda)$ \eqref{W-1} and $W_Q^{(2)}(\lambda)$ \eqref{W-2} is not identically zero.
\end{proposition}

\begin{proof}
Let us suppose that both functions $W_Q^{(1)}(\lambda)$ and $W_Q^{(2)}(\lambda)$ are identically zero, i.e. that $Q(\lambda)$ satisfies the relations
\begin{equation}
   \left\{
   \begin{aligned}
    &\frac{Q(\lambda+\pi)}{Q(\lambda+\pi-\eta)}=(-1)^\mathsf{y}\, \frac{Q(\lambda)}{Q(\lambda-\eta)},\\
    &\frac{Q(\lambda+\pi\omega)}{Q(\lambda+\pi\omega-\eta)}=(-1)^\mathsf{x}\, e^{-i\mathsf{N}\eta}\, 
    \frac{Q(\lambda)}{Q(\lambda-\eta)}.
   \end{aligned}
   \right.
\end{equation}
This means that the function $\frac{Q(\lambda)}{Q(\lambda-\eta)}$ is a double-Bloch function which  can therefore be written in the form
  \begin{equation}\
     \frac{Q(\lambda)}{Q(\lambda-\eta)}=c_P\, e^{\beta \lambda}
     \prod_{j=1}^\mathsf{M}\frac{\theta(\lambda-\mu_j)}
                                                     {\theta(\lambda-\nu_j)}
  \end{equation}
for some integer $\mathsf{M}$, some complex parameter $\beta$ and some sets of zeroes and poles $\mu_j$ and $\nu_j$, $1\le j \le \mathsf{M}$.
It follows that $Q(\lambda)$ should be of the form \eqref{Q-form} which, as already mentioned, is not compatible with the quasi-periodicity properties of \eqref{hom-eq}.
\end{proof}

\begin{corollary}\label{cor-W}
Let $\bar{\mathsf{t}}(\lambda)$  satisfying \eqref{periodt-1} and \eqref{periodt-2} and let us suppose that there exists a function $Q(\lambda)$ solution to the homogeneous functional equation \eqref{hom-eq} associated with $\bar{\mathsf{t}}(\lambda)$.
\vspace{-\topsep}
\begin{itemize}
\setlength\itemsep{0em}
\item If $\mathsf{x}=0$ and if $Q(\lambda)$ is of the form \eqref{formQ-x0}, then $W_Q^{(1)}(\lambda)$ is not identically zero and the sum rule \eqref{sum-rule1} is satisfied. Moreover, for all $n\in\{1,\ldots\mathsf{N}\}$, $\big(Q(\xi_n-\eta),Q(\xi_n-\eta+\pi)\big)\not=(0,0)$.
\item If $\mathsf{y}=0$ and if $Q(\lambda)$ is of the form \eqref{formQ-y0}, then $W_Q^{(2)}(\lambda)$ is not identically zero and the sum rule \eqref{sum-rule2} is satisfied. Moreover, for all $n\in\{1,\ldots\mathsf{N}\}$, $\big(Q(\xi_n-\eta),Q(\xi_n-\eta+\pi\omega)\big)\not=(0,0)$.
\item If $\mathsf{x=y}=0$ and if $Q(\lambda)$ is of the form \eqref{formQx=y}, then $W_Q^{(1)}(\lambda)$ and $W_Q^{(2)}(\lambda)$ are not identically zero and the sum rule \eqref{rule-x=y} is satisfied. Moreover, for all $n\in\{1,\ldots\mathsf{N}\}$, $\big(Q(\xi_n-\eta),Q(\xi_n-\eta+\pi)\big)\not=(0,0)$ and $\big(Q(\xi_n-\eta),Q(\xi_n-\eta+\pi\omega)\big)\not=(0,0)$.
\end{itemize}
Moreover, in all these cases, $\bar{\mathsf{t}}(\lambda)$ is a transfer matrix eigenvalue.
\end{corollary}

\begin{proof}
Let us first consider the case $\mathsf{x=0}$. The first part of the assertion is a direct consequence of Proposition~\ref{prop-W1-W2}, of the fact that $W_Q^{(2)}(\lambda)$ is identically zero for $Q(\lambda)$ of the form \eqref{formQ-x0} and of the previous study. 
Therefore $W_Q^{(1)}(\lambda)=c_W^{(1)}\, e^{\alpha_1\lambda}\, \mathsc{d}(\lambda)$ with $c_W^{(1)}\not=0$. This is clearly incompatible with the possibility that, for some $n\in\{1,\ldots,\mathsf{N}\}$, 
$\big(Q(\xi_n-\eta),Q(\xi_n-\eta+\pi)\big)=(0,0)$, since in that case we would have $W_Q^{(1)}(\xi_n-\eta)=0$. Hence the second assertion follows.

The cases $\mathsf{y=0}$ and $\mathsf{x=y}$ can be proven similarly.
\end{proof}

It is also easy to see that, whenever such solutions exist for $\bar{\mathsf{t}}(\lambda)\in\Sigma_{\overline{\mathcal{T}}}$, they are unique, which means that the characterization of a given eigenvalue (and corresponding eigenvector) of the antiperiodic transfer matrix in terms of the corresponding type of solutions of Bethe equation is uniquely determined (i.e the same eigenvalue cannot be obtained by two different solutions of the same type).

\begin{proposition}\label{prop-uniqueness}
For each $\bar{\mathsf{t}}(\lambda)\in\Sigma_{\overline{\mathcal{T}}}$, the equation \eqref{hom-eq} admits
\vspace{-\topsep}
\begin{itemize}
\setlength\itemsep{0em}
\item at most one independent solution of the form \eqref{Q-01} if $\mathsf{(x,y)}=(0,1)$;
\item at most one independent solution of the form \eqref{Q-10} if $\mathsf{(x,y)}=(1,0)$;
\item at most one independent solution of the form \eqref{Q-11} if $\mathsf{(x,y)}=(1,1)$;
\item  at most one independent solution of the form \eqref{Q-00-1} with the constraint \eqref{sum-00-1}, 
 at most one independent solution of the form \eqref{Q-00-2}, and at most one independent solution of the form \eqref{Q-00-3}  with the constraint \eqref{sum-00-3} if $\mathsf{(x,y)}=(0,0)$.
\end{itemize}
\vspace{-\topsep}
\end{proposition}

\begin{proof}
Let us suppose that there exists, for $\mathsf{(x,y)}\not=(0,0)$ and for a given  $\bar{\mathsf{t}}(\lambda)\in\Sigma_{\overline{\mathcal{T}}}$, two solutions $Q_1(\lambda)$ and $Q_2(\lambda)$ of the form \eqref{Q-01} (respectively \eqref{Q-10}, respectively \eqref{Q-11}).
Then their quantum Wronskian \eqref{W12} satisfies the periodicity property $W_{12}(\lambda+2\pi)=W_{12}(\lambda)$.
Hence $W_{12}(\lambda)=f_{12}(\lambda)\,\mathsc{d}(\lambda)$, where $f_{12}(\lambda)$ is an entire function with two incommensurable periods $2\pi$ and $2\eta$ (we suppose that $\eta\notin\pi\mathbb{Q}$), i.e. it is a constant. The latter is equal to zero if $(\mathsf{x},\mathsf{y})\not=(0,0)$ due to \eqref{rel-W2}.

Let now $\mathsf{(x,y)}=(0,0)$ and let us suppose that, for  a given  $\bar{\mathsf{t}}(\lambda)\in\Sigma_{\overline{\mathcal{T}}}$, there exist two solutions $Q_1(\lambda)$ and $Q_2(\lambda)$ of the form \eqref{Q-00-1} with the constraint \eqref{sum-00-1} (respectively of the form \eqref{Q-00-2}, respectively of the form \eqref{Q-00-3}  with the constraint \eqref{sum-00-3}). 
By considering the quasi-periodicity properties of their quantum Wronskian \eqref{W12} $W_{12}(\lambda)=f_{12}(\lambda)\,\mathsc{d}(\lambda)$ with respect to shifts of $2\pi$ and $\pi\omega$ (respectively of $\pi$ and $2\pi\omega$, respectively of $2\pi$ and $\pm\pi+\pi\omega$), one obtains that $f_{12}(\lambda)$ is of the form $f_{12}(\lambda)=c_{12} e^{\alpha_{12}\lambda}$, and one obtains some conditions on $\alpha_{12}$ if $c_{12}\not=0$. These conditions are clearly incompatible with the supposed form of $Q_1(\lambda)$ and $Q_2(\lambda)$, which means here also that $c_{12}=0$ and that the two solutions are not independent.
\end{proof}

We shall see in the next subsection that  the reformulation of the spectral problem for $\overline{\mathcal{T}}(\lambda)$ in terms of solutions of the form \eqref{Q-01}, \eqref{Q-10} or \eqref{Q-11} is effectively equivalent to the SOV characterization of the spectrum of Theorem~\ref{thm-eigen-t}, at least  in the case $\mathsf{N}$ even.


\subsection{On the reformulation of the SOV characterization of the spectrum in terms of solutions of an homogeneous functional $T$-$Q$ equation: proof of the completeness for $\mathsf{N}$ even}

We shall now prove the converse of Theorem~\ref{th-sens1} in the case of even $\mathsf{N}$ and for $Q(\lambda)$ of the form \eqref{Q-01}, \eqref{Q-10} or \eqref{Q-11},
i.e. the completeness of the solutions to the corresponding Bethe equations.
More precisely, we shall prove the following result:

\begin{theorem}\label{th-sens2}
Let us suppose that the inhomogeneity parameters $\xi_{1},\ldots ,\xi _{\mathsf{N}}$ satisfy \eqref{cond-inh} and let $\bar{\mathsf{t}}(\lambda )$ be an eigenvalue of the antiperiodic transfer matrix $\overline{\mathcal{T}}(\lambda )$ (i.e. $\bar{\mathsf{t}}(\lambda )\in \Sigma _{\overline{\mathcal{T}}}$).
Then, if $\mathsf{N}$ is even, there exists a function $Q(\lambda )$ of the form
\begin{equation}\label{Q-X}
  Q(\lambda)=\prod_{j=1}^\mathsf{N}\theta_{\mathsf{X}}(\lambda-\lambda_j),
\end{equation}
for some set of roots $\lambda_1,\ldots,\lambda_\mathsf{N}$, such that $\bar{\mathsf{t}}(\lambda )$ and $Q(\lambda )$ satisfy the homogeneous functional equation \eqref{hom-eq}.
Here the notation  $\theta_\mathsf{X}(\lambda)$ stands, respectively, for the function \eqref{theta-x0} if $\mathsf{(x,y)}=(0,1)$, for the function \eqref{theta-y0} if $\mathsf{(x,y)}=(1,0)$, or for the function \eqref{theta-xy} if $\mathsf{(x,y)}=(1,1)$.
%
\end{theorem}

From Proposition~\ref{prop-uniqueness} we know that such a solution is unique (up to normalization). 
We also know from Corollary~\ref{cor-W} that this solution is such that condition {\it 2'.} of Proposition~\ref{prop-Qinh} is satisfied.
Hence, in the case of even $\mathsf{N}$, it means that the discrete SOV characterization of the spectrum is {\em equivalent} to the description in terms of Bethe equations based on solutions 
of the form \eqref{Q-X} of the homogeneous functional equation \eqref{hom-eq}.
The description of the eigenstates can then be obtained as in Theorem~\ref{th-sens1}, or more generally as in its variant Propostion~\ref{prop-Qinh}.
Note here that it is possible to rewrite the corresponding eigenvectors in a form more similar to what we have in the ABA framework, i.e. by multiple action, on a given reference state, of a product of operators evaluated at the corresponding Bethe roots.
However, due to the fact that the function $Q(\lambda)$ \eqref{Q-X} is not a theta function of the same type (i.e. with the same quasi-periods) as the other functions defining the model, and unlike what happens in the context of the reformulation of Appendix~\ref{app-inhom} (see Corollary~\ref{Eigen-SOV-Bethe}), we do not use the whole operator $\mathcal{D}$ to construct the eigenstates, but only ``parts'' of this operator:  in the present case, the operator $\mathcal{D}$ is split into a product of commuting operators, and these are the latter  which are used to construct the eigenstates.
More precisely, for each $\mathsf{N}$-tuple $\boldsymbol{\beta}=(\beta_1,\ldots,\beta_\mathsf{N})\in\{0,1\}^\mathsf{N}$, let us define in $\mathrm{End}(\mathbb{D}_{\mathsf{(6VD)},\mathsf{N}})$ the following diagonal operator on the SOV basis given by the vectors $|\mathbf{h},r\rangle$:
\begin{equation}\label{def-Dbeta}
     D_{\boldsymbol{\beta}}(\lambda)\, |\mathbf{h},r\rangle
     =\prod_{n=1}^\mathsf{N}\!
     \left\{ 
     \Big[c_\mathsf{X}^{\vphantom{-1}}\, e^{\frac{i\pi(\mathsf{x+y-xy})h_n}{\mathsf{N}}+i\delta_\mathsf{y=0}(\lambda-\xi_n^{(h_n)})} \Big]^{\beta_n}\,
     \theta_\mathsf{X}\big(\lambda-\xi_n^{(h_n)}+\beta_n\pi_\mathsf{X}^{\vphantom{1}}\big) \right\}
      |\mathbf{h},r\rangle,
\end{equation}
where
\begin{align}
  &c_\mathsf{X}^{-1}=
  \begin{cases}
  \theta_4(0|\omega) &\text{if } \mathsf{x}=0,\\
  \frac{i}{2}\, e^{-i\frac{\pi\omega}{2}}\, \theta_2(0|\omega) &\text{if } \mathsf{y}=0,\\
  \frac{1}{2}\, e^{-i\frac{\pi\omega}{2}}\, \theta_2(0|\omega)\,\theta_3(0|\omega)\,\theta_4(0|\omega)\quad&\text{if } \mathsf{x=y},
  \end{cases}
  \\
  &\pi_\mathsf{X}^{\vphantom{1}}=(1-\delta_{\mathsf{y}=0})\,\pi+\delta_{\mathsf{y}=0}\, \pi\omega.
\end{align}
We have that
\begin{equation}
    D_{\boldsymbol{\beta}}(\lambda)\, D_{\boldsymbol{1-\beta}}(\lambda)\, |\mathbf{h},r\rangle
    =e^{\frac{i\pi(\mathsf{x+y-xy})}{\mathsf{N}}\sum_{n=1}^\mathsf{N}h_n}
    \prod_{n=1}^\mathsf{N} \theta(\lambda-\xi_n^{(h_n)})\,
     |\mathbf{h},r\rangle,
\end{equation}
where $\boldsymbol{1-\beta}$ stands for the $\mathsf{N}$-tuple $(1-\beta_1,\ldots,1-\beta_\mathsf{N})$. This means that
\begin{equation}
   \mathcal{D}(\lambda)=\mathsf{T_{\! D}}\, D_{\boldsymbol{\beta}}(\lambda)\, D_{1-\boldsymbol{\beta}}(\lambda),
\end{equation}
where $\mathsf{T_{\! D}}$ is the shift operator defined on the SOV basis as
\begin{equation}
   \mathsf{T_{\! D}}\,  |\mathbf{h},r\rangle
   = e^{\frac{i\pi(\mathsf{x+y-xy})}{\mathsf{N}}\frac{s_\mathbf{h}-s_\mathbf{0}}{2}}\, 
   e^{-i\mathsf{y}\eta\frac{s_\mathbf{h}-s_\mathbf{1}}{2}}\,
\frac{\theta(t_{r+1,\mathbf{h}})}{\theta (t_{r+1,\mathbf{1}})}\, 
    |\mathbf{h},r+1\rangle.
\end{equation}
Then one can formulate the following corollary:

\begin{corollary}\label{Cor-Eigen-Bethe}
Under the condition \eqref{cond-inh} and if $\mathsf{N}$ is even, there exists a one-to-one correspondence between $\Sigma _{\overline{\mathcal{T}}}$ and the set $\Sigma_\text{BAE}$ of different (up to the real quasi-period of the function $\theta_\mathsf{X}$) Bethe roots $\Lambda=\{\lambda_1,\ldots,\lambda_\mathsf{N}\}$ such that
\begin{enumerate}
  \item the function 
\begin{equation}\label{Bethe-fct}
    \mathsc{a}_\mathsf{x,y}(\lambda)\, \frac{Q(\lambda-\eta)}{Q(\lambda)}
    +\mathsc{d}(\lambda)\, \frac{Q(\lambda+\eta)}{Q(\lambda)}
\end{equation}
is entire and satisfies the quasi-periodicity properties \eqref{periodt-1}-\eqref{periodt-2},
  \item $\forall n\in\{1,\ldots,\mathsf{N}\}$, there exists $\beta_n\in\{0,1\}$ such that $Q(\xi_n+\beta_n\pi_\mathsf{X}^{\vphantom{1}})\not= 0$,
\end{enumerate}
where $Q(\lambda)$ is defined in terms of $\Lambda$ by \eqref{Q-X}.
The eigenvalue $\bar{\mathsf{t}}(\lambda )\in\Sigma _{\overline{\mathcal{T}}}$ associated with $\Lambda\in\Sigma_\text{BAE}$ is then given by the entire function \eqref{Bethe-fct}.
For any $\kappa\in\mathbb{C}\setminus\{0\}$, the corresponding one-dimensional  right and left  eigenspaces of the $\kappa$-twisted transfer matrix $\overline{\mathcal{T}}^{(\kappa)}(\lambda)$  are the one-dimensional subspaces of $\mathbb{\bar{D}}_{\mathsf{(6VD)},\mathsf{N}}^{(0,\mathcal{R/L})}$  spanned by all vectors of the type
\begin{equation}\label{Bethe-eigen}
|\Psi_{\Lambda,\boldsymbol{\beta}}^{( \kappa ) }\rangle
    =\prod_{j=1}^{\mathsf{N}} D_{\boldsymbol{\beta}}(\lambda_j)\,
       |\Omega^{(\kappa)}\rangle ,
       \qquad\text{respectively}\quad
\langle \Psi_{\Lambda,\boldsymbol{\beta}}^{( \kappa ) }|
    =\langle \Omega^{(\kappa)}|
      \prod_{j=1}^{\mathsf{N}} D_{\boldsymbol{\beta}}(\lambda_j) ,
\end{equation}
for any $\mathsf{N}$-tuple $\boldsymbol{\beta}\in\{0,1\}^\mathsf{N}$.
In \eqref{Bethe-eigen}, the operators $D_{\boldsymbol{\beta}}(\lambda)$ are defined as in \eqref{def-Dbeta}, and the reference states $|\Omega^{(\kappa)}\rangle$ and $\langle \Omega^{(\kappa)}|$ are
\begin{align}
& |\Omega^{(\kappa)}\rangle 
   =\sum_{\mathbf{h}\in \{0,1\}^{\mathsf{N}}}
     \prod_{a=1}^{\mathsf{N}}\left( 
     \frac{e^{i\mathsf{y}\eta}\, \mathsc{a}_{\mathsf{x,y}}(\xi _{a})}{\kappa \,\mathsc{d}(\xi _{a}-\eta )}\right) ^{\!h_{a}}
     \det_{\mathsf{N}}\big[\Theta ^{(0,\mathbf{h})}\big]\,
     |\mathbf{h},0\rangle ,  \label{ref1} \\
& \langle  \Omega^{(\kappa)}|
   =\sum_{\mathbf{h}\in \{0,1\}^{\mathsf{N}}}\prod_{a=1}^{\mathsf{N}}\Big( 
     \kappa\, e^{i\mathsf{y}\eta }\Big) ^{h_{a}}
     \det_{\mathsf{N}}\big[\Theta ^{(0,\mathbf{h})}\big]\,
     \langle 0,\mathbf{h}|.  \label{ref2}
\end{align}
\end{corollary}

The proof of Theorem \ref{th-sens2} is based on the following lemma:

\begin{lemma}\label{lemma3}
Let us suppose that the inhomogeneity parameters $\xi_{1},\ldots ,\xi _{\mathsf{N}}$ satisfy \eqref{cond-inh} and let $\mathsf{N}=2\mathsf{M}$ be even. Then, for each $\bar{\mathsf{t}}(\lambda )\in\Sigma _{\overline{\mathcal{T}}}$, there exist $\Delta\in\mathbb{C}$ and two non-zero entire functions $Q_+(\lambda )$ and $Q_-(\lambda)$ with the following quasi-periodicity properties:
\begin{align}
   &Q_\pm(\lambda +\pi )=(\pm1)^\mathsf{y} \, Q_\pm(\lambda ), \label{Qeps-pi}\\
   &Q_\pm(\lambda +\pi \omega )=(\pm1)^\mathsf{x}\, (-e^{-2i\lambda -i\pi \omega })^\mathsf{M}\,
    e^{2i\Delta -i\mathsf{M}\pi\omega}\, Q_\pm(\lambda ), \label{Qeps-piom}
\end{align}
which satisfy the following system of $2\mathsf{N}$ equations:
\begin{equation}\label{sys-inh}
\left\{
\begin{aligned}
&\bar{\mathsf{t}}(\xi _{j})\,Q_{+}(\xi _{j})=-\mathsc{a}(\xi _{j})\,Q_{-}(\xi_{j}-\eta ),\\
&\bar{\mathsf{t}}(\xi _{j}-\eta )\, Q_{-}(\xi _{j}-\eta )=\mathsc{d}(\xi _{j}-\eta )\,Q_{+}(\xi _{j}),
\end{aligned}
\right.
\qquad 1\leq j\leq \mathsf{N}.
\end{equation}
\end{lemma}

\begin{proof}
On the one hand, since $\bar{\mathsf{t}}(\lambda )\in \Sigma _{\overline{\mathcal{T}}}$ the system \eqref{sys-inh} is equivalent to the following system of only $\mathsf{N}$ equations:
\begin{equation}
\bar{\mathsf{t}}(\xi _{j})\,Q_{+}(\xi _{j})=-\mathsc{a}(\xi _{j})\,Q_{-}(\xi_{j}-\eta ),\qquad
1\leq j\leq \mathsf{N}.  
\label{+-system}
\end{equation}
On the other hand, the entireness and quasi-periodicity properties of the two functions $Q_\pm(\lambda)$ are equivalent to the fact that the functions
\begin{equation}
   \widetilde{Q}_\pm(\lambda)\equiv e^{-i[\mathsf{y}\frac{1\mp1}{2}-\mathsf{M}]\lambda}\, Q_\pm(\lambda)
\end{equation}
are theta functions of $\lambda$ of quasi-periods $(\pi,\pi\omega)$, of order $\mathsf{M}$ and of norm 
\begin{equation}
   \Delta_\pm=\Delta-\frac{1\mp1}{2}\left(\mathsf{x}\frac{\pi}{2}+\mathsf{y}\frac{\pi\omega}{2}\right).
\end{equation}
This is also equivalent to the fact that the functions $Q_\pm(\lambda)$ can be represented in the following form,
\begin{align}
    &Q_+(\lambda ) 
    =\sum_{k=1}^{\mathsf{M}} e^{-i\mathsf{M}(\lambda-\xi_{i_k})}\,
    \frac{\theta (\lambda -\xi_{i_k} +\sum_{\ell=1}^\mathsf{M}\xi_{i_\ell}-\Delta_+)}
            {\theta (\sum_{\ell=1}^\mathsf{M}\xi_{i_\ell}-\Delta_+)}\,
    \prod_{\substack{ \ell =1 \\ \ell \not=k}}^\mathsf{M}
    \frac{\theta (\lambda -\xi _{i_\ell } )}{\theta (\xi_{i_k}-\xi _{i_\ell })}\,
    Q_+(\xi _{i_k}), \label{Q+}
    \\
   &Q_-(\lambda ) 
   =\sum_{k=1}^{\mathsf{M}} e^{i(\mathsf{y}-\mathsf{M})(\lambda-\xi_{i_k}+\eta)}\,
   \frac{\theta (\lambda -\xi_{i_k}+\eta+\sum_{\ell=1}^\mathsf{M}\xi_{i_\ell}-\mathsf{M}\eta-\Delta_-)}
          {\theta (\sum_{\ell=1}^\mathsf{M}\xi_{i_\ell}-\mathsf{M}\eta-\Delta_-)}\,
          \nonumber\\
  &\hspace{9cm}\times
   \prod_{\substack{ \ell =0 \\ \ell \not=k}}^{\mathsf{M}}
   \frac{\theta (\lambda-\xi _{i_\ell }+\eta)}{\theta (\xi _{i_k}-\xi _{i_\ell })}\,Q_-(\xi _{i_k}-\eta),
   \label{Q-}
\end{align}
in terms of some arbitrary $\mathsf{M}$-tuples $\big( Q_+(\xi_{i_1}),\ldots,Q_+(\xi_{i_\mathsf{M}})\big)$ and $\big( Q_-(\xi_{i_1}-\eta),\ldots,Q_-(\xi_{i_\mathsf{M}}-\eta)\big)$, provided that $\Delta_+-\sum_{\ell=1}^\mathsf{M}\xi_{i_\ell}\notin\Gamma$ and $\Delta_--\sum_{\ell=1}^\mathsf{M}\xi_{i_\ell}+\mathsf{M}\eta\notin\Gamma$.
Here, we have arbitrarily split the set of the $\mathsf{N}$ inhomogeneity parameters $\{\xi_1,\ldots,\xi_\mathsf{N} \}$ into two disjoint subsets $\{\xi_{i_1},\ldots,\xi_{i_\mathsf{M}} \}$ and $\{ \xi_{i_{\mathsf{M}+1}},\ldots,\xi_{i_\mathsf{N}} \}$ of the same cardinality $\mathsf{M}=\mathsf{N}/2$.

Hence, the system \eqref{+-system} for $Q_+(\lambda)$ and $Q_-(\lambda)$ is equivalent to the following linear system for these two $\mathsf{M}$-tuples:
\begin{equation}\label{sys+-bis}
\left\{
   \begin{aligned}
   &Q_{-}(\xi _{i_j}-\eta )=-\frac{\bar{\mathsf{t}}(\xi _{i_j})}{\mathsc{a}(\xi _{i_j})}\,Q_{+}(\xi _{i_j}),\\
   &\sum_{k=1}^{\mathsf{M}}\big[\mathcal{X}_{\bar{\mathsf{t}}}(\Delta )\big]_{jk}\ Q_{+}(\xi _{i_k})=0,
  \end{aligned}
\right.
\qquad j=1,\ldots ,\mathsf{M},
\end{equation}
where $\mathcal{X}_{\bar{\mathsf{t}}}(\Delta )$ is the $\mathsf{M}\times \mathsf{M}$
matrix of elements: 
\begin{multline}\label{mat-X}
   \big[\mathcal{X}_{\bar{\mathsf{t}}}(\Delta )\big]_{ab}
    =e^{-i\mathsf{M}(\xi_{i_{a+\mathsf{M}}}-\xi_{i_b})}\,
    \prod_{\substack{ \ell=1 \\ \ell \not=b}}^{\mathsf{M}}
    \frac{\theta (\xi _{i_{a+\mathsf{M}}}-\xi _{i_\ell})}{\theta (\xi _{i_b}-\xi _{i_\ell })}
    \left[ 
    \frac{\theta (\xi _{i_{a+\mathsf{M}}}-\xi _{i_b}+\sum_{\ell=1}^\mathsf{M}\xi_{i_\ell}-\Delta_+)}
           {\theta(\sum_{\ell=1}^\mathsf{M}\xi_{i_\ell}-\Delta_+)}\,
   \frac{\bar{\mathsf{t}}(\xi _{i_{a+\mathsf{M}}})}{\mathsc{a}(\xi _{i_{a+\mathsf{M}}})}\right.  
    \\
    \left.
    -e^{i\mathsf{y}(\xi_{i_{a+\mathsf{M}}}-\xi_{i_b})}\,
    \frac{\theta (\xi _{i_{a+\mathsf{M}}}-\xi _{i_b}+\sum_{\ell=1}^\mathsf{M}\xi_{i_\ell}-\mathsf{M}\eta-\Delta_-)}
           {\theta (\sum_{\ell=1}^\mathsf{M}\xi_{i_\ell}-\mathsf{M}\eta-\Delta_-)}\,
    \frac{\bar{\mathsf{t}}(\xi_{i_b})}{\mathsc{a}(\xi _{i_b})}
    \right] .
\end{multline}
The first line in \eqref{sys+-bis} corresponds to the equations $i_1,\ldots,i_\mathsf{M}$ of \eqref{+-system}, which completely fixe the $\mathsf{M}$-tuple $\big( Q_-(\xi_{i_1}-\eta),\ldots,Q_-(\xi_{i_\mathsf{M}}-\eta)\big)$ in terms of the $\mathsf{M}$-tuple $\big( Q_+(\xi_{i_1}),\ldots,Q_+(\xi_{i_\mathsf{M}})\big)$, whereas the second line  in \eqref{sys+-bis} corresponds to the equations $i_{\mathsf{M}+1},\ldots,i_\mathsf{N}$  of \eqref{+-system}, in which we have used the representations \eqref{Q+} and \eqref{Q-} that we have rewritten (by means of the first line of \eqref{sys+-bis}) in terms of the  $\mathsf{M}$-tuple $\big( Q_+(\xi_{i_1}),\ldots,Q_+(\xi_{i_\mathsf{M}})\big)$ only.

Hence, this system admits a non-zero solution if and only if the determinant of the matrix $\mathcal{X}_{\bar{\mathsf{t}}}(\Delta )$ \eqref{mat-X} vanishes.
It is easy to see that this determinant is a quasi-periodic function of $\Delta$ with quasi-periodicity properties
\begin{align}
    &\det_\mathsf{M}\big[\mathcal{X}_{\bar{\mathsf{t}}}(\Delta+\pi )\big]
    =\det_\mathsf{M}\big[\mathcal{X}_{\bar{\mathsf{t}}}(\Delta )\big],\\
    &\det_\mathsf{M}\big[\mathcal{X}_{\bar{\mathsf{t}}}(\Delta+\pi\omega )\big]
    =e^{2i\sum_{a=1}^\mathsf{M}(\xi_{i_{a+\mathsf{M}}}-\xi_{i_a})}\,\det_\mathsf{M}\big[\mathcal{X}_{\bar{\mathsf{t}}}(\Delta )\big],
\end{align}
so that, if not identically zero, it is a non-constant\footnote{at least for some adequate splitting of $\{\xi_1,\ldots,\xi_\mathsf{N}\}$.} function of $\Delta$ which can be written in the form
\begin{equation}
   \det_\mathsf{M}\big[\mathcal{X}_{\bar{\mathsf{t}}}(\Delta )\big]
   =c_\mathcal{X}\,\frac{\prod_{j=1}^\mathsf{2M}\theta(\Delta-\Delta_j)}{\theta(\Delta_+-\sum_{\ell=1}^\mathsf{M}\xi_{i_\ell})^\mathsf{M}\ \theta(\Delta_--\sum_{\ell=1}^\mathsf{M}\xi_{i_\ell}+\mathsf{M}\eta)^\mathsf{M}},
\end{equation}
in terms of some roots $\Delta_j$ which are not all equal to the roots appearing in the denominator.
This ends the proof of Lemma~\ref{lemma3}.
\end{proof}

{\it Proof of Theorem~\ref{th-sens2}.}
Let $\bar{\mathsf{t}}(\lambda )\in \Sigma _{\overline{\mathcal{T}}}$ and $\mathsf{N}=2\mathsf{M}$ be even.
Then Lemma~\ref{lemma3} implies that there exists two non-zero entire functions $Q_{\pm }(\lambda )$ which satisfy the quasi-periodicity properties \eqref{Qeps-pi}-\eqref{Qeps-piom} and the system \eqref{sys-inh} for some $\Delta\in\mathbb{C}$.
This implies that there exist two entire functions $f_\pm(\lambda)$, with quasi-periodicity properties
\begin{align}
   &f_\pm(\lambda +\pi )=(\pm1)^\mathsf{y} \, f_\pm(\lambda ), \label{feps-pi}\\
   &f_\pm(\lambda +\pi \omega )=(\pm1)^\mathsf{x}\, (-e^{-2i\lambda -i\pi \omega })^\mathsf{M}\,
    e^{2i\Delta -i\mathsf{M}\pi\omega}\, f_\pm(\lambda ), \label{feps-piom}
\end{align}
which satisfy together with the functions $Q_{\pm }(\lambda )$ the following system of functional equations:
\begin{equation}\label{sys-Qf}
    \left\{
    \begin{aligned}
    &\bar{\mathsf{t}}(\lambda )\,Q_{+}(\lambda ) 
    =-\mathsc{a}(\lambda)\,Q_{-}(\lambda -\eta )+\mathsc{d}(\lambda )\,f_-(\lambda+\eta ), \\
    &\bar{\mathsf{t}}(\lambda )\,Q_{-}(\lambda ) 
    =-\mathsc{a}(\lambda)\,f_+(\lambda-\eta )+\mathsc{d}(\lambda )\,Q_{+}(\lambda +\eta ).
    \end{aligned}
    \right.
\end{equation}
Particularizing the first line of \eqref{sys-Qf} at the points $\xi_j-\eta$, $j=1,\ldots,\mathsf{N}$, and the second line at the points $\xi_j$, $j=1,\ldots,\mathsf{N}$, and using the fact that $\bar{\mathsf{t}}(\lambda )\in \Sigma _{\overline{\mathcal{T}}}$, we obtain that
\begin{equation}\label{Q+Q--inh}
     Q_{+}(\xi _{j}-\eta )\,Q_{-}(\xi _{j})=f_+(\xi _{j}-\eta )\, f_-(\xi _{j}),
     \qquad  1\leq j\leq \mathsf{N}.
\end{equation}
Note that the function $Q_+(\lambda-\eta)\, Q_-(\lambda)$ and the function $f_+(\lambda-\eta)\, f_-(\lambda)$ satisfy the same quasi-periodicity properties:
\begin{align}
     &F_{+-}(\lambda+\pi)=F_{+-}(\lambda),\\
     &F_{+-}(\lambda+\pi\omega)=(-e^{-2i\lambda-i\pi\omega})^\mathsf{N}\, e^{4i\Delta+i\mathsf{N}\eta-i\mathsf{N}\pi\omega}\, F_{+-}(\lambda),
\end{align}
so that $F_{+-}(\lambda)$ is a theta function of order $\mathsf{N}$ and of norm $2\Delta+\mathsf{M}\eta-\mathsf{M}\pi\omega$, where $F_{+-}(\lambda)$ stands either for the function $Q_+(\lambda-\eta)\, Q_-(\lambda)$ or for the function $f_+(\lambda-\eta)\, f_-(\lambda)$.
Hence \eqref{Q+Q--inh} implies that the identity is in fact valid at the functional level, i.e. for any $\lambda\in\mathbb{C}$:
\begin{equation}\label{id-Q+Q-}
Q_{+}(\lambda -\eta )\, Q_{-}(\lambda )=f_+(\lambda -\eta )\, f_-(\lambda ).
\end{equation}
One can therefore eliminate $f_+(\lambda-\eta)$ in the system \eqref{sys-Qf}, and we obtain the following functional system for the function $f_-(\lambda)$:
\begin{equation}\label{sys-Qf-}
    \left\{
    \begin{aligned}
    &\bar{\mathsf{t}}(\lambda )\,Q_{+}(\lambda ) 
    =-\mathsc{a}(\lambda)\,Q_{-}(\lambda -\eta )+\mathsc{d}(\lambda )\,f_-(\lambda+\eta ), \\
    &\bar{\mathsf{t}}(\lambda )\,Q_{-}(\lambda ) 
    =-\mathsc{a}(\lambda)\,\frac{Q_+(\lambda-\eta)\, Q_-(\lambda)}{f_-(\lambda )}+\mathsc{d}(\lambda )\,Q_{+}(\lambda +\eta ).
    \end{aligned}
    \right.
\end{equation}
Note that the ratio in the second line of \eqref{sys-Qf-} is in fact an entire function due to \eqref{id-Q+Q-}.
This systems implies that
\begin{equation}\label{Wf1f2}
     \mathsc{a}(\lambda)\, f_1(\lambda)= \mathsc{d}(\lambda)\, f_2(\lambda+\eta),
\end{equation}
in which we have set
\begin{align}
    &f_1(\lambda)= Q_-(\lambda-\eta)\, Q_-(\lambda)- \frac{Q_+(\lambda-\eta)\, Q_-(\lambda)}{f_-(\lambda )}\, Q_+(\lambda),\\
   &f_2(\lambda) = f_-(\lambda)\, Q_-(\lambda-\eta)- Q_+(\lambda)\, Q_+(\lambda-\eta).
\end{align}
$f_1(\lambda)$ and $f_2(\lambda)$ are two entire functions of $\lambda$ which  are both theta functions of order $\mathsf{N}$, and therefore the relation \eqref{Wf1f2} implies that there exists two constants $c_1,c_2\in\mathbb{C}$ such that
\begin{equation}
   \left\{
   \begin{aligned}
   & f_1(\lambda)= c_1\, \mathsc{d}(\lambda),\\
   & f_2(\lambda+\eta)=c_2\, \mathsc{a}(\lambda),\\
   & c_1=c_2,
   \end{aligned}
   \right.
   \qquad
   \text{i.e.}
   \qquad
   f_1(\lambda)=f_2(\lambda)=c_1\, \mathsc{d}(\lambda).
\end{equation}
Hence
\begin{equation}
    f_-(\lambda)=Q_-(\lambda), \qquad
    \text{and}\qquad
    f_+(\lambda)=Q_+(\lambda),
\end{equation}
which means that $Q_-(\lambda)$ and $Q_+(\lambda)$ satisfy the following system:
\begin{equation}\label{sys-Q-Q+}
    \left\{
    \begin{aligned}
    &\bar{\mathsf{t}}(\lambda )\,Q_{+}(\lambda ) 
    =-\mathsc{a}(\lambda)\,Q_{-}(\lambda -\eta )+\mathsc{d}(\lambda )\,Q_-(\lambda+\eta ), \\
    &\bar{\mathsf{t}}(\lambda )\,Q_{-}(\lambda ) 
    =-\mathsc{a}(\lambda)\,Q_+(\lambda-\eta )+\mathsc{d}(\lambda )\,Q_{+}(\lambda +\eta ).
    \end{aligned}
    \right.
\end{equation}

Let us now define the entire functions
\begin{equation}
Q(\lambda )=\frac{Q_{+}(\lambda )+Q_{-}(\lambda )}{2},
\qquad
\text{and}
\qquad
\overline{Q}(\lambda )=\frac{Q_{+}(\lambda )-Q_{-}(\lambda )}{2}.
\end{equation}
with quasi-periodicity properties
\begin{align}
   &Q(\lambda +\pi )=(1-\mathsf{y})\, Q(\lambda )+\mathsf{y}\, \overline{Q}(\lambda ),\\
   &Q(\lambda +\pi \omega )=(-e^{-2i\lambda -i\pi \omega })^{\mathsf{M}}\, e^{2i\Delta-i\mathsf{M}\pi\omega }\left[ (1-\mathsf{x})\, Q(\lambda )+\mathsf{x}\, \overline{Q}(\lambda )\right] ,
\end{align}
and
\begin{align}
   &\overline{Q}(\lambda +\pi )=(1-\mathsf{y})\,\overline{Q}(\lambda )+\mathsf{y} \, Q(\lambda ),\\
   &\overline{Q}(\lambda +\pi \omega )=(-e^{-2i\lambda -i\pi \omega })^{\mathsf{M}}\, e^{2i\Delta-i\mathsf{M}\pi\omega }\left[ (1-\mathsf{x})\,\overline{Q}(\lambda )+\mathsf{x}\, Q(\lambda )\right] .
\end{align}
From these quasi-periodicity relations, it is easy to see that $Q(\lambda)$ and $\overline{Q}(\lambda)$  are both (up to a constant normalization factor) of the form \eqref{Q-X}.
Moreover they satisfy the following respective homogeneous functional equations:
\begin{align}
  &\bar{\mathsf{t}}(\lambda )\, Q(\lambda ) 
  =-\mathsc{a}(\lambda )\, Q(\lambda-\eta )+\mathsc{d}(\lambda )\, Q(\lambda +\eta ), \\
  &\bar{\mathsf{t}}(\lambda )\, \overline{Q}(\lambda ) 
  =\mathsc{a}(\lambda )\, \overline{Q}(\lambda -\eta )-\mathsc{d}(\lambda )\,\overline{Q}(\lambda +\eta ).
\end{align}
This ends the proof of Theorem~\ref{th-sens2}.
\qed

\section{Local operators and dynamical Yang-Baxter algebra}
\label{sec-inv}

In the next section, we shall compute determinant representations for form factors of local operators in the $\mathbb{\bar{D}}_{\mathsf{(6VD)},\mathsf{N}}^{(0,\mathcal{L}/\mathcal{R})}$-basis of the eigenstates of the $\kappa$-twisted transfer matrix.
As in the algebraic Bethe Ansatz framework \cite{KitMT99,LevT13a}, such representations are based on the solution of the quantum inverse problem, i.e. on the fact that one can reconstruct the local operators we consider in terms of the generators of the Yang-Baxter algebra.
A particularly crucial point in this respect comes from the fact that the positions on the lattice of these local operators are given in terms of propagators written as products of transfer matrices, so that their action on the corresponding eigenstates merely contributes as simple numerical coefficients.

For the dynamical 6-vertex model that we consider in this paper, the local operators of interest are essentially of two types: local {\em spin} operators and local {\em height} operators.
In our framework, local spin operators correspond to elementary operators acting non-trivially on only one factor of the space tensor product $\mathbb{V}_{\mathsf{N}}\equiv \otimes_{n=1}^\mathsf{N} V_n$ and can be expressed in terms of the usual basis of elementary matrices  $\big\{ E_{n}^{ij}\, ;\, i,j\in \{+,-\}, 1\le n \le \mathsf{N} \big\}$ of $\mathrm{End}(\mathbb{V}_\mathsf{N})$ defined by
\begin{equation}\label{local-spin}
E_{n}^{ij}
=
\begin{pmatrix}
\delta _{+,i }\,\delta _{+,j } & \delta _{+,i }\,\delta _{-,j }
\\ 
\delta _{-,i }\,\delta _{+,j } & \delta _{-,i }\,\delta _{-,j }
\end{pmatrix}_{\! [n]}\qquad i ,j \in \{+,-\}, 1\le n \le \mathsf{N},
\end{equation}
whereas the local height operators are dynamical operators acting non-trivially on the dynamical space $\mathbb{D}$.
In the physical context of classical SOS face models, the local operators that we should especially consider for the computation of correlation functions are the following (see for instance Section~3 of \cite{LevT14a}): the spin operators $E_n^{++}$ and $E_n^{--}$, which are associated with the values of the classical spin variable ($+$ or $-$) on the $n$-th vertical bond of a given vertical line of the lattice, and the local height operators $\widehat{\delta}^{(j)}_t$, $t\in t_0+\eta\mathbb{Z}$, which
are associated with the values $t$ of the local height variable on the $j$-th site of a given vertical line of the lattice.
Note in particular that the local height operators $\widehat{\delta}^{(1)}_t$ at site 1 (the reference site) have the following simple action on the dynamical-spin basis \eqref{dyn-spin-basis}:
\begin{align}
  &(\otimes _{n=1}^{\mathsf{N}}\langle n,h_{n}|)\otimes \langle t(a)|\,   \widehat{\delta}^{(1)}_t 
 = \delta_{t, t(a)}\, (\otimes _{n=1}^{\mathsf{N}}\langle n,h_{n}|)\otimes \langle t(a)|, \label{act1}\\
  & \widehat{\delta}^{(1)}_t \, (\otimes _{n=1}^{\mathsf{N}}|n,h_{n}\rangle )\otimes |t(a)\rangle
  = \delta_{t, t(a)}\, (\otimes _{n=1}^{\mathsf{N}}|n,h_{n}\rangle )\otimes |t(a)\rangle . \label{act2}
\end{align}
Its action on the SOV basis~\eqref{D-left-eigenstates}-\eqref{D-right-eigenstates} hence simply follows from \eqref{act-tau}:
\begin{equation}
   \langle r,\mathbf{h}|\, \widehat{\delta}^{(1)}_t=\delta_{t,t_{r,\mathbf{h}}}\,  \langle r,\mathbf{h}|,
   \qquad
   \widehat{\delta}^{(1)}_t \, |\mathbf{h},r\rangle =\delta_{t,t_{r,\mathbf{h}}}\, |\mathbf{h},r\rangle .
\end{equation}
In the following, we shall denote this operator as $\widehat{\delta}^{(1)}_t\equiv \delta_t(\tau)$.

Reconstruction of local spin operators in terms of the elements of the dynamical periodic monodromy matrix \eqref{mon-op} has been obtained in \cite{LevT13a}, generalizing to the dynamical case the proof of \cite{KitMT99,MaiT00}. In this paper, we need instead to express these operators in terms of the elements of the $\kappa$-twisted {\em antiperiodic} monodromy matrix \eqref{anti-p-6vD-M} (using in particular propagators given in terms of the $\kappa$-twisted transfer matrix \eqref{anti-transfer}) so as to be able to easily compute their action on the eigenstates \eqref{eigenT-r}-\eqref{eigenT-l}.

So as to remain as general as possible and to present a proof that is valid in both periodic and antiperiodic cases, we shall instead consider a quasi-periodic monodromy matrix of the form
\begin{equation}\label{Y-mon}
    \mathcal{M}_0^{(Y)}(\lambda)\equiv Y_0 \,\mathcal{M}_0(\lambda),
\end{equation}
where $\mathcal{M}$ is given by \eqref{mon-op} and $Y$ stands for any numerical invertible matrix (the $\kappa$-twisted antiperiodic case corresponding to the particular choice $Y\equiv X^{(\kappa)}\sigma^x$). The corresponding quasi-periodic transfer matrix will be denoted by $\mathcal{T}^{(Y)}(\lambda)\equiv\mathrm{tr}_0\big[\mathcal{M}_0^{(Y)}(\lambda)\big]$. The aim is therefore to express local operators in terms of the operator entries of \eqref{Y-mon}.
As in the periodic case, this reconstruction is based on the following lemma:

\begin{lemma}\label{lemme-mon}\cite{LevT13a}
We have the following identity between products of monodromy matrices:
\begin{multline}\label{id-T2}
 \overrightarrow{ \prod\limits_{k=1}^m} \mathsf{M}_{a_k,1\ldots N}\bigg(\xi_k|\tau+\eta\sum_{l=1}^{k-1} \sigma^z_{a_l} \bigg) 
 =\overleftarrow{ \prod\limits_{k=1}^m} 
 \mathsf{M}_{a_k, k\, k+1\ldots N\, a_1\, a_2\ldots a_{k-1}} \bigg(\xi_k|\tau+\eta\sum_{l=1}^{k-1}\sigma^z_l+\eta\sum_{l=k+1}^m \sigma^z_{a_l} \bigg),
\end{multline} 
in which the symbols $\overrightarrow{ \prod\limits_{k=1}^m}$ (respectively $\overleftarrow{ \prod\limits_{k=1}^m}$) means that the product is ordered from $1$ to $m$ (respectively from $m$ to $1$).
In this expression, $\mathsf{M}_{a_k, k\, k+1\ldots N\, a_1\, a_2\ldots a_{k-1}} $ denotes the monodromy matrix of a chain of $N$ sites labelled in this order by $k,\, k+1,\ldots N,\, a_1,\, a_2,\ldots a_{k-1}$ with inhomogeneity parameters $\xi_{k},\xi_{k+1},\ldots\xi_N,\xi_1,\xi_2,\ldots\xi_{k-1}$.
\end{lemma}
We will also use the following result:
\begin{lemma}\label{lemme-inh}
The trace of the inverse of the quasi-periodic monodromy matrix \eqref{Y-mon} evaluated at some inhomogeneity parameter $\xi_n$, $n\in\{1,\ldots,\mathsf{N}\}$, is equal to the inverse of the quasi-periodic transfer matrix evaluated at $\xi_n$, i.e.
\begin{equation}
  \mathrm{tr}_0\big[\mathcal{M}^{(Y)}_0 (\xi_n )^{-1} \big]
  = \mathcal{T}^{(Y)}(\xi_n )^{-1}.
\end{equation}
\end{lemma}

\begin{proof}
For any given numerical matrices $Z$ and $\tilde{Z}$, we independently compute the two traces $\mathrm{tr}_0\big[ Z_0\,\mathcal{M}_0(\xi_n)\big]$ and $\mathrm{tr}_0\big[ \mathcal{M}_0(\xi_n)^{-1}\,\tilde{Z}_0\big]$ as products of $R$-matrices or of inverse $R$-matrices, using the fact that the $R$-matrix \eqref{mat-R} becomes proportional to the permutation operator $P_{12}$ when evaluated at $\lambda=0$: $R_{12}(0|\tau)=a(0)P_{12}$.
On the one hand, we obtain
\begin{align*}
   \mathrm{tr}_0\big[ Z_0\,\mathcal{M}_0(\xi_n)\big]
   &=\mathrm{tr}_0\Bigg[ Z_0 \, 
        \overleftarrow{ \prod\limits_{k=n+1}^\mathsf{N}}\!\! R_{0 k}\bigg(\xi_n-\xi_k |\tau+\eta\sum_{l=1}^{k-1}\sigma_l^z\bigg)
        \cdot a(0) P_{0n}\nonumber\\
   &\hspace{6cm}\times     
        \overleftarrow{ \prod\limits_{k=1}^{n-1}} R_{0 k}\bigg(\xi_n-\xi_k |\tau+\eta\sum_{l=1}^{k-1}\sigma_l^z\bigg)\,
        \mathsf{T}_\tau^{\sigma_0^z}
        \Bigg]
       \displaybreak[0] \nonumber\\
   &=      \overleftarrow{ \prod\limits_{k=1}^{n-1}} R_{n k}\bigg(\xi_n-\xi_k |\tau+\eta\sum_{l=1}^{k-1}\sigma_l^z\bigg)\,
        \mathsf{T}_\tau^{\sigma_n^z}\nonumber\\
   &\hspace{4cm}\times
   \mathrm{tr}_0\Bigg[ Z_0 \, 
        \overleftarrow{ \prod\limits_{k=n+1}^\mathsf{N}} R_{0 k}\bigg(\xi_n-\xi_k |\tau+\eta\sum_{\substack{l=1\\ l\not=n}}^{k-1}\sigma_l^z\bigg)
        \cdot a(0) P_{0n}\Bigg]
       \displaybreak[0]  \nonumber\\
   &=      a(0)\, \overleftarrow{ \prod\limits_{k=1}^{n-1}} R_{n k}\bigg(\xi_n-\xi_k |\tau+\eta\sum_{l=1}^{k-1}\sigma_l^z\bigg)\,
        \mathsf{T}_\tau^{\sigma_n^z}\,  Z_n\, 
        \overleftarrow{ \prod\limits_{k=n+1}^\mathsf{N}}\!\! R_{n k}\bigg(\xi_n-\xi_k |\tau+\eta\sum_{\substack{l=1\\ l\not=n}}^{k-1}\sigma_l^z\bigg),
\end{align*}
where we have used the zero-weight property of the $R$-matrix, namely $[R_{nk},\sigma_n^z+\sigma_k^z]=0$, as well as the commutation relation \eqref{Dyn-op-comm}.
The second trace can be computed similarly, leading to
\begin{equation*}
 \mathrm{tr}_0\big[ \mathcal{M}_0(\xi_n)^{-1}\,\tilde{Z}_0\big]
 =\frac{1}{a(0)}
 \overrightarrow{ \prod\limits_{k=n+1}^\mathsf{N}} \!\! R_{n k}^{-1}\bigg(\xi_n-\xi_k |\tau+\eta\sum_{\substack{l=1\\ l\not=n}}^{k-1}\sigma_l^z\bigg)\,
 \tilde{Z}_n\, \mathsf{T}_\tau^{-\sigma_n^z}\,
 \overrightarrow{ \prod\limits_{k=1}^{n-1}} R_{n k}^{-1}\bigg(\xi_n-\xi_k |\tau+\eta\sum_{l=1}^{k-1}\sigma_l^z\bigg),
\end{equation*}
hence the result for $Z=\tilde{Z}^{-1}=Y$.
\end{proof}

\begin{rem}\label{rem-inh}
It follows from Lemma~\ref{lemme-inh} and from the inversion relation \eqref{inv-mon} that
\begin{equation}
  \frac{\mathcal{T}(\xi_n^{(0)})\,\mathcal{T}(\xi_n^{(1)})}{\mathsc{a}(\xi_n^{(0)})\,\mathsc{d}(\xi_n^{(1)})}
  =\frac{e^{-i\mathsf{y}\eta\mathsf{S}}\,\theta(\tau)}{\theta(\tau+\eta\mathsf{S})},
  \qquad\quad
 \frac{\overline{\mathcal{T}}^{(\kappa)}(\xi_n^{(0)})\,
          \overline{\mathcal{T}}^{(\kappa)}(\xi_n^{(1)})}{\mathsc{a}(\xi_n^{(0)})\,\mathsc{d}(\xi_n^{(1)})}
  =-\frac{e^{-i\mathsf{y}\eta\mathsf{S}}\,\theta(\tau)}{\theta(\tau+\eta\mathsf{S})}. 
\end{equation}
Similarly, building  on the inversion relation \eqref{inv-mon} and on the proof of Lemma~\ref{lemme-inh} for different choices of $Z$, $\tilde{Z}$, one can deduce several other useful identities, such as for instance the following cancellation identities, analog to those obtained in \cite{KitKMNST07} in the non-dynamical case:
\begin{equation}\label{cancel-inh}
  \mathcal{M}_{ij}(\xi_n^{(0)})\,\mathcal{M}_{ik}(\xi_n^{(1)})=0, \quad \forall i,j,k=\pm.
\end{equation}
\end{rem}

\begin{rem}
By using similar arguments as in the proof of Lemma~\ref{lemme-inh}, one can also easily show that
\begin{equation}
   \overrightarrow{ \prod_{j=1}^\mathsf{N}}\frac{\mathcal{T}^{(Y)}(\xi_j)}{\mathsc{a}(\xi_j)}
   = \overrightarrow{ \prod_{j=1}^\mathsf{N} } \left\{\mathsf{T}_\tau^{\sigma_j^z}\, Y_j\right\}.
\end{equation}
\end{rem}

Let us now formulate the solution of the quantum inverse problem for the quasi-periodic monodromy matrix \eqref{Y-mon}.

\begin{theorem}\label{th-inv-pb}
The local spin operators $E_n^{i j}$ \eqref{local-spin}, understood as operators acting on $\mathbb{D}_{\mathsf{(6VD)},\mathsf{N}}$, can be expressed
in terms of the entries of the quasi-periodic monodromy matrix \eqref{Y-mon} or of its inverse in the following way:
  \begin{align}
    E_n^{i j}   &=
     \overrightarrow{\prod_{k=1}^{n-1}} \mathcal{T}^{(Y)}(\xi_k) \cdot   
     \big[\mathcal{M}^{(Y)} (\xi_n) \big]_{ji}
        \cdot \overleftarrow{\prod_{k=1}^n}\big[ \mathcal{T}^{(Y)}(\xi_k) \big]^{-1} 
        \cdot\, \mathsf{T}_{\tau}^{j-i},
        \label{inv-pb1}\\
                     &=      
     \overrightarrow{\prod_{k=1}^{n}} \mathcal{T}^{(Y)}(\xi_k) \cdot   
     \big[ \mathcal{M}^{(Y)} (\xi_n)^{-1} \big]_{ji}
        \cdot \overleftarrow{\prod_{k=1}^{n-1}}\big[ \mathcal{T}^{(Y)}(\xi_k) \big]^{-1} 
        \cdot\, \mathsf{T}_{\tau}^{j-i}.
        \label{inv-pb2}
  \end{align}
\end{theorem}

\begin{rem}
The reconstructions of Theorem~\ref{th-inv-pb} are valid on the whole representation space $\mathbb{D}_{\mathsf{(6VD)},\mathsf{N}}$, on which the transfer matrices do not {\em a priori} commute. Hence, we have to pay attention to the order in the corresponding products.
\end{rem}

\begin{rem}\label{rem-IPinv}
In \cite{LevT13a} was only formulated the analog of \eqref{inv-pb1} in the periodic case. The relation \eqref{inv-pb2} is instead useful to express local operators in terms of elements of the monodromy matrix with shifted inhomogeneity parameters.
In fact, using the inversion relation \eqref{inv-mon} for the monodromy matrix, the relation \eqref{inv-pb2} can, in the periodic or antiperiodic case,  be respectively rewritten as
 \begin{align}
  E_n^{i j} 
  &=(-1)^{\frac{j-i}{2}}\,
  \overrightarrow{\prod_{k=1}^{n}} {\mathcal{T}}(\xi_k) \cdot   
  \frac{\big[{\mathcal{M}} (\xi_n^{(1)}) \big]_{-i,-j}}{\mathsc{a}(\xi_n^{(0)})\,\mathsc{d}(\xi_n^{(1)})}\,
  e^{i\mathsf{y}\eta\mathsf{S}}\frac{\theta(\tau+\eta\mathsf{S})}{\theta(\tau)}
  \cdot \overleftarrow{\prod_{k=1}^{n-1}}\big[ {\mathcal{T}}(\xi_k) \big]^{-1} 
        \cdot\, \mathsf{T}_{\tau}^{j-i}.\label{inv-pb3bis}
  \\
  &=-(-1)^{\frac{j-i}{2}}\,
  \overrightarrow{\prod_{k=1}^{n}} \overline{\mathcal{T}}^{(\kappa)}(\xi_k) \cdot   
  \frac{\big[\overline{\mathcal{M}}^{(\kappa)} (\xi_n^{(1)}) \big]_{-i,-j}}{\mathsc{a}(\xi_n^{(0)})\,\mathsc{d}(\xi_n^{(1)})}\,
  e^{i\mathsf{y}\eta\mathsf{S}}\frac{\theta(\tau+\eta\mathsf{S})}{\theta(\tau)}
  \cdot \overleftarrow{\prod_{k=1}^{n-1}}\big[ \overline{\mathcal{T}}^{(\kappa)}(\xi_k) \big]^{-1} 
        \cdot\, \mathsf{T}_{\tau}^{j-i}.\label{inv-pb3}
 \end{align}
\end{rem}

\begin{proof}
Let us first show \eqref{inv-pb1} for $n=1$. The proof is based, as usual (see \cite{KitMT99,MaiT00} and the proof of Lemma~\ref{lemme-inh}), on the crucial fact that the $R$-matrix \eqref{mat-R} becomes proportional to the permutation operator $P_{12}$ when evaluated at $\lambda=0$.
Expressing as in Lemma~\ref{lemme-inh} the matrix element $[\mathcal{M}^{(Y)}(\xi_1)]_{ji}$ as a trace $\mathrm{tr}_0\big[\mathcal{M}_0^{(Y)}(\xi_1)\, E_0^{ij}\big]$ over some auxiliary space 0, representing $\mathcal{M}_0^{(Y)}(\xi_1)$ in terms of a product of $R$-matrices and 
%
%
moving in this expression the operator $E_0^{ij}$ from right to left, using successively that $\mathsf{T}_\tau^{\sigma_0^z} E_0^{ij}=E_0^{ij}\mathsf{T}_\tau^{\sigma_0^z+i-j}$, that $P_{01} E_0^{ij}=E_1^{ij}P_{01}$ and that $(\tau+\eta \sigma_1^z) E_1^{ij} \mathsf{T}_\tau^{i-j}=E_1^{ij} \mathsf{T}_\tau^{i-j}(\tau+\eta \sigma_1^z)$, we get
\begin{equation}\label{reconst2}
\big[\mathcal{M}^{(Y)}(\xi_1)\big]_{ji}
 = E_1^{ij}\, \mathsf{T}_\tau^{i-j}\, \mathcal{T}^{(Y)}(\xi_1).
\end{equation}
The general case can be deduced from the case $n=1$ by means of Lemma~\ref{lemme-mon}.
Let us express part of the product in the left hand side of \eqref{inv-pb1} as a multiple trace over auxiliary spaces:
 \begin{align}
  \overrightarrow{\prod_{k=1}^{n-1} } \mathcal{T}^{(Y)}(\xi_k) \cdot 
 \big[\mathcal{M}^{(Y)} (\xi_n)\big]_{ji}
 &= \mathrm{tr}_{a_1 a_2\ldots a_n} \Bigg[ \,
 \overrightarrow{\prod_{k=1}^{n}} \left( Y_{a_k}\,  \mathsf{M}_{a_k}(\xi_k|\tau)\, \mathsf{T}_\tau^{\sigma_{a_k}^z} \right) \, E_{a_n}^{ij} \Bigg]
 \nonumber\\
 &= \mathrm{tr}_{a_1 a_2\ldots a_n} \Bigg[ \,
 {\prod_{k=1}^{n}} Y_{a_k} \,
 \overrightarrow{\prod_{k=1}^{n}}  \mathsf{M}_{a_k}\bigg(\xi_k|\tau+\eta\sum_{l=1}^{k-1}\sigma^z_{a_l}\bigg)\,
 {\prod_{k=1}^{n}}  \mathsf{T}_\tau^{\sigma_{a_k}^z}  \cdot
  E_{a_n}^{ij} \Bigg] \nonumber
 \end{align}
 in which we have used \eqref{Dyn-op-comm}.
 The product of monodromy matrices can now be rewritten as in Lemma~\ref{lemme-mon}, and reorganized thanks to \eqref{Dyn-op-comm} such that
 \begin{multline}
 \overrightarrow{\prod_{k=1}^{n-1}} \mathcal{T}^{(Y)}(\xi_k)\cdot 
 \big[\mathcal{M}^{(Y)} (\xi_n)\big]_{ji}
  = \mathrm{tr}_{a_1 a_2\ldots a_{n-1}} \Bigg[ \,
 {\prod_{k=1}^{n-1}} Y_{a_k} 
 \\
  \times
 \mathrm{tr}_{a_n}\bigg[ 
 \mathsf{M}^{(Y)}_{a_n, n\, n+1\ldots N a_1 a_2\ldots a_{n-1}} \bigg(\xi_n|\tau+\eta\sum_{l=1}^{n-1}\sigma^z_l \bigg)\,\mathsf{T}_\tau^{\sigma_{a_n}^z} \,  E_{a_n}^{ij}\bigg]  
 \\
 \times
 \overleftarrow{ \prod\limits_{k=1}^{n-1}} 
 \mathsf{M}_{a_k, k\, k+1\ldots N a_1 a_2\ldots a_{k-1}} \bigg(\xi_k|\tau+\eta\sum_{l=1}^{k-1}\sigma^z_l+\eta\sum_{l=k+1}^{n-1} \sigma^z_{a_l} \bigg) \ 
 {\prod_{k=1}^{n-1}}  \mathsf{T}_\tau^{\sigma_{a_k}^z} 
 \Bigg]. \label{reconst4}
 \end{multline}
 The trace over $a_n$ can now be explicitly computed similarly as in \eqref{reconst2}:
 \begin{multline}
 \mathrm{tr}_{a_n}\bigg[ 
 \mathsf{M}^{(Y)}_{a_n, n\, n+1\ldots N a_1 a_2\ldots a_{n-1}} \bigg(\xi_n|\tau+\eta\sum_{l=1}^{n-1}\sigma^z_l \bigg)\,\mathsf{T}_\tau^{\sigma_{a_n}^z} \,  E_{a_n}^{ij}\bigg]  
  \\
 = E_n^{ij}\, \mathsf{T}_\tau^{i-j}\, 
 \mathrm{tr}_{a_n}\bigg[ 
 \mathsf{M}^{(Y)}_{a_n, n\, n+1\ldots N a_1 a_2\ldots a_{n-1}} \bigg(\xi_n|\tau+\eta\sum_{l=1}^{n-1}\sigma^z_l \bigg)\, \mathsf{T}_\tau^{\sigma_{a_n}^z} \bigg],
 \end{multline}
 so that the product $E_n^{ij}\, \mathsf{T}^{i-j}$ can be moved out of the trace from the left in \eqref{reconst4}. The remaining multiple trace can then be re-expressed as a product of transfer matrices (using again Lemma~\ref{lemme-mon}), leading to \eqref{inv-pb1}.
 
 The proof of \eqref{inv-pb2} can be performed in a similar way. Considering first the case $n=1$ we obtain, as in \eqref{reconst2}, that
 \begin{equation}
   \mathrm{tr}_0 \big[ E_0^{ij}\, \mathcal{M}_0^{(Y)}(\xi_1) ^{-1}\big]
   = \mathrm{tr}_0 \big[ \mathcal{M}_0^{(Y)}(\xi_1) ^{-1}\big]\, \mathsf{T}_\tau^{i-j}\, E_1^{ij}
   =\mathcal{T}^{(Y)}(\xi_1)^{-1}\, \mathsf{T}_\tau^{i-j}\, E_1^{ij},
 \end{equation}
 where we have used Lemma~\ref{lemme-inh}.
 The general case can be proven by means of Lemma~\ref{lemme-mon} and Lemma~\ref{lemme-inh}, by computing 
 \begin{equation}
 \big[ \mathcal{M}^{(Y)} (\xi_n)^{-1} \big]_{ji}
        \cdot \overleftarrow{\prod_{k=1}^{n-1}}\big[ \mathcal{T}^{(Y)}(\xi_k) \big]^{-1}
       = \mathrm{tr}_{a_1 a_2\ldots a_n}\Bigg[\, E_{a_n}^{ij}\,
 \overleftarrow{\prod_{k=1}^{n}} \left( Y_{a_k}\, \mathsf{M}_{a_k}(\xi_k|\tau)\, \mathsf{T}_\tau^{\sigma_{a_k}^z} \right)^{-1} \Bigg].
 \end{equation}
 Reorganizing the product inside the trace and using identity \eqref{id-T2} for inverse monodromy matrices, we obtain, by a similar reasoning as in the previous case, that
 \begin{align}
 \big[ \mathcal{M}^{(Y)} (\xi_n)^{-1} \big]_{ji}
        \cdot \overleftarrow{\prod_{k=1}^{n-1}}\big[ \mathcal{T}^{(Y)}(\xi_k) \big]^{-1}
        &= \overleftarrow{\prod_{k=1}^{n}} \mathrm{tr}_{a_k}\big[\, \mathcal{M}^{(Y)} (\xi_k)^{-1} \big]\cdot \mathsf{T}_\tau^{i-j}\, E_n^{i-j}
        \nonumber\\
        &= \overleftarrow{\prod_{k=1}^{n}}\big[ \mathcal{T}^{(Y)}(\xi_k) \big]^{-1}
        \cdot \mathsf{T}_\tau^{i-j}\, E_n^{i-j},
 \end{align}
 which ends the proof of Theorem~\ref{th-inv-pb}.
\end{proof}

As a consequence, one gets

\begin{corollary}
The local spin operators $E_n^{++}$, $E_n^{--}$, and the local height operators $\widehat{\delta}_t^{(n)}$, $1\le n\le \mathsf{N}$, admit in $\mathbb{\bar{D}}_{\mathsf{(6VD)},\mathsf{N}}^{(0,\mathcal{L}/\mathcal{R})}$ the following reconstruction in terms of the entries of the $\kappa$-twisted antiperiodic monodromy matrix:
\begin{align}
 &E_n^{++}
 = \prod_{k=1}^{n-1} \overline{\mathcal{T}}^{(\kappa)}(\xi _k )\cdot
      \kappa\,\mathcal{C}(\xi_n)\cdot
      \prod_{k=1}^{n}\big[ \overline{\mathcal{T}}^{(\kappa)}(\xi _k  )\big]^{-1}, \label{IP1}\\
 &\hphantom{E_n^{++}}
 = - \prod_{k=1}^{n} \overline{\mathcal{T}}^{(\kappa)}(\xi _k)\cdot   
       \frac{\kappa^{-1}\mathcal{B}(\xi_n-\eta)}{\mathrm{det}_q \mathsf{M}(\xi_n)}\cdot
       \prod_{k=1}^{n-1}\big[ \overline{\mathcal{T}}^{(\kappa)}(\xi _k  )\big]^{-1}, \label{IP1bis}
 \displaybreak[0]\\
 &E_n^{--}
 = \prod_{k=1}^{n-1} \overline{\mathcal{T}}^{(\kappa)}(\xi _k )\cdot
      \kappa^{-1}\mathcal{B}(\xi_n)\cdot
      \prod_{k=1}^{n}\big[ \overline{\mathcal{T}}^{(\kappa)}(\xi _k  )\big]^{-1}, \label{IP2}\\
 &\hphantom{E_n^{--}}
 = - \prod_{k=1}^{n} \overline{\mathcal{T}}^{(\kappa)}(\xi _k )\cdot   
       \frac{\kappa\,\mathcal{C}(\xi_n-\eta)}{\mathrm{det}_q \mathsf{M}(\xi_n)}\cdot
       \prod_{k=1}^{n-1}\big[ \overline{\mathcal{T}}^{(\kappa)}(\xi _k  )\big]^{-1}, \label{IP2bis}
 \displaybreak[0]\\  
 &\widehat{\delta}_t^{(n)}
 = \prod_{k=1}^{n-1} \overline{\mathcal{T}}^{(\kappa)}(\xi _k )\cdot \delta_t(\tau)\cdot
    \prod_{k=1}^{n-1}\big[ \overline{\mathcal{T}}^{(\kappa)}(\xi _k  )\big]^{-1}.
    \label{IP3}
\end{align}
\end{corollary}

\begin{proof}
The representations \eqref{IP1}-\eqref{IP2bis} follow directly from \eqref{inv-pb1}-\eqref{inv-pb2} taking into account Remark~\ref{rem-IPinv} and the fact that we restrict ourselves to $\mathbb{\bar{D}}_{\mathsf{(6VD)},\mathsf{N}}^{(0,\mathcal{L}/\mathcal{R})}$.
The relation \eqref{IP3} follows, as in Theorem~3.2 of \cite{LevT14a}, from a trivial recursion on $n$ using the fact that
\begin{equation}
   \widehat{\delta}_t^{(n)}= \widehat{\delta}_{t-1}^{(n-1)}\, E_{n-1}^{++}+ \widehat{\delta}_{t+1}^{(n-1)}\, E_{n-1}^{--}
\end{equation}
and the solution of the inverse problem for $E_{n-1}^{++}$ and $E_{n-1}^{--}$.
\end{proof}

\section{Form factors of local operators}
\label{sec-ff}

We are now in position to compute matrix elements of local operators (spin and height) between eigenstates of the $\kappa$-twisted antiperiodic transfer matrix. 
We obtain the following result:

\begin{theorem}\label{th-LS}
The matrix elements of $E_{n}^{++}$ and $E_{n}^{--}$ between generic $\langle \Psi_{\bar{\mathsf{t}}}^{(\kappa)}|$\ and $|\Psi_{\bar{\mathsf{t}}^{\prime}}^{(\kappa) }\rangle $ left and right eigenstates of  $\mathcal{\overline{T}}^{(\kappa)}(\lambda )$ on $\mathbb{\bar{D}}_{\mathsf{(6VD)},\mathsf{N}}^{(0,\mathcal{L}/\mathcal{R})}$ admit the following determinant representations:
\begin{align}
 &\langle \Psi_{\bar{\mathsf{t}}} ^{(\kappa)} | E_{n}^{++} |\Psi_{\bar{\mathsf{t}}'}^{ (\kappa) }\rangle 
 =\frac{\prod_{b=1}^{n-1}\bar{\mathsf{t}}(\xi _{b})}{\prod_{b=1}^{n}\bar{\mathsf{t}}^{\prime}(\xi _{b})}\,
    \det_{N+1} \big[ \mathcal{S}_{\bar{\mathsf{t}}, \bar{\mathsf{t}}' }(\xi_n)\big]
 =-\frac{\prod_{b=1}^{n}\bar{\mathsf{t}}(\xi _{b})}{\prod_{b=1}^{n-1}\bar{\mathsf{t}}^{\prime}(\xi _{b})}\, 
   \frac{\det_{N+1} \big[ \mathcal{S}_{\bar{\mathsf{t}}', \bar{\mathsf{t}} }(\xi_n-\eta)\big]  }
           {\det_{q}\mathsf{M}(\xi _{n})},  \label{ME1} 
 \\
 &\langle \Psi_{\bar{\mathsf{t}}} ^{(\kappa)} | E_{n}^{--} |\Psi_{\bar{\mathsf{t}}'}^{ (\kappa) }\rangle  
 =\frac{\prod_{b=1}^{n-1}\bar{\mathsf{t}}(\xi _{b})}{\prod_{b=1}^{n}\bar{\mathsf{t}}^{\prime}(\xi _{b})}\,
    \det_{N+1} \big[ \mathcal{S}_{\bar{\mathsf{t}}', \bar{\mathsf{t}}}(\xi_n)\big]
 =-\frac{\prod_{b=1}^{n}\bar{\mathsf{t}}(\xi _{b})}{\prod_{b=1}^{n-1}\bar{\mathsf{t}}^{\prime}(\xi _{b})}\, 
  \frac{  \det_{N+1} \big[ \mathcal{S}_{\bar{\mathsf{t}}, \bar{\mathsf{t}}' }(\xi_n-\eta)\big] }
           {\det_{q}\mathsf{M}(\xi _{n})}, \label{ME2} 
\end{align}
where $\mathcal{S}_{\bar{\mathsf{t}}, \bar{\mathsf{t}}' }(\xi_n^{(\epsilon)})$, $\epsilon\in\{0,1\}$, is an $(\mathsf{N}+1)\times(\mathsf{N}+1)$ matrix which corresponds to the matrix $\mathcal{F}_{\bar{\mathsf{t}}, \bar{\mathsf{t}}' }$ \eqref{mat-F} of the scalar product \eqref{sp-eigen} with one additional line and one additional column:
\begin{alignat}{2}
  &\big[ \mathcal{S}_{\bar{\mathsf{t}}, \bar{\mathsf{t}}' }(\xi_n^{(\epsilon)})\big]_{a,b}
    = \big[\mathcal{F}_{\bar{\mathsf{t}}, \bar{\mathsf{t}}' }\big]_{a,b},
  & \quad &\text{for}\quad a,b\in\{1,\ldots,\mathsf{N}\},
  \label{mat-S-1}\\
  &\big[ \mathcal{S}_{\bar{\mathsf{t}}, \bar{\mathsf{t}}' }(\xi_n^{(\epsilon)})\big]_{a,\mathsf{N}+1}
  = e^{i\mathsf{y}\xi _{a}}\,\mathsc{a}_\mathsf{x,y}(\xi _{a})\,
     \mathsf{q}_{\bar{\mathsf{t}},a}^{(0)}\, \mathsf{q}_{\bar{\mathsf{t}}^{\prime },a}^{(1)},
     &\qquad &\text{for}\quad a\in\{1,\ldots,\mathsf{N}\},
  \\
  &\big[ \mathcal{S}_{\bar{\mathsf{t}}, \bar{\mathsf{t}}' }(\xi_n^{(\epsilon)})\big]_{\mathsf{N}+1,b}
  =- e^{-i\mathsf{y} \xi_n^{(\epsilon)}}\, \vartheta_{b-1}(\xi_n^{(\epsilon)}-\bar{\xi}_0),
  &\qquad &\text{for}\quad b\in\{1,\ldots,\mathsf{N}\},
  \\
  &\big[ \mathcal{S}_{\bar{\mathsf{t}}, \bar{\mathsf{t}}' }(\xi_n^{(\epsilon)})\big]_{\mathsf{N}+1,\mathsf{N}+1}
  = 0.\label{mat-S-4}
\end{alignat}
\end{theorem}

The representations \eqref{ME1} and \eqref{ME2} straightforwardly follow from the solution of the quantum inverse problem \eqref{IP1}-\eqref{IP2bis}, which enables us to write
\begin{align*}
 &\langle \Psi_{\bar{\mathsf{t}}} ^{(\kappa)} | E_{n}^{++} |\Psi_{\bar{\mathsf{t}}'}^{ (\kappa) }\rangle 
 =\kappa\,\frac{\prod_{b=1}^{n-1}\bar{\mathsf{t}}(\xi _{b})}{\prod_{b=1}^{n}\bar{\mathsf{t}}^{\prime}(\xi _{b})}\,
   \langle \Psi_{\bar{\mathsf{t}}} ^{(\kappa)} | \,\mathcal{C}(\xi _{n} )\, |\Psi_{\bar{\mathsf{t}}'}^{ (\kappa) }\rangle 
 =-\kappa^{-1}\frac{\prod_{b=1}^{n}\bar{\mathsf{t}}(\xi _{b})}{\prod_{b=1}^{n-1}\bar{\mathsf{t}}^{\prime}(\xi _{b})}\, 
   \frac{\langle \Psi_{\bar{\mathsf{t}}} ^{(\kappa)} |\,\mathcal{B}(\xi _{n}-\eta )\,|\Psi_{\bar{\mathsf{t}}'}^{ (\kappa) }\rangle  }
           {\det_{q}\mathsf{M}(\xi _{n})},  
 \\
 &\langle \Psi_{\bar{\mathsf{t}}} ^{(\kappa)} | E_{n}^{--} |\Psi_{\bar{\mathsf{t}}'}^{ (\kappa) }\rangle  
 =\kappa^{-1}\frac{\prod_{b=1}^{n-1}\bar{\mathsf{t}}(\xi _{b})}{\prod_{b=1}^{n}\bar{\mathsf{t}}^{\prime}(\xi _{b})}\,
   \langle \Psi_{\bar{\mathsf{t}}} ^{(\kappa)} | \,\mathcal{B}(\xi _{n}  )\, |\Psi_{\bar{\mathsf{t}}'}^{ (\kappa) }\rangle 
 =-\kappa\,\frac{\prod_{b=1}^{n}\bar{\mathsf{t}}(\xi _{b})}{\prod_{b=1}^{n-1}\bar{\mathsf{t}}^{\prime}(\xi _{b})}\, 
  \frac{ \langle \Psi_{\bar{\mathsf{t}}} ^{(\kappa)} |\,\mathcal{C}(\xi _{n}-\eta)\,|\Psi_{\bar{\mathsf{t}}'}^{ (\kappa) }\rangle  }
           {\det_{q}\mathsf{M}(\xi _{n})}, 
\end{align*}
and from the following lemma:

\begin{lemma}
The matrix elements of the operators $\mathcal{B}(\xi _{n}^{(\epsilon)}  )$ and $\mathcal{C}(\xi _{n}^{(\epsilon)}  )$, $\epsilon\in\{0,1\}$, between eigenstates  $\langle \Psi_{\bar{\mathsf{t}}}^{(\kappa)}|$\ and $|\Psi_{\bar{\mathsf{t}}^{\prime}}^{(\kappa) }\rangle $ of the antiperiodic transfer matrix are given by the following determinants:
\begin{equation}\label{mat-BC}
 \langle \Psi_{\bar{\mathsf{t}}} ^{(\kappa)} | \,\mathcal{B}(\xi _{n}^{(\epsilon)}  )\, |\Psi_{\bar{\mathsf{t}}'}^{ (\kappa) }\rangle 
 =\kappa\, \det_{N+1} \big[ \mathcal{S}_{\bar{\mathsf{t}}', \bar{\mathsf{t}} }(\xi_n^{(\epsilon)})\big],
 \qquad
 \langle \Psi_{\bar{\mathsf{t}}} ^{(\kappa)} |\, \mathcal{C}(\xi _{n}^{(\epsilon)}  )\, |\Psi_{\bar{\mathsf{t}}'}^{ (\kappa) }\rangle 
 = \kappa^{-1} \det_{N+1} \big[ \mathcal{S}_{\bar{\mathsf{t}}, \bar{\mathsf{t}}' }(\xi_n^{(\epsilon)})\big],
\end{equation} 
where $\mathcal{S}_{\bar{\mathsf{t}}, \bar{\mathsf{t}}' }(\xi_n^{(\epsilon)})$, $\epsilon\in\{0,1\}$, is the $(\mathsf{N}+1)\times(\mathsf{N}+1)$ matrix with elements \eqref{mat-S-1}-\eqref{mat-S-4}.
\end{lemma}

\begin{proof}
From \eqref{eigenT-l}, \eqref{C-SOV_D-left}, one can easily compute the
action of $\mathcal{C}(\xi _{n}^{(\epsilon )})$, $\epsilon =0,1$, on a left $\overline{\mathcal{T}}^{(\kappa )}$-eigenstate: 
\begin{multline}
    \langle \Psi_{\bar{\mathsf{t}}} ^{(\kappa)} | \,\mathcal{C}(\xi _{n}^{(\epsilon)})
    =\sum_{a=1}^{\mathsf{N}}
    \sum_{\substack{\mathbf{h}\in \{0,1\}^{\mathsf{N}} \\ h_a=0}}
    e^{i\mathsf{y}(\xi _{a}^{(h_{a})}-\xi _{n}^{(\epsilon )})}\,
    \frac{\theta (t_{0,\mathbf{h}}-\xi _{n}^{(\epsilon )}+\xi _{a}^{(h_{a})})}{\theta (t_{0,\mathbf{h}})}\,
    \prod_{b\neq a}
    \frac{\theta (\xi _{n}^{(\epsilon )}-\xi_{b}^{(h_{b})})}{\theta (\xi _{a}^{(h_{a})}-\xi _{b}^{(h_{b})})}
     \\
\times 
    \mathsc{d}(\xi _{a}^{(1-h_{a})})\,
    \prod_{b=1}^{\mathsf{N}}\left( 
    e^{i\mathsf{y}\eta h_{b}}\,\kappa ^{h_{b}}\,\mathsf{q}_{\bar{\mathsf{t}},b}^{(h_{b})}\right) \,
    \det_{\mathsf{N}}\big[\Theta ^{(0,\mathbf{h})}\big]\,
    \langle 0,\mathsf{T}_{a}^{+}\mathbf{h}|\,.
\end{multline}
Then, using \eqref{eigenT-r} and \eqref{def-theta_j}, one gets 
\begin{multline}
    \langle \Psi_{\bar{\mathsf{t}}} ^{(\kappa)} | \,\mathcal{C}(\xi _{n}^{(\epsilon )})\,
    |\Psi_{\bar{\mathsf{t}}'}^{ (\kappa) }\rangle
     =\kappa ^{-1}
     \sum_{a=1}^{\mathsf{N}}
      \sum_{\substack{\mathbf{h}\in \{0,1\}^{\mathsf{N}} \\ h_a=0}}
     e^{i\mathsf{y}(\xi _{a}-\xi _{n}^{(\epsilon)})}\,
     \mathsc{a}(\xi _{a})\,
     \mathsf{q}_{\bar{\mathsf{t}},a}^{(0)}\,\mathsf{q}_{\bar{\mathsf{t}}^{\prime },a}^{(1)}\, 
     \\
\times 
     \prod_{b\neq a}^{\mathsf{N}}\left[ \left( 
     \frac{e^{i\mathsf{y}\eta }\,\mathsc{a}_\mathsf{x,y}(\xi _{b})}
            {\mathsc{d}(\xi _{b}-\eta)}\right) ^{\!h_b}\,
            \mathsf{q}_{\bar{\mathsf{t}},b}^{(h_{b})}\,\mathsf{q}_{\bar{\mathsf{t}}^{\prime },b}^{(h_{b})}
            \right]\,
            (-1)^{\mathsf{N}+a}\,
            \det_{\mathsf{N}}\Big[\Theta _{\lbrack \widehat{a},\mathsf{N}]}^{(0,\mathbf{h})}(\xi _{n}^{(\epsilon )})\Big]\,. 
             \label{sum}
\end{multline}
In this expression, 
the $\mathsf{N}\times \mathsf{N}$ matrix $\Theta _{\lbrack \widehat{a},\mathsf{N}]}^{(0,\mathbf{h})}(\xi _{n}^{(\epsilon )})$ is obtained from $\Theta ^{(0,\mathbf{h})}$ by eliminating the $a$-th row (containing the elements $\vartheta_{b-1}(\xi _{a}^{(h_{a})}-\bar{\xi}_{0})$, $1\leq b\leq \mathsf{N}$ ) and
inserting a new row at position $\mathsf{N}$ with elements $\vartheta_{b-1}(\xi _{n}^{(\epsilon )}-\bar{\xi}_{0})$, $1\leq b\leq \mathsf{N}$.
Indeed, it follows from \eqref{det-thetaj} that 
\begin{equation}
      (-1)^{\mathsf{N}+a}\,
      \det_{\mathsf{N}}\Big[\Theta _{\lbrack \widehat{a},\mathsf{N}]}^{(0,\mathbf{h})}(\xi _{n}^{(\epsilon )})\Big]
      =\frac{\theta(t_{0,\mathbf{h}}-\xi _{n}^{(\epsilon )}+\xi _{a}^{(h_{a})})}
                {\theta (t_{0,\mathbf{h}})}\,
      \prod_{b\neq a}\frac{\theta (\xi _{n}^{(\epsilon )}-\xi_{b}^{(h_{b})})}
                                          {\theta (\xi _{a}^{(h_{a})}-\xi _{b}^{(h_{b})})}\,
      \det_{\mathsf{N}}\big[\Theta ^{(0,\mathbf{h})}\big].
\end{equation}
It remains to notice that the sum in \eqref{sum} corresponds precisely to
the development of the determinant of the $(\mathsf{N}+1)\times (\mathsf{N}+1)$ matrix $\mathcal{S}_{\bar{\mathsf{t}},\bar{\mathsf{t}}^{\prime }}(\xi _{n}^{(\epsilon )})$ w.r.t. the column $\mathsf{N}+1$.
The proof for the other formula in \eqref{mat-BC} is  similar.
\end{proof}

\begin{theorem}\label{th-LH}
 The matrix elements of local heights operators $\widehat{\delta}_s^{(n)}$ fixing the value of the height $s\in\big\{t_0+\eta\tilde{s};\tilde{s}\in\{0,1,\ldots,\mathsf{N}\}\big\}$ at a given site $n$ between generic $\langle \Psi_{\bar{\mathsf{t}}} ^{(\kappa)} | $\ and $|\Psi_{\bar{\mathsf{t}}'}^{ (\kappa) }\rangle $ left and right eigenstates of  $\mathcal{\overline{T}}^{(\kappa)}(\lambda  )$ on $\mathbb{\bar{D}}_{\mathsf{(6VD)},\mathsf{N}}^{(0,\mathcal{L}/\mathcal{R})}$ can be written as the following sum of $\mathsf{N}+1$ determinants:
 \begin{equation}\label{ff-LH}
  \langle \Psi_{\bar{\mathsf{t}}} ^{(\kappa)} | \, \widehat{\delta}_s^{(n)}\, |\Psi_{\bar{\mathsf{t}}'}^{ (\kappa) }\rangle
 =\prod_{b=1}^{n-1}\frac{\bar{\mathsf{t}}(\xi _{b})}{\bar{\mathsf{t}}^{\prime}(\xi _{b})}\,
   \frac{1}{\mathsf{N}+1} \sum_{j=0}^\mathsf{N} e^{-2\pi i\frac{j(s-t_{0})}{\eta(\mathsf{N}+1)}}\,
   \det_{\mathsf{N}}\big[\widetilde{\mathcal{F}}_{\bar{\mathsf{t}},\bar{\mathsf{t}}^{\prime } }^{(j)}\big],
 \end{equation}
 where $\widetilde{\mathcal{F}}_{\bar{\mathsf{t}},\bar{\mathsf{t}}^{\prime } }^{(j)}$ is the $\mathsf{N}\times\mathsf{N}$ matrix with elements
\begin{equation}\label{mat-LH}
\big[\widetilde{\mathcal{F}}_{\bar{\mathsf{t}},\bar{\mathsf{t}}^{\prime } }^{(j)}\big]_{a,b}
=
\sum_{h=0}^{1}  e^{2\pi i \frac{j\, h}{\mathsf{N}+1}} 
\left( e^{i\mathsf{y}\eta }\,
        \frac{\mathsc{a}_\mathsf{x,y}(\xi_{a})}{\mathsc{d}(\xi_{a}-\eta)}\right)^{\! h}
        \mathsf{q}_{\bar{\mathsf{t}},a}^{(h)}\,
        \mathsf{q}_{\bar{\mathsf{t}}^{\prime },a}^{(h)}\ 
        \vartheta _{b-1}(\xi _{a}^{(h)}-\bar{\xi}_{0}).
\end{equation}
\end{theorem}

\begin{proof}
Using \eqref{IP3}, \eqref{eigenT-r}, \eqref{eigenT-l} and \eqref{sc-rh}, one obtains
 \begin{equation}
  \langle \Psi_{\bar{\mathsf{t}}} ^{(\kappa)} | \, \widehat{\delta}_s^{(n)}\, |\Psi_{\bar{\mathsf{t}}'}^{ (\kappa) }\rangle 
 =\prod_{b=1}^{n-1}\frac{\bar{\mathsf{t}}(\xi _{b})}{\bar{\mathsf{t}}^{\prime}(\xi _{b})}\,
 \sum_{\mathbf{h}\in\{0,1\}^\mathsf{N}} \delta_{s,t_{0,\mathbf{h}}}
 \prod_{a=1}^\mathsf{N}\left[
 \left( e^{i\mathsf{y}\eta }\,
        \frac{\mathsc{a}_\mathsf{x,y}(\xi_{a})}{\mathsc{d}(\xi_{a}-\eta)}\right)^{\! h_a}
        \mathsf{q}_{\bar{\mathsf{t}},a}^{(h_a)}\,
        \mathsf{q}_{\bar{\mathsf{t}}^{\prime },a}^{(h_a)}
        \right]\,
\det_{\mathsf{N}}\big[\Theta^{(0, \mathbf{h}) }\big].
 \end{equation}
 Noticing that $t_{0,\mathbf{h}}\in\big\{t_0+\eta k; k\in\{0,1,\ldots,\mathsf{N}\}\big\}$, one can rewrite $\delta_{s,t_{0,\mathbf{h}}}$ as
 \begin{equation}
  \delta_{s,t_{0,\mathbf{h}}}
  =\frac{1}{\mathsf{N}+1} \sum_{j=0}^\mathsf{N} e^{-2\pi i\frac{j(s-t_{0,\mathbf{h}})}{\eta(\mathsf{N}+1)}}
  =\frac{1}{\mathsf{N}+1} \sum_{j=0}^\mathsf{N} e^{-2\pi i\frac{j(s-t_{0,\mathbf{0}})}{\eta(\mathsf{N}+1)}}
     \prod_{a=1}^\mathsf{N}e^{\frac{2\pi i j h_a}{\mathsf{N}+1}},
 \end{equation}
which leads to \eqref{ff-LH}-\eqref{mat-LH}.
\end{proof}

To conclude this section, let us briefly comment about these results. Although the matrix elements of local spin operators can quite straightforwardly be expressed in terms of a single determinant as in Theorem~\ref{th-LS}, hence generalizing the simpler (non-dynamical) six-vertex case \cite{Nic13}, the situation seems slightly more complicated for the matrix elements of local height operators. The latter can nevertheless be expressed as a sum of determinants as in Theorem~\ref{th-LH}. The number of terms of this sum being related to the size $\mathsf{N}$ of the model, this may (or may not, depending on the behavior of the different terms) be a problem for the study of the thermodynamic limit. Note however that our study concerns the {\em unrestricted} SOS model, for which the heights $s$ are {\em a priori} allowed to take an infinite number of values\footnote{In fact, due to the antiperiodic boundary conditions that we consider here, the heights of the model are only allowed to take $\mathsf{N}+1$ values. The consideration of such boundary conditions indeed reduces the actual space of states of the model to the finite-dimensional subspace $\mathbb{\bar{D}}_{\mathsf{(6VD)},\mathsf{N}}^{(0,\mathcal{L}/\mathcal{R})}$ of the whole representation space $\mathbb{D}_{\mathsf{(6VD)},\mathsf{N}}^{\mathcal{L}/\mathcal{R}}$.}. In the literature, one usually consider {\em restricted} models, such as the ABF model \cite{AndBF84} or the CSOS model \cite{KunY88,PeaS88}, for which the crossing parameter $\eta$ of the model is rational ($\eta=r/L$) and the heights $s$ are only allowed to take a finite number ($L$) of values. Our present study does not directly apply to the ABF case (which corresponds to the case for which there may be poles in \eqref{def-abc}) but could easily be adapted to the study of the CSOS case. In the latter case, the matrix elements of local height probabilities would be reduced to the sum of only $L$ terms. Although the structure of the determinants at stake are {\em a priori} quite different, the situation is somehow similar to what happens for the periodic model, which can be studied by means of ABA: matrix elements of local spin operators between Bethe eigenstates of the transfer matrix can be expressed (at least in the case of the CSOS model) as a single determinant \cite{LevT13a}, but it seems that matrix elements of local height operators, related to local height probabilities, can only be expressed as sums of determinants \cite{LevT14a}.

\section{Conclusion}

We have here studied the antiperiodic dynamical 6-vertex model in the SOV framework, for different configurations of the representation space corresponding to the different possible values of the global shift $t_0$ of the heights of the model by half of the periods of the theta function, associated with a couple of parameters $\mathsf{(x,y)}\in\{0,1\}^2$.
We have diagonalized the corresponding antiperiodic transfer matrix, hence obtaining a complete characterization of all eigenvalues and eigenstates in terms of a system of discrete equations involving the inhomogeneity parameters of the model. We have discussed the rewriting of this characterization in terms of functional equations of Baxter's type, and notably in terms of certain classes of solutions of the usual homogeneous function $T$-$Q$ equations. We have also obtained determinant representations for the form factors of the model.

Several interesting problems remain to be solved. For instance, we have shown the complete equivalence between the SOV discrete characterization of the spectrum and eigenstates and the reformulation in terms of solutions of the homogeneous $T$-$Q$ equation (i.e. in terms of Bethe-type equations) in the case of an even number of sites only. We plan to consider the case of an odd number of sites in a further study. We also expect to be able to use this reformulation so as to consider the homogeneous and thermodynamic limit of the form factors formulas that we have obtained here, similarly as what has been recently done in the XXX case \cite{KitMNT15}.

Finally, we would like to mention that our results can be used to study, in the SOV framework, the XYZ (or eight-vertex) model with various types of quasi-periodic boundary conditions (related to the values of $\mathsf{x}$ and $\mathsf{y}$). This interesting problem will be the subject of the paper \cite{NicT15b}.

\section*{Acknowledgements}

G.N. and V.T. are supported by CNRS.
We also acknowledge  the support from the ANR grant DIADEMS 10 BLAN 012004.

\appendix

\section{Theta functions, elliptic polynomials and useful identities}
\label{app-theta}

In this paper, $\theta(\lambda)\equiv\theta_1(\lambda|\omega)$ denotes the usual theta-function \cite{GraR07L,WhiW27L} with quasi-periods $\pi$ and $\pi\omega$ ($\Im\omega>0$),
\begin{align}\label{theta1}
  \theta (z)
  &=-i\sum_{k=-\infty}^{\infty} (-1)^k e^{i\pi\omega (k+\frac12)^2} e^{2i (k+\frac12)z},
  \\
  &= 2 e^{i\pi\frac{\omega}{4}}\,\sin z\, \prod_{n=1}^\infty\big(1-e^{2i(n\pi\omega-z)}\big)
       \big(1-e^{2i(n\pi\omega+z)}\big)\big(1-e^{2in\pi\omega}\big),
\end{align}
which satisfies
\begin{equation}\label{periods}
   \theta(z+\pi )=-\theta (z), \qquad
   \theta(z+\pi\omega)= -e^{-i\pi\omega}\, e^{-2 i z}\, \theta(z).
\end{equation}

Throughout the paper, we use the following terminology \cite{FelS99,Ros09,PakRS08}.

Let $\Gamma\equiv \Gamma^{(\pi,\pi\omega)}=\pi\mathbb{Z}+\pi\omega\mathbb{Z}$. Let $\chi:\Gamma\to\mathbb{C}^\times$ be a group homomorphism. We say that a function $f$ of $z\in\mathbb{C}$ is a theta function of quasi-periods $(\pi,\pi\omega)$, of order $n$ and character $\chi$ if $f$ is a holomorphic function satisfying the quasi-periodicity properties
\begin{equation}
  f(z+\pi)=\chi(\pi )\, f(z), \qquad f(z+\pi\omega)= \chi(\pi\omega)\, e^{-in(2z+\pi\omega)}\, f(z).
\end{equation}
We say that $f$ is a theta function (or an elliptic polynomial) of quasi-periods $(\pi,\pi\omega)$, of order $n$ and norm $\alpha_f$ if $f$ is a theta function of order $n$ and character given by $\chi(\pi)=(-1)^n$ and $\chi(\pi\omega)=(-1)^n\,e^{2i\alpha_f}$. This is equivalent to the fact that there exist constants $\lambda_1,\ldots,\lambda_n$ and $C$ with $\lambda_1+\ldots+\lambda_n=\alpha_f$ such that
\begin{equation}\label{theta-n}
f(z)=C\prod_{k=1}^n\theta(z-\lambda_k).
\end{equation}
%

We have the following properties (see for instance \cite{FelS99,PakRS08}):
\begin{enumerate}
  \item Let $\Theta^{(\pi,\pi\omega)}_{n,\alpha}$ be the space of theta functions of quasi-periods $(\pi,\pi\omega)$, of order $n\in\mathbb{N}$ and of norm $\alpha$. Then $\dim \Theta_{n,\alpha}^{(\pi,\pi\omega)}=n$.
  \item Let $f,g\in\Theta_{n,\alpha}^{(\pi,\pi\omega)}$ which coincide at $n$ points $x_1,\ldots,x_n\in\mathbb{C}$: $f(x_i)=g(x_j)$, $1\le j \le n$. If $x_1,\ldots,x_n$ are independent (i.e. if $x_i-x_j\notin\Gamma$ and $\sum_{j=1}^nx_j-\alpha\notin\Gamma$) then $f=g$.
  It means that there exists  a unique theta function (elliptic polynomial) of quasi-periods $(\pi,\pi\omega)$, of order $n$ and of norm $\alpha$ with values $f(x_1),\ldots,f(x_n)$ at the respective independent points $x_1,\ldots,x_n$. It is given by the following interpolation formula:
\begin{equation}\label{interpolation}
   f(\lambda)=\sum_{j=1}^n \frac{\theta(\alpha-\sum_{k=1}^n x_k +x_j-\lambda)}{\theta(\alpha-\sum_{k=1}^n x_k)} \prod_{\substack{k=1 \\ k\not=j}}^n \frac{\theta(\lambda-x_k)}{\theta(x_j-x_k)}\, f(x_j).
\end{equation}  

  \item Let $\{\vartheta_j\}_{1\le j\le n}$ be a basis of $\Theta_{n,\alpha}^{(\pi,\pi\omega)}$. Then, for any $(x_1,\ldots,x_n)\in\mathbb{C}^n$, the determinant of the matrix $(\vartheta_j(x_i))_{1\le i,j,\le n}$ is of the form
  \begin{equation}\label{det-thetaj}
  \det_{1\le i,j\le n}\big[\vartheta_j(x_i)\big]= C\cdot \theta\Big(\sum_{l=1}^n x_l-\alpha\Big)\cdot \prod_{i<j}\theta(x_i-x_j),
  \end{equation}
  where $C$ is some constant.
\end{enumerate}

We also recall Frobenius determinant formula, for any $n$-tuples $(x_1,\ldots,x_n), (y_1,\ldots, y_n)\in\mathbb{C}^n$ (with $x_i - y_j\notin\Gamma$, $\forall i,j$) and any $t\in\mathbb{C}$ (with $t\notin\Gamma$):
\begin{equation}\label{Frob-det}
   \det_{1\le i,j\le n}\left[\frac{\theta(x_i-y_j+t)}{\theta(x_i-y_j)\,\theta(t)}\right]
   =\frac{\theta\big(\sum_{j=1}^n(x_j-y_j)+t\big)}{\theta(t)}
    \frac{\prod_{1\le i<j\le n} \theta(x_i-x_j)\,\theta(y_j-y_i)}
                                              {\prod_{i,j=1}^n \theta(x_i-y_j)} .
\end{equation}
%

\section{Inhomogeneous Baxter equation as reformulation of SOV spectrum}
\label{app-inhom}

In this appendix we explain how one can show the equivalence of the SOV discrete characterization of the spectrum of Theorem~\ref{thm-eigen-t} with the description in terms of elliptic polynomial solutions, with quasi-periods $(\pi,\pi\omega)$, of some particular functional $T$-$Q$ equations with an extra inhomogenous term.

As explained in Section~\ref{sec-hom-inhom}, one can modify the functional equation with respect to \eqref{hom-eq} so as to force this equation to admit  elliptic polynomial solutions with quasi-periods $(\pi,\pi\omega)$.
This means introducing some gauge transformation as in \eqref{def-ad}, so as to adjust the quasi-periodicity properties of the three terms $\bar{\mathsf{t}}(\lambda )\,Q(\lambda )$, $\bar{\mathsc{a}}(\lambda)\, Q(\lambda -\eta )$ and $\bar{\mathsc{d}}(\lambda)\,Q(\lambda +\eta )$ for $\bar{\mathsf{t}}(\lambda )$ satisfying \eqref{periodt-1}-\eqref{periodt-2}.
A possible (and somehow minimal) way to do it is to choose the function $f(\lambda)\equiv f_\mu^{(\beta)}(\lambda)$ as in \eqref{def-f-Q} in terms of two parameters $\beta$ and $\mu$.
Of course, the sum of the aforementioned three terms does not in general cancel, so that one should also add to the equation an inhomogeneous term as in \eqref{inhom}. As explained in Section~\ref{sec-hom-inhom}, this still enables ones to recover the condition \eqref{syst-t} as long as this inhomogeneous term cancels at all points $\xi_n^{(h_n)}$, $n\in\{1,\ldots,\mathsf{N}\}$, i.e. contains the factor $\mathsc{a}(\lambda)\,\mathsc{d}(\lambda)$.


For the class of $Q(\lambda)$ of the form \eqref{def-f-Q} with $\mathsf{M}=\mathsf{N}$, the function $F(\lambda)\equiv F^{(\beta)}_{\mu,Q}(\lambda)$ appearing in the inhomogeneous equation \eqref{inhom} is defined in terms of $\beta$, $\mu$ and $Q(\lambda)$ as
\begin{multline}\label{F-Q}
  F^{(\beta)}_{\mu,Q}(\lambda)
  =
      \frac{\beta^{-1} (-1)^{\mathsf{x}+\mathsf{y}+\mathsf{x}\mathsf{y}}\,e^{-i\mathsf{y}\lambda}\,\theta(t_{0,\mathbf{0}})}{\theta(t_{0,\mathbf{0}}+\alpha_Q-\sum_k\xi_k+\mathsf{N}\eta)}
      \frac{Q(\mu-\eta-t_{0,\mathbf{0}})}{\mathsc{d}(\mu-t_{0,\mathbf{0}})}
      \frac{\theta(\lambda-\mu-\alpha_Q+\sum_k\xi_k-\mathsf{N}\eta)}{\theta(\lambda-\mu+t_{0,\mathbf{0}})}
      \\
      +
      \frac{\beta\,e^{i\mathsf{y}(\lambda+\eta)} \,\theta(t_{0,\mathbf{0}})}{\theta(\mathsf{y}\pi\omega-t_{0,\mathbf{0}}-\alpha_Q+\sum_k\xi_k-\mathsf{N}\eta)}
      \frac{Q(\mu)}{\mathsc{a}(\mu-\eta)}
      \frac{\theta(\lambda-\mu+\eta+\mathsf{y}\pi\omega-t_{\mathbf{0},0}-\alpha_Q+\sum_k\xi_k-\mathsf{N}\eta)}{\theta(\lambda-\mu+\eta)}
\end{multline}
with $\alpha_Q\equiv\sum_{j=1}^\mathsf{N}\lambda_j$ being the norm of the theta function   $Q(\lambda)$ of order $\mathsf{N}$.
Note that \eqref{inhom}-\eqref{F-Q} for $\mathsf{M}=\mathsf{N}$ can be seen as an elliptic generalization of the trigonometric inhomogeneous functional equation that was obtained in \cite{NicT15}. Indeed, under some simple assumptions on the functional dependence of the zeros of the $\bar{\mathsf{t}}(\lambda )\in\Sigma _{\overline{\mathcal{T}}}$, when taking the XXZ limit  $\omega\to +i\infty$ (see Remark~\ref{rem-trig-lim}), one indeed recovers the equation of Theorem~4.1 of  \cite{NicT15}. The elliptic analog of Theorem~4.1 of \cite{NicT15} can then be formulated as follows:

\begin{theorem}\label{th-SOV-Baxter}
Let us suppose that the inhomogeneity parameters $\xi _{1},\ldots,\xi _{\mathsf{N}}$ satisfy \eqref{cond-inh} and let us set $\mathsf{M}=\mathsf{N}$.
Then the following two propositions are equivalent:
\begin{enumerate}
\item\label{cond1} $\bar{\mathsf{t}}(\lambda )$ is an eigenvalue function of the antiperiodic transfer matrix $\overline{\mathcal{T}}(\lambda  )$ (i.e. $\bar{\mathsf{t}}(\lambda )\in\Sigma _{\overline{\mathcal{T}}}$);
\item\label{cond2} $\bar{\mathsf{t}}(\lambda)$ is an entire function of $\lambda$ and, for some $\beta\in\mathbb{C}\setminus\{0\}$, there exists a function $Q(\lambda)$ of the form \eqref{def-f-Q} such that $\big(Q(\xi_j), Q(\xi_j-\eta)\big)\not=(0,0)$, $1\le j\le \mathsf{N}$, and that $\bar{\mathsf{t}}(\lambda )$ and $Q(\lambda)$ satisfy the inhomogeneous functional equation \eqref{inhom}-\eqref{def-f-Q}-\eqref{F-Q}.
%
\newcounter{enumTemp}
    \setcounter{enumTemp}{\theenumi}
\end{enumerate}
If $\eta\in\mathbb{C}\setminus\mathbb{R}$, these propositions are also equivalent to:
\begin{enumerate}
\setcounter{enumi}{\theenumTemp}
\item\label{cond3} $\bar{\mathsf{t}}(\lambda)$ is an entire function of $\lambda$ and, for any $\beta\in\mathbb{C}\setminus\{0\}$, there exists a function $Q(\lambda)$ of the form \eqref{def-f-Q} such that $\big(Q(\xi_j), Q(\xi_j-\eta)\big)\not=(0,0)$, $1\le j\le \mathsf{N}$, and that $\bar{\mathsf{t}}(\lambda )$ and $Q(\lambda)$ satisfy the inhomogeneous functional equation \eqref{inhom}-\eqref{def-f-Q}-\eqref{F-Q}.
\end{enumerate}
%
\end{theorem}

{\it Proof of Theorem~\ref{th-SOV-Baxter}.}
Obviously {\it \ref{cond3}}. implies {\it \ref{cond2}}.

So as to prove that {\it \ref{cond2}.} implies {\it \ref{cond1}.}, let us suppose that, for some $\beta\in\mathbb{C}\setminus\{0\}$, there exists an elliptic polynomial $Q(\lambda)$ of order $\mathsf{N}$  such that 
the function $\bar{\mathsf{t}}(\lambda)$ defined as
 \begin{equation}\label{def-tQ}
   \bar{\mathsf{t}}(\lambda)\equiv
    \frac{f^{(\beta)}_\mu(\lambda)\,\mathsc{a}_{\mathsf{x},\mathsf{y}}(\lambda)\,Q(\lambda -\eta )
+\big[ f^{(\beta)}_\mu(\lambda+\eta)\big]^{-1}\, \mathsc{d}(\lambda)\,  Q(\lambda +\eta )
-\mathsc{a}(\lambda)\,\mathsc{d}(\lambda)\, F^{(\beta)}_{\mu,Q}(\lambda)}{Q(\lambda)}
 \end{equation} 
 is an entire function of $\lambda$.
 Then the function $e^{i\mathsf{y}\lambda}\,\bar{\mathsf{t}}(\lambda)$ is a theta function of order $\mathsf{N}$ and of norm $\alpha_{\bar{\mathsf{t}}}\equiv \sum_{k=1}^\mathsf{N}\xi_k+t_{0,\mathbf{0}}$. Moreover, the particularization of \eqref{def-tQ} at the $2\mathsf{N}$ points $\xi_j$ and $\xi_j-\eta$, $1\le j\le \mathsf{N}$, gives
 \begin{equation}
    Q(\xi_j)\, \bar{\mathsf{t}}(\xi_j)-f^{(\beta)}_\mu(\xi_j)\, \mathsc{a}_{\mathsf{x},\mathsf{y}}(\xi_j)\,Q(\xi_j -\eta )=0,
    \qquad
    Q(\xi_j -\eta )\,\bar{\mathsf{t}}(\xi_j-\eta)-\frac{\mathsc{d}(\xi_j-\eta)}{ f^{(\beta)}_\mu(\xi_j)}\,  Q(\xi_j )=0,
 \end{equation}
 for each $j\in\{1,\ldots,\mathsf{N}\}$ which, provided that $\big(Q(\xi_j), Q(\xi_j-\eta)\big)\not=(0,0)$,
 means that the matrix \eqref{mat-Dn} has zero determinant and therefore that $\bar{\mathsf{t}}(\lambda)$ satisfies \eqref{syst-t}. Hence $\bar{\mathsf{t}}(\lambda )\in\Sigma _{\overline{\mathcal{T}}}$.
 
Let us now prove that {\it \ref{cond1}}. implies  {\it \ref{cond2}}. and, in the case $\eta\notin\mathbb{R}$, {\it \ref{cond3}}. Let $\beta\in\mathbb{C}\setminus\{0\}$ and let $\bar{\mathsf{t}}(\lambda )\in\Sigma _{\overline{\mathcal{T}}}$. For any elliptic polynomial $Q(\lambda)$ of degree $\mathsf{N}$ and  $F^{(\beta)}_{\mu,Q}(\lambda)$ defined in terms of $Q(\lambda)$ by \eqref{F-Q}, the function
 \begin{equation}
   e^{i\mathsf{y}}\, \left\{ \bar{\mathsf{t}}(\lambda )\,Q(\lambda )
   -f^{(\beta)}_\mu(\lambda)\,\mathsc{a}_{\mathsf{x},\mathsf{y}}(\lambda)\,Q(\lambda -\eta )
-\frac{\mathsc{d}(\lambda)}{f^{(\beta)}_\mu(\lambda+\eta)}\,  Q(\lambda +\eta )
+\mathsc{a}(\lambda)\,\mathsc{d}(\lambda)\, F^{(\beta)}_{\mu,Q}(\lambda) \right\}
\end{equation}
is a theta function of order $2\mathsf{N}$ and of norm $\sum_{k=1}^\mathsf{N}\xi_k+t_{0,\mathbf{0}}+\alpha_Q$, where $\alpha_Q=\sum_{j=1}^\mathsf{N}\lambda_j$ is the sum of the roots of the elliptic polynomial $Q(\lambda)$.
Then the equation \eqref{inhom} is satisfied for $\bar{\mathsf{t}}$ and $Q$  if (and only if) it is satisfied in $\mathsf{N}$ independent points, namely
\begin{itemize}
\item for $\lambda=\xi_j$, $j=1,\ldots\mathsf{N}$:
\ \
$ \bar{\mathsf{t}}(\xi_j)\, \theta(\xi_j-\mu+t_{0,\mathbf{0}})\, Q(\xi_j)
   = \beta^{-1} e^{-i\mathsf{y}\xi_j} \mathsc{a}_{\mathsf{x},\mathsf{y}}(\xi_j)\, \theta(\xi_j-\mu)\, Q(\xi_j-\eta)$,
\item for $\lambda=\xi_j-\eta$, $j=1,\ldots\mathsf{N}$:
\ \ 
$ \bar{\mathsf{t}}(\xi_j-\eta)\, \theta(\xi_j-\mu)\, Q(\xi_j-\eta)
   = \beta\, e^{i\mathsf{y}\xi_j} \mathsc{d}(\xi_j-\eta)\, \theta(\xi_j-\mu+t_{0,\mathbf{0}})\, Q(\xi_j)$,
\end{itemize}
provided that
\begin{equation}\label{indep1}
  \sum_{k=1}^\mathsf{N}\xi_k+t_{0,\mathbf{0}}-\alpha_Q \notin \Gamma.
\end{equation}
Since $\bar{\mathsf{t}}(\lambda)\in\Sigma_{\overline{\mathcal{T}}}$ satisfies \eqref{syst-t}, the above system is therefore equivalent to the following system of $\mathsf{N}$ equations:
\begin{equation}\label{sys-N}
   \bar{\mathsf{t}}(\xi_j)\, \theta(\xi_j-\mu+t_{0,\mathbf{0}})\, Q(\xi_j)
   = \beta^{-1} e^{-i\mathsf{y}\xi_j} \mathsc{a}_{\mathsf{x},\mathsf{y}}(\xi_j)\, \theta(\xi_j-\mu)\, Q(\xi_j-\eta),
   \qquad j=1,\ldots,\mathsf{N}.
\end{equation}
In general, an elliptic polynomial of order $\mathsf{N}$ and of norm $\alpha_Q$ is completely characterized by its values at $\mathsf{N}$ independent points. Hence it can be written in the following form:
\begin{equation}
   Q(\lambda)= \sum_{k=1}^\mathsf{N} \frac{\theta(\lambda-\xi_k+\sum_\ell \xi_\ell-\alpha_Q)}{\theta(\sum_\ell \xi_\ell-\alpha_Q)}\,
   \prod_{\substack{\ell=1\\ \ell\not= k}}^\mathsf{N} \frac{\theta(\lambda-\xi_\ell)}{\theta(\xi_k-\xi_\ell)}\,
   Q(\xi_k),
\end{equation}
provided $ \sum_\ell \xi_\ell-\alpha_Q\notin \Gamma$.
Hence, the system \eqref{sys-N} is in fact a system of $\mathsf{N}$ homogeneous linear equations in the $\mathsf{N}$ unknowns $Q(\xi_n)$, $n\in\{1,\ldots,\mathsf{N}\}$, which can be written as:
\begin{equation}\label{system-N}
   \sum_{k=1}^{\mathsf{N}} \big[ C_{\bar{\mathsf{t}}}(\beta,\alpha_Q) \big]_{jk}\, Q(\xi_k) =0,
   \qquad j=1,\ldots,\mathsf{N},
\end{equation} 
where $C_{\bar{\mathsf{t}}}(\beta,\alpha_Q)$ is the $\mathsf{N}\times\mathsf{N}$ matrix of elements
\begin{equation}\label{C-N}
  \big[ C_{\bar{\mathsf{t}}}(\beta,\alpha_Q) \big]_{ab}
  = \delta_{ab}\, \frac{ \beta\, e^{i\mathsf{y}\xi_a}\,\bar{\mathsf{t}}(\xi_a)}{\mathsc{a}_{\mathsf{x},\mathsf{y}}(\xi_a)}
     \frac{\theta(\xi_a-\mu+t_{0,\mathbf{0}})}{\theta(\xi_a-\mu)}
     -\frac{\theta(\xi_a-\xi_b-\eta+\sum_\ell \xi_\ell-\alpha_Q)}{\theta(\sum_\ell \xi_\ell-\alpha_Q)}
     \prod_{\substack{\ell=1\\ \ell\not= b}}^\mathsf{N} \frac{\theta(\xi_a-\xi_\ell-\eta)}{\theta(\xi_b-\xi_\ell)}.
\end{equation}
Note that, using Frobenius determinant formula \eqref{Frob-det} and the formula for the determinant of the sum of two matrices, one can express the determinant of the matrix \eqref{C-N} in the form
\begin{multline}\label{det-C-N}
 \det_\mathsf{N} \big[ C_{\bar{\mathsf{t}}}(\beta,\alpha_Q) \big]
  =  \sum_{n=0}^\mathsf{N} (-1)^n\beta^{\mathsf{N}-n}\frac{\theta(\sum_\ell \xi_\ell-\alpha_Q-n\eta)}{\theta(\sum_\ell \xi_\ell-\alpha_Q)} \\
  \times
     \sum_{\substack{P\subset\{1,\ldots,\mathsf{N}\}\\ \#P=n}}
     \prod_{a\notin P}\left\{  e^{i\mathsf{y}\xi_a}\, \frac{\bar{\mathsf{t}}(\xi_a)}{\mathsc{a}_{\mathsf{x},\mathsf{y}}(\xi_a)}
     \frac{\theta(\xi_a-\mu+t_{0,\mathbf{0}})}{\theta(\xi_a-\mu)} \prod_{b\in P}\frac{\theta(\xi_a-\xi_b+\eta)}{\theta(\xi_a-\xi_b)} \right\}.
\end{multline}
where the second summation in \eqref{det-C-N} runs over all subsets $P$ of the set $\{1,\ldots, \mathsf{N}\}$ with cardinality $n$.

The system \eqref{system-N} admits some non-zero solution $\big(Q(\xi_1),\ldots,Q(\xi_\mathsf{N})\big)$ if and only if the determinant of the matrix \eqref{C-N} is zero.
At $\beta=0$, this determinant simplifies into
\begin{equation}
   \det_\mathsf{N} \big[ C_{\bar{\mathsf{t}}}(0,\alpha_Q) \big]
   =(-1)^\mathsf{N}\,\frac{\theta(\alpha_Q-\sum_\ell\xi_\ell+\mathsf{N}\eta)}{\theta(\alpha_Q-\sum_\ell\xi_\ell)},
\end{equation}
so that 
\begin{align}
   &\det_\mathsf{N} \big[ C_{\bar{\mathsf{t}}}(0,\alpha_Q) \big]=0\ \Leftrightarrow\ \exists (k_1,k_2)\in\mathbb{Z}^2,\ \alpha_Q=\sum_\ell\xi_\ell-\mathsf{N}\eta+\pi k_1+\pi\omega k_2,
    \\
   &\frac{\partial  \det_\mathsf{N} \big[ C_{\bar{\mathsf{t}}}(\beta,\alpha_Q) \big]}{\partial \alpha_Q}\bigg|_{\substack{\alpha_Q=\sum_\ell\xi_\ell-\mathsf{N}\eta+\pi k_1+\pi\omega k_2\\ \beta=0}}\not= 0.
\end{align}
Hence we can apply the implicit function theorem for holomorphic functions:
$\forall (k_1,k_2)\in\mathbb{Z}^2$, there exist some open vicinities $U_{(k_1,k_2)}$ and $V_{(k_1,k_2)}$ of $0$ and of $\alpha^{(0)}_{(k_1,k_2)}\equiv \sum_\ell\xi_\ell-\mathsf{N}\eta+\pi k_1+\pi\omega k_2$ respectively, and there exists a unique holomorphic function $\alpha_{(k_1,k_2)}:U_{(k_1,k_2)}\to V_{(k_1,k_2)}$ with $\alpha_{(k_1,k_2)}(0)=\alpha^{(0)}_{(k_1,k_2)}$ such that 
\begin{equation}
   \Big\{ (\beta,\alpha_Q)\in U_{(k_1,k_2)}\times V_{(k_1,k_2)} \mid \det_\mathsf{N} \big[ C_{\bar{\mathsf{t}}}(\beta,\alpha_Q) \big]=0\Big\}
   = \Big\{ (\beta, \alpha_{(k_1,k_2)}(\beta) ) \mid \beta\in U_{(k_1,k_2)} \Big\}.
\end{equation}
Hence the system 
\begin{equation}\label{sys-U}
   \sum_{k=1}^{\mathsf{N}} \big[ C_{\bar{\mathsf{t}}}\big(\beta,\alpha_{(k_1,k_2)}(\beta)\big) \big]_{jk}\, q_k =0,
   \qquad j=1,\ldots,\mathsf{N},
\end{equation} 
admits, for all $\beta\in U_{(k_1,k_2)}$, a non-zero solution $(q_1,\ldots,q_\mathsf{N})$. Due to the form of \eqref{sys-U}, one can choose this solution such that all $q_j\equiv q_j(\beta)$ are continuous function of $\beta$ in $U_{(k_1,k_2)}$.

Note that, at $\beta=0$, we have $q_j(0)\not=0$ for all $j\in\{1,\ldots,\mathsf{N}\}$ (it is clear from the system \eqref{sys-N} that all the roots of the solution $Q(\lambda)$ are, up to $\Gamma$-periodicity, at the points $\xi_j-\eta$, $1\le j \le \mathsf{N}$, when $\beta=0$). Hence it is always possible to choose $U_{(k_1,k_2)}$ such that $q_j(\beta)\not=0$ for all  $j\in\{1,\ldots,\mathsf{N}\}$ and for all  $\beta$ in $U_{(k_1,k_2)}$.
Moreover, since $\alpha_{(k_1,k_2)}(0)-\sum_\ell\xi_\ell\notin\Gamma$ and $\alpha_{(k_1,k_2)}(0)-\sum_\ell\xi_\ell-t_{0,\mathbf{0}}\notin\Gamma$, and since the function $\alpha_{(k_1,k_2)}(\beta)$ is holomorphic, one can also choose $U_{(k_1,k_2)}$ such that $\alpha_{(k_1,k_2)}(\beta)-\sum_\ell\xi_\ell\notin\Gamma$ and that $\alpha_{(k_1,k_2)}(\beta)-\sum_\ell\xi_\ell-t_{0,\mathbf{0}}\notin\Gamma$ for any $\beta\in U_{(k_1,k_2)}$.
Hence we have shown  {\it \ref{cond2}}.

Let us now notice that, up to a similarity transformation, the matrix $C_{\bar{\mathsf{t}}}(\beta,\alpha_Q+\pi\omega)$ is proportional to the matrix $C_{\bar{\mathsf{t}}}(\beta e^{2i\eta},\alpha_Q)$. Hence the system \eqref{system-N} admits a non-zero solution for $(\beta,\alpha_Q)$ if and only if it is the case for the system corresponding to $(\beta e^{2i n \eta},\alpha_Q+n\pi\omega)$ ($\forall n\in\mathbb{Z}$).
 
Let  $\beta\in\mathbb{C}\setminus\{0\}$, and let us suppose moreover that  $\eta\notin\mathbb{R}$. Hence, there exists $n\in\mathbb{Z}$ such that $\beta\, e^{2ik\eta}\in U_{(0,0)}$, so that the system associated with $\big(\beta e^{2i n \eta},\alpha_{(0,0)}(\beta e^{2i n \eta})\big)$ admits a non-zero solution such that each unknown (i.e. each $Q(\xi_j)$ solution to this system) is itself non-zero.
It follows from the previous remark that the system for $\big(\beta,\alpha_{(0,0)}(\beta e^{2i n \eta})-n\pi\omega\big)$ admits also a non-zero solution which is such that $Q(\xi_j)\not=0$ for each $j\in\{1,\ldots,\mathsf{N}\}$.
Hence we have shown   {\it \ref{cond3}}.
\qed

The inhomogeneous functional equation \eqref{inhom}-\eqref{def-f-Q}-\eqref{F-Q} for $\mathsf{M}=\mathsf{N}$ is not the only functional equation that can be considered in this framework. For instance, it is also possible to completely characterize the SOV spectrum in terms of a functional equation of the same type (still using \eqref{def-f-Q} to define the gauge function $f(\lambda)$ in terms of some arbitrary parameter $\mu$) but for an elliptic polynomial $Q(\lambda)$ of degree $\mathsf{M}=\mathsf{N}+1$ and arbitrary norm $\alpha_Q$. Of course the inhomogeneous term has to be adapted accordingly and it appears slightly more complicated in this case. 
In that way, increasing the degree of the elliptic polynomial $Q(\lambda)$ corresponds to increasing the number of free parameters in the inhomogeneous equation: apart from the gauge parameter $\beta$, we have one free parameter (the parameter $\mu$) for the degree $\mathsf{N}$,
and two free parameters ($\mu$ and $\alpha_Q$) for the degree $\mathsf{N}+1$.
Hence, the choice of the degree $\mathsf{N}-1$ for $Q(\lambda)$ seems to be the minimal possible if one considers equations of the form \eqref{inhom}-\eqref{def-f-Q}: in that case we have still  $\mathsf{N}$ unknown parameters which are the parameter $\mu$ appearing in the definition \eqref{def-f-Q} of $f(\lambda)$, as well as the $\mathsf{N}-1$ roots of $Q(\lambda)$, and there does not remain any free parameter (except $\beta$), cf. footnote~\ref{foot-N-1}.

It is interesting to remark that, if we have a complete description of the transfer matrix spectrum in terms of the elliptic polynomial solutions $Q(\lambda)$ of degree $\mathsf{M}$ of some inhomogeneous functional equation of the form \eqref{inhom} with associated function $f(\lambda)$, it is possible to rewrite the transfer matrix eigenvectors in a generalized Bethe form, in terms of the roots of the corresponding elliptic polynomial $Q(\lambda)$. More precisely, defining the states,
\begin{align}
& |\Omega _{\mathsf{M},f}^{(\kappa)}\rangle 
   =\sum_{\mathbf{h}\in \{0,1\}^{\mathsf{N}}}
     \prod_{a=1}^{\mathsf{N}}\left( 
     \frac{e^{i\mathsf{y}\eta}\, \mathsc{a}_{\mathsf{x,y}}(\xi _{a})}{\kappa \,\mathsc{d}(\xi _{a}-\eta )}\,
     f(\xi_a)\right) ^{\!h_{a}}
     \det_{\mathsf{N}}\big[\Theta ^{(0,\mathbf{h})}\big]\,
     |\mathbf{h},-\mathsf{M}\rangle ,  \label{refT-r-Q} \\
& \langle  \Omega_{\mathsf{M},f}^{(\kappa)}|
   =\sum_{\mathbf{h}\in \{0,1\}^{\mathsf{N}}}\prod_{a=1}^{\mathsf{N}}\Big( 
     \kappa\, e^{i\mathsf{y}\eta }\, f(\xi_a)\Big) ^{h_{a}}
     \det_{\mathsf{N}}\big[\Theta ^{(0,\mathbf{h})}\big]\,
     \langle \mathsf{M},\mathbf{h}|,  \label{refT-l-Q}
\end{align}
we have the following result:

\begin{corollary}\label{Eigen-SOV-Bethe}
Let $\bar{\mathsf{t}}(\lambda )$ be an eigenvalue of the $\kappa $-twisted
antiperiodic transfer matrix (i.e. $\bar{\mathsf{t}}(\lambda )\in \Sigma _{\overline{\mathcal{T}}}$), and let
\begin{equation}
    Q(\lambda)=\prod_{j=1}^\mathsf{M}\theta(\lambda-\lambda_j)
\end{equation}
be such that $\bar{\mathsf{t}}(\lambda )$ and $Q(\lambda)$ satisfy the functional equation \eqref{inhom} for some function $f(\lambda)$.
Then the $\overline{\mathcal{T}}(\lambda)$- left and right eigenstates with eigenvalue $\bar{\mathsf{t}}(\lambda )$ can be represented as
\begin{equation}\label{eigen-Bethe}
|\Psi_{\bar{\mathsf{t}}}^{( \kappa ) }\rangle
    =\prod_{a=1}^{\mathsf{M}} \Big[ e^{i\mathsf{y}\tau }\, \theta(\tau ) ^{-1}\, \mathcal{D}(\lambda _{a})\big] 
       |\Omega _{\mathsf{M},f}^{(\kappa)}\rangle ,
       \qquad
\langle \Psi_{\bar{\mathsf{t}}}^{( \kappa ) }|
    =\langle \Omega _{\mathsf{M},f}^{(\kappa)}|
      \prod_{a=1}^{\mathsf{M}}
      \big[ \mathcal{D}(\lambda _{a})\, \theta(\tau )^{-1}\,e^{i\mathsf{y}\tau }\big] ,
\end{equation}
where the order of the operators in each bracket $[ \ldots ] $ has to
be kept as it appears.
\end{corollary}

In other words, if $f(\lambda)$ is fixed, which is the case for the complete characterization of the transfer spectrum that we have obtained in Theorem~\ref{th-SOV-Baxter} in terms of elliptic polynomials solutions of \eqref{inhom}-\eqref{def-f-Q} for $\mathsf{M}=\mathsf{N}$, then \eqref{refT-r-Q} and \eqref{refT-l-Q} have to be understood as some fixed pseudo-vacuum states. The corresponding eigenstates are then obtained by multiple action, on these pseudo-vacuum states, of the (slightly dressed) operator $\mathcal{D}(\lambda)$ evaluated at the roots of the Bethe equations, i.e. in a form very similar to what happens in the context of  ABA.

Let us finally insist on the fact  that the possibility to write such an ABA-type  representation for the SOV transfer matrix eigenstates is very general for models solved by SOV, as soon as we have some characterization of the transfer matrix spectrum in terms of (generalized) polynomial solutions to some (homogeneous or inhomogeneous) $T$-$Q$ functional equation. It is not restricted to reformulations of the SOV spectrum by inhomogeneous Baxter equations; indeed, we can construct similar reformulations, using multiple action of some slightly different operator, in the framework of the characterization in terms of the solutions of the homogeneous $T$-$Q$ equation as obtained in Section~\ref{sec-diag-t}. It is not restricted to the model under consideration; indeed, a similar rewriting is possible for all the integrable models so far solved in the SOV framework \cite{NicT10,Nic12,GroN12,FalN14,FalKN14,Nic13,Nic13b,NicT15} as soon as we have a characterization of the SOV spectrum by a finite system of (generalized) Bethe equations.  



\providecommand{\bysame}{\leavevmode\hbox to3em{\hrulefill}\thinspace}
\providecommand{\MR}{\relax\ifhmode\unskip\space\fi MR }
\providecommand{\MRhref}[2]{%
  \href{http://www.ams.org/mathscinet-getitem?mr=#1}{#2}
}
\providecommand{\href}[2]{#2}

\end{document}